\documentclass[a4paper,10pt]{article}
\usepackage{float}

 \usepackage{latexstyle}
 \usepackage{bbm}
 
\usepackage{mathtools}
 
\usepackage{amsmath}
\makeatletter
\renewcommand*\env@matrix[1][*\c@MaxMatrixCols c]{%
  \hskip -\arraycolsep
  \let\@ifnextchar\new@ifnextchar
  \array{#1}}
\makeatother

\title{Structure dependent sampling in compressed sensing:\\
theoretical guarantees for tight frames}
\author{Clarice Poon\thanks{cmhsp2@cam.ac.uk}\\ Department of Applied Mathematics and Theoretical Physics\\ University of Cambridge}
\date{May 2015; Revised September 2015}

\begin{document}
\maketitle

\begin{abstract}

Many of the applications of compressed sensing have been based on variable density sampling, where certain sections of the sampling coefficients are sampled more densely. Furthermore, it has been observed that these sampling schemes are dependent not only on sparsity but also on the sparsity structure of the underlying signal.  
This paper extends the result of  (Adcock, Hansen, Poon and Roman, arXiv:1302.0561, 2013) to the case where the sparsifying system forms a tight frame. By dividing the sampling coefficients into levels, our main result will describe how the amount of subsampling in each level is determined by the \textit{local coherences} between the sampling and sparsifying operators and the \textit{localized level sparsities} -- the sparsity in each level under the sparsifying operator. 

\end{abstract}

\section{Introduction}
Over the past decades, much of the research in signal processing has been based on the assumption that natural signals can be sparsely represented. One of the achievements resulting from this realization was compressed sensing, which made it possible to recover a sparse signal from very few non-adaptive linear measurements. Compressed sensing is typically modelled as follows. Given an unknown vector $x \in \mathbb{C}^N$ and a measurement device represented by a matrix $V$, one aims to recover $x$ from a highly incomplete set of measurements by solving
\be{\label{eq:min_findim}
R(x, \Omega) \in\mathop{\mathrm{arg min}}_{z\in\bbC^N} \nm{D z}_{\ell^1} \text{ subject to }   P_\Omega V z = P_\Omega V x,
}
where $\Omega$ indexes the given measurements, $P_\Omega$ is a projection matrix which restricts a vector to its coefficients indexed by $\Omega$ and $D$ is a sparsifying matrix under which $Dx$ is assumed to be sparse. Typical results in compressed sensing describe how under certain conditions, one can guarantee recovery when the number of measurements  $\abs{\Omega}$ scales up  to a log factor linearly with sparsity \cite{candes2006stable,candes2006robust,candes2011probabilistic}.

A large part of the theoretical development of compressed sensing has revolved around the construction of random sampling matrices (such as matrices constructed from random Gaussian ensembles) where the choice of the samples is completely independent of the sparsifying system \cite{donohoCS,rauhut2008compressed,rudelson2008sparse,wojtaszczyk2010stability}. The use of overcomplete dictionaries in compressed sensing has also been studied in works such as \cite{candes2011compressed,giryes2014greedy,krahmer2011new}, but again, recovery guarantees were obtained only for randomised sampling matrices or subsampled structured matrices with randomised column signs. However, in the majority of applications where compressed sensing has been of interest, one is concerned with the recovery of a signal from structured measurements, without the possibility of first randomising the underlying signal. For example, the measurements in magnetic resonance imaging (MRI) are modelled via the Fourier transform, while the measurements in radio interferometry are modelled via the Radon transform. In these cases, how one can achieve subsampling is \textit{highly dependent} on the sparsifying transform. To explain this statement, we recall some results of compressed sensing on the recovery of a vector of length $N$ from its discrete Fourier coefficients under various sparsifying transforms.
 \begin{enumerate}
 \item[(1)] If the underlying vector is $s$-sparse in its canonical basis, then one can guarantee perfect recovery from $\ord{s\log N}$ Fourier coefficients drawn uniformly at random \cite{candes2006robust}.
 \item[(2)] If the underlying vector is $s$-sparse with respect to its total variation \cite{candes2006robust}, then $\ord{s\log N}$ Fourier coefficients drawn uniformly at random will again guarantee perfect recovery, however, in the presence of noise and approximate sparsity, then one can obtain superior error bounds with sampling strategies which sample more densely at low frequency coefficients instead \cite{tv_poon}.
 \item[(3)] If the underlying vector is $s$-sparse with respect to some wavelet basis, then it is impossible to guarantee recovery from $\ord{s\log N}$ samples from sampling uniformly at random. This is a phenomenon which has been observed since the early days of compressed sensing  and there has been extensive investigations into how subsampling is still achievable by sampling more densely at low frequencies \cite{Lustig,Lustig2,studer2012compressive,wiaux2009compressed,puy2011variable}. These approaches were often referred to as variable density sampling and theoretical guarantees for these approaches were recently derived in \cite{ward2013stable} and \cite{adcockbreaking}.
 \end{enumerate}
 More generally, whether one can sample uniformly at random depends on whether the sampling and sparsifying matrices are sufficiently incoherent. In the absence of incoherence (as is the case in (3) above), how one should choose $\Omega$ in (\ref{eq:min_findim}) becomes a far more delicate issue. To explain the use of compressed sensing in this case, a theoretical framework was developed in \cite{adcockbreaking}
 on the basis of three new principles: multilevel sampling, asymptotic incoherence and asymptotic sparsity. By modelling a nonuniform sampling strategies via multilevel sampling, the need for dense sampling at low frequencies in (3) is due to the following two reasons.
 \begin{enumerate}
 \item[(i)] The high correspondence between Fourier  and wavelet bases at low Fourier  frequencies and low wavelet scales, but the low correspondence at high Fourier frequencies and high wavelet scales (asymptotic incoherence).
 \item[(ii)] Typical signals or images exhibit distinctive sparsity patterns in their wavelet coefficients, and  become increasingly sparse at higher wavelet scales (asymptotic sparsity). 
 \end{enumerate}
 In contrast to the large body of results in compressed sensing where the strategy is based on sparsity alone, the results of \cite{adcockbreaking} demonstrated that one of the driving forces behind the success of variable density sampling strategies is their correspondence to the sparsity structure of the underlying signals of interest. These new principles provide a framework under which one can understand how to exploit both the sparsity structure of the underlying signal, and the correspondences between the sampling and sparsifying systems to devise optimal subsampling strategies \cite{Siemens, Roman}.

\subsection{Contribution and overview}
The paper \cite{adcockbreaking} is concerned only with the case where the sparsifying system is an orthonormal basis. On the other hand, many of the sparsifying transforms in applications tend to be constructed from overcomplete dictionaries, such as contourlets \cite{do2005contourlet}, curvelets \cite{Cand,Candes2002}, shearlets  \cite{Gitta,Gitta3} and wavelet frames \cite{daubechies2003framelets,dong2010mra}.  
 
With this in mind, the recent work of \cite{KrahmerNW15} derives theoretical guarantees for certain nonuniform sampling strategies in the case of sparsity with respect to a tight frame. By defining the localization factor $\eta_{s,D}$ with respect to a sparsifying transform $D\in\bbC^{N\times n}$ and a sparsity level $s$ as
\be{\label{eq:KNW_factor}
\eta_{s,D}= \eta = \sup\br{ \frac{\nm{D g}_{\ell^1}}{\sqrt{s}}:  \abs{\Delta} = s, 
g\in\cR(D^*P_\Delta),  \nm{ g}_{\ell^2} = 1},
}
 their result is as follows.

\begin{theorem}[\cite{KrahmerNW15}]\label{thm:knw}
Let $N\in\bbN$ and let $s<N$. Suppose that the rows  $\br{d_1,\ldots, d_n}$ of  $D\in\bbC^{N\times n}$ form a Parceval frame, the rows    $\br{v_1,\ldots, v_n}$ of $V\in\bbC^{n\times n}$ form an orthonormal basis of $\bbC^n$  and suppose that $\sup_{1\leq j\leq N}\abs{\ip{d_j}{v_k}} \leq \mu_k$. Let $\nu$ be a probability measure on $\br{1,\ldots, n}$ given by $\nu(k) = \mu_k^2/\norm{\boldsymbol{\mu}}_{\ell^2}^2$, where $\boldsymbol{\mu} = (\mu_k)_{k=1}^n$, and let $W\in\bbC^{n\times n}$ be a diagonal matrix with diagonal entries $(\norm{\boldsymbol{\mu}}_{\ell^2}^2/\mu_k^2)_{k=1}^n$. Let $\Omega$ be a set of $m$  independently and identically distributed indices drawn from $\br{1,\ldots, n}$ with the measure $\nu$. If
$$
m\geq C \eta^2 \norm{\boldsymbol{\mu}}_2^2 s \max\br{ \log^3(s\eta^2) \log(N), \log(\epsilon^{-1})},
$$ for some absolute constant $C$,
then with probability $1-\epsilon$, the following holds for every $f\in\bbC^n$: the solution $\tilde f$ of
\be{\label{eq:knw}
\mathrm{argmin}_{g \in\bbC^n} \norm{D g}_{\ell^1} \text{ subject to } \norm{\frac{1}{\sqrt{m}} W(P_\Omega V g -y)}_{\ell^2}\leq \delta
}
with $y= P_\Omega V f + e $ for noise $e$ with weighted error $\norm{\frac{1}{\sqrt{m}} W e}_{\ell^2} \leq \delta$ satisfies
$$
\norm{\tilde f - f}_{\ell^2} \leq C_1 \delta + C_2 \sigma_s(Df) s^{-1/2}
$$
where $C_1$ and $C_2$ are absolute constants and given any vector $x\in\bbC^n$, $\sigma_s(x) = \min_{z \in\bbC^n, \norm{z}_{\ell^0} = s} \norm{x-z}_{\ell^1}$.
\end{theorem}

Although this theorem
guarantees the recovery of all sparse vectors under a (fixed) nonuniform sampling distribution, it does not reveal any dependence between the sampling strategy and any sparsity structure. In the case of subsampling the Fourier transform,  this result implies that the sampling cardinality is $m=\ord{s\log^3(s)\log^2(n)}$ when $D$ is an orthonormal Haar wavelet basis, and $m=\ord{s\log^3(s\log(n))\log^3(n)}$ when $D$ is a redundant Haar frame. Due to the relatively large number of $\log$ factors, these sampling bounds are still substantially more pessimistic than what is often observed empirically, and one possible reason for this could be the lack of structure dependence considered in the theorem: in  \S \ref{sec:discreteHaar}, we will present a numerical example  to explain why an understanding of this dependence is crucial to achieving subsampling. 

Therefore, the purpose of this paper is to develop a theory on how to structure one's samples based on the sparsity structure with respect to a tight frame. The minimization problem tackled in this paper is also slightly different from (\ref{eq:knw}) as we consider solutions of the more standard problem (\ref{eq:min_orth}) with a uniform noise assumption, without additional weighting factors.  We remark also that if there exists a strong dependence between the sampling strategy and the underlying sparsity structure, then a direct implication is that there does not exist a fixed optimal sampling distribution for all sparse signals, and this will be reflected in our main result as we account for recovery under various sampling distributions using the framework of multilevel sampling.

The outline of this paper is as follows.
\S \ref{sec:onb} recalls the key principles from \cite{adcockbreaking} and a result on solutions of (\ref{eq:min_findim}) in the case where $\rD$ is constructed from an orthonormal basis. The main result of this paper is presented in \S \ref{sec:main}, where we reveal how the main result of \cite{adcockbreaking} can be extended in the case where $\rD$ is constructed from a tight frame. The remainder of this paper will be devoted to proving the result of \S \ref{sec:main}.
 
\paragraph{Notation}
Given Banach spaces $X$ and $Y$, let $\cB(X,Y)$ denote the space of bounded linear operators from $X$ to $Y$ and let $\cB(X)$ denote the space of bounded linear operators from $X$ to $X$. Let $\cH$ be a Hilbert space and given any subspace $\cS \subseteq \cH$, $ Q_\cS$ denotes the orthogonal projection onto $\cS$.
We say that $\br{\varphi_j : j\in\bbN}$ is a frame for $\cH$ if there exists $c,C>0$ such that
$$
c\nm{g}_\cH^2 \leq \sum_{j\in\bbN}\abs{\ip{g}{\varphi_j}}^2 \leq C \nm{g}_\cH^2, \qquad \forall g\in\cH.
$$
We say that $\br{\varphi_j: j\in\bbN}$ is a tight frame if $c=C$. If $c=C=1$, then $\br{\varphi_j: j\in\bbN}$ is said to be a Parseval frame. Given any linear operator $U$, let $\cR(U)$ denote its range and let $\cN(U)$ denote its null space.

We will also consider the sequence spaces $\ell^p(\bbN)$ for $p\in[1,\infty]$. Let $\br{e_j:j\in\bbN}$ denote the canonical basis for the $\ell^p(\bbN)$ space under consideration. Given any $\Delta \subset \bbN$, $P_\Delta$ denotes the orthogonal projection onto $\overline{ \mathrm{span}\br{e_j: j\in\Delta}}$. Given $M\in\bbN$, let $[M]:=\br{1,\ldots, M}$. Given $z\in\ell^2(\bbN)$, let $\sgn(z) \in\ell^\infty(\bbN)$ be such that for each $j\in\bbN$, $$\sgn(z)_j = \begin{cases} z_j/\abs{z_j}& z_j\neq 0\\
0& \text{otherwise} \end{cases}.$$
Given $q\in (0,\infty]$, the $\ell^q$ norm (or quasi-norm if $q\in (0,1)$) is defined for $z=(z_j)_{j\in\bbN}$ as
$$
\nm{z}_{\ell^q}^q = \sum_{j}\abs{z_j}^q, \quad q\in (0,\infty), \qquad \nm{z}_{\ell^\infty } = \sup_j \abs{z_j}, 
$$  
Let $\nmu{\cdot}_{\ell^p \to \ell^q}$ denote the operator norm of  $\cB(\ell^p(\bbN),\ell^q(\bbN))$ for $p,q\in [1,\infty]$. If $X$ and $Y$ are Hilbert spaces, we will simply denote the operator norm of $\cB(X,Y)$ by $\nmu{\cdot}$.
Given $a,b\in\bbR$, $a\lesssim b$ denotes $a\leq C\, b$ where $C$ is a constant which is independent of all variables under consideration. The identity operator is denoted by $I$, and the space on which this is defined will be clear from context.

\section{The need for structure dependent sampling}\label{sec:discreteHaar}

To illustrate the need to account for sparsity structure when devising subsampling strategies, let us consider the case of recovering finite dimensional vectors, where we are given access to a subset of their Fourier coefficients and the sparsifying system is the two redundant discrete Haar wavelet frame. The Haar frame is defined in detail in the appendix \S \ref{sec:discreteHaar_def}. In the following example, $A$ will denote the discrete Fourier transform, and $D$ will denote the discrete Haar wavelet transform.

\paragraph{A numerical example}
Let $N=1024$ and consider the recovery of the two signals $x_1$ and $x_2$ shown in Figure \ref{fig:sparsity_struct} from subsampling their discrete Fourier coefficients by solving (\ref{eq:min_findim}). These signals are constructed such that $\nm{D x_1}_{0} = \nm{D x_2}_{0} = 100$, where we define the sparsity measure of a signal by $\nm{z}_{0} := \abs{\br{j: \br{\abs{z_j}\neq 0}}}$ for any $z\in\bbC^M$ with $M\in\bbN$. The sparsity patterns of $D x_1$ and $D x_2$ are shown in Figure \ref{fig:sparsity_struct}. Observe that compared to $Dx_2$, $D x_1$  has a higher proportion of large coefficients with respect to the higher scale frame elements. 
Let $\Omega_V$ index 130 of the rows of $A$ (12.7\% subsampling), so that the indices correspond to the first 41 Fourier coefficients of lowest frequencies plus 89 of the remaining coefficients drawn uniformly at random. The reconstruction of $R(x_1,\Omega_V)$ and $R(x_2, \Omega_V)$ from their partial Fourier samples are shown in the top row of Figure \ref{fig:signals}. Note that although the same sampling pattern is used for both reconstructions, and both signals have the same sparsity with respect to $D$, $R(x_2,\Omega_V)$ is an exact reconstruction of $x_2$ whilst $R(x_1,\Omega_V)$ incurs a relative error of  34.85\%. This simple example suggests that to subsample efficiently, it is not sufficient to consider sparsity alone. We remark also that unlike sampling with unstructured operators such as random Gaussian matrices, uniform random sampling will yield  poor reconstructions for both signals. The second row of Figure \ref{fig:signals} shows the reconstruction $R(x_1,\Omega_U)$ and $R(x_2,\Omega_U)$, where $\Omega_U$ indexes $130$ of the available coefficients uniformly at random. Finally, it is interesting to note that the high frequency samples indexed by $\Omega_V$ are required for an exact reconstruction of $x_2$ as an error is incurred when one simply samples the Fourier coefficients of lowest frequency (see the bottom row of Figure \ref{fig:signals}). 

\begin{remark}
In the context of sampling the Fourier transform of a signal, which is sparse with respect to some multiscale transform (such as wavelets, curvelets or shearlets), it is now commonly observed that uniform random sampling yields highly inferior results, when compared with variable density sampling patterns which focus on low frequencies. The numerical example in this section simply highlights this observation, and reminds us that the performance of these variable density sampling patterns are highly dependent on the sparsity structure of the underlying signal, and not just the sparsity level alone. Thus, there is a need for a theory which describes how the sparsity structure of the underlying signal should impact the choice of the sampling pattern.  

\end{remark}

\begin{figure}
\begin{center}
\begin{tabular}{@{\hspace{0pt}}c@{\hspace{0pt}}c@{\hspace{0pt}}}
 Zoom of $x_1$ & $x_2$\\
\includegraphics[width = 0.4\textwidth]{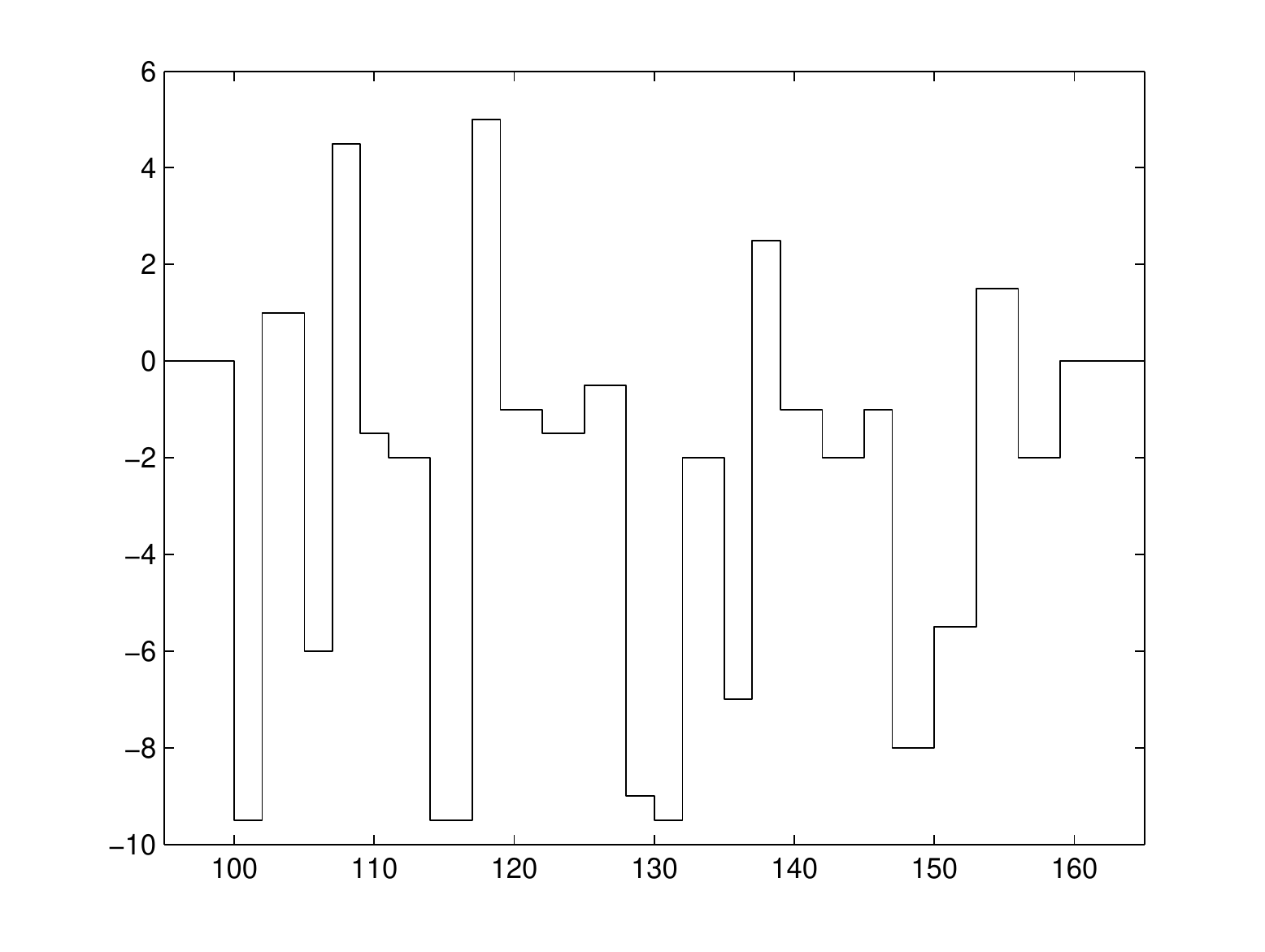}
& \includegraphics[width = 0.4\textwidth]{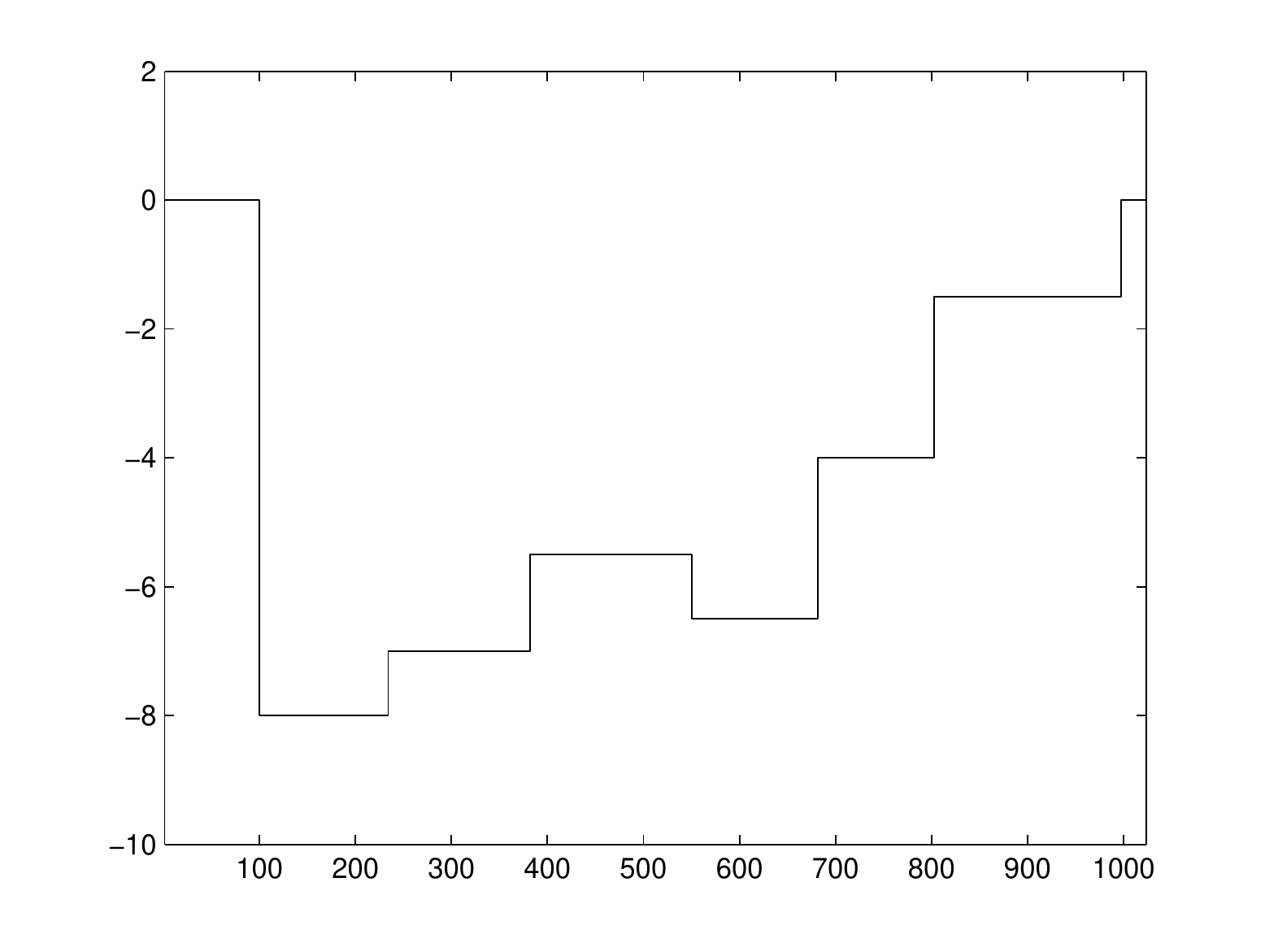}  \\
$\abs{D x_1}$ & $\abs{D x_2}$\\
\includegraphics[width = 0.4\textwidth]{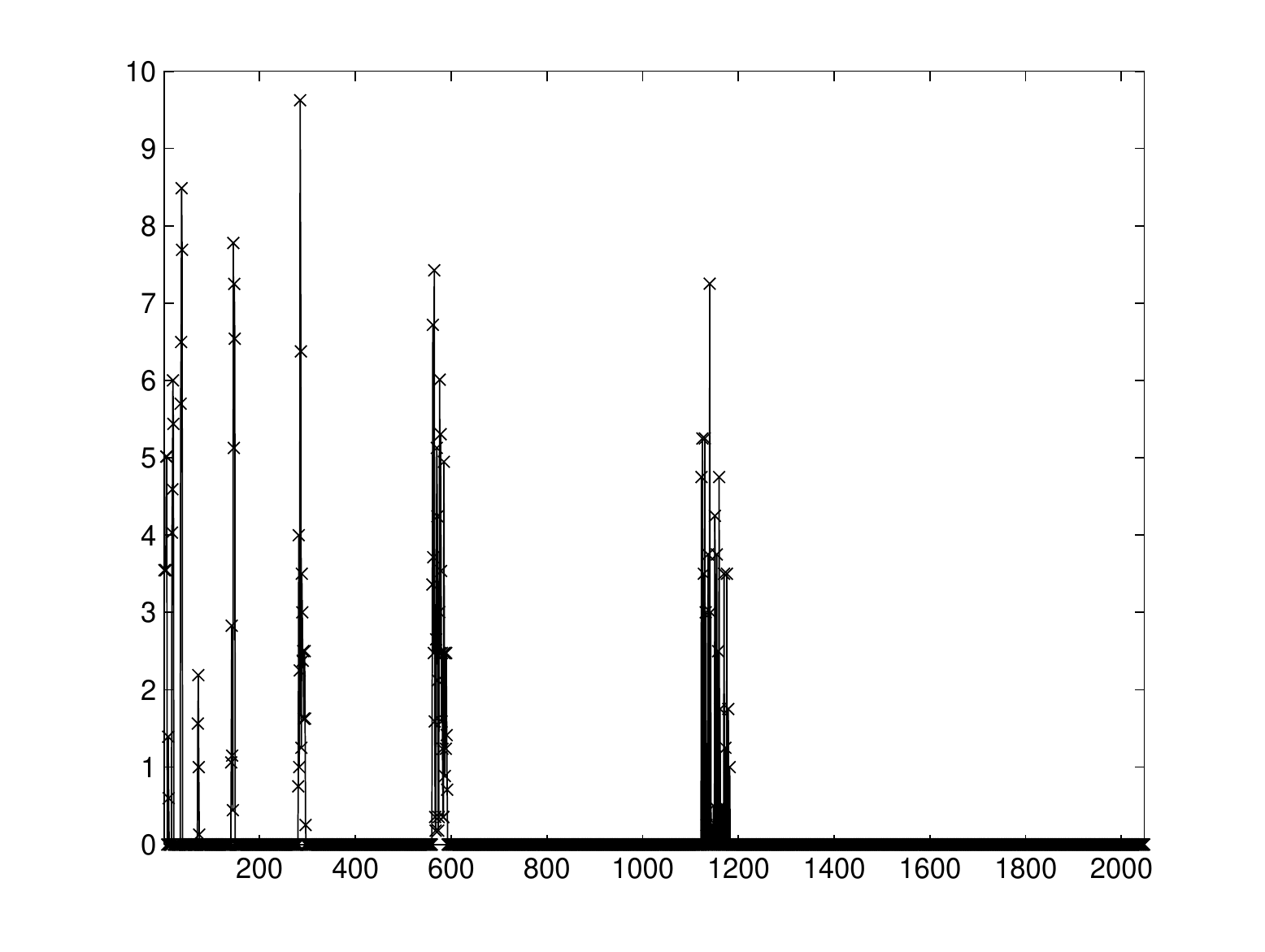} &\includegraphics[width = 0.4\textwidth]{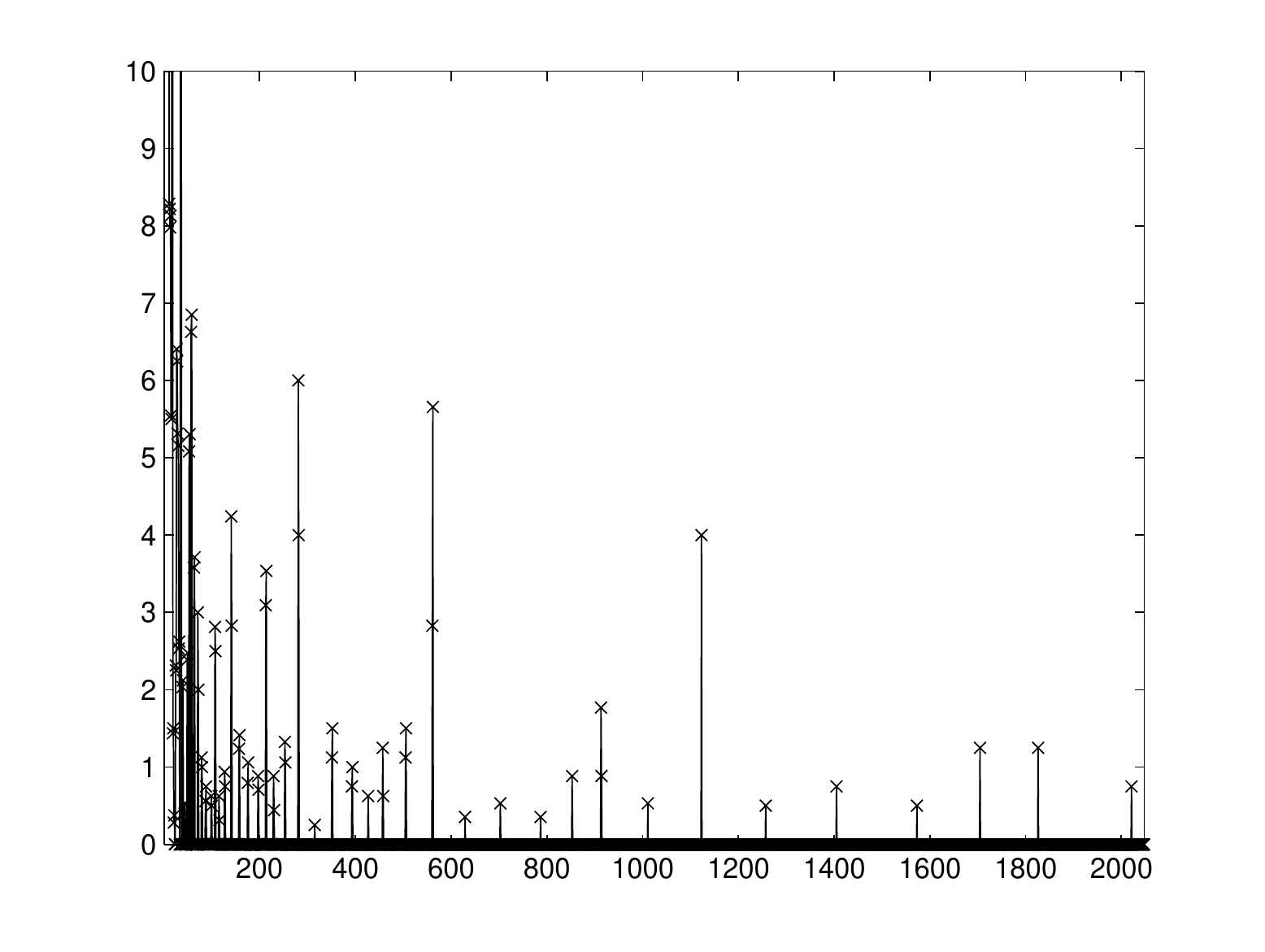}
\end{tabular}
\end{center}
\caption{
Top row: Two test signals. Only a zoom of $x_1$ is shown since it is supported only on the indices ranging between 100 and 158.  Both signals have equal sparsity -- for each $i=1,2$, $\norm{D x_i}_{0} = 100$.  The second to the bottom rows show the reconstructions from different sampling maps. Bottom row: the sparsity structure of $D x_1$ and $D x_2$. The graph for $\abs{D x_2}$ has been capped off at 20 to allow for a clear comparison with $\abs{D x_1}$. \label{fig:sparsity_struct}}
\end{figure}

\begin{figure}
\begin{center}
\begin{tabular}{@{\hspace{0pt}}c@{\hspace{0pt}}c@{\hspace{0pt}}c@{\hspace{0pt}}}
$\Omega_V$ (half-half)   & Zoom of $R(x_1,\Omega_V)$, Err = 34.9\%    & $R(x_2,\Omega_V)$, Err = 0\% \\
\includegraphics[width = 0.32\textwidth]{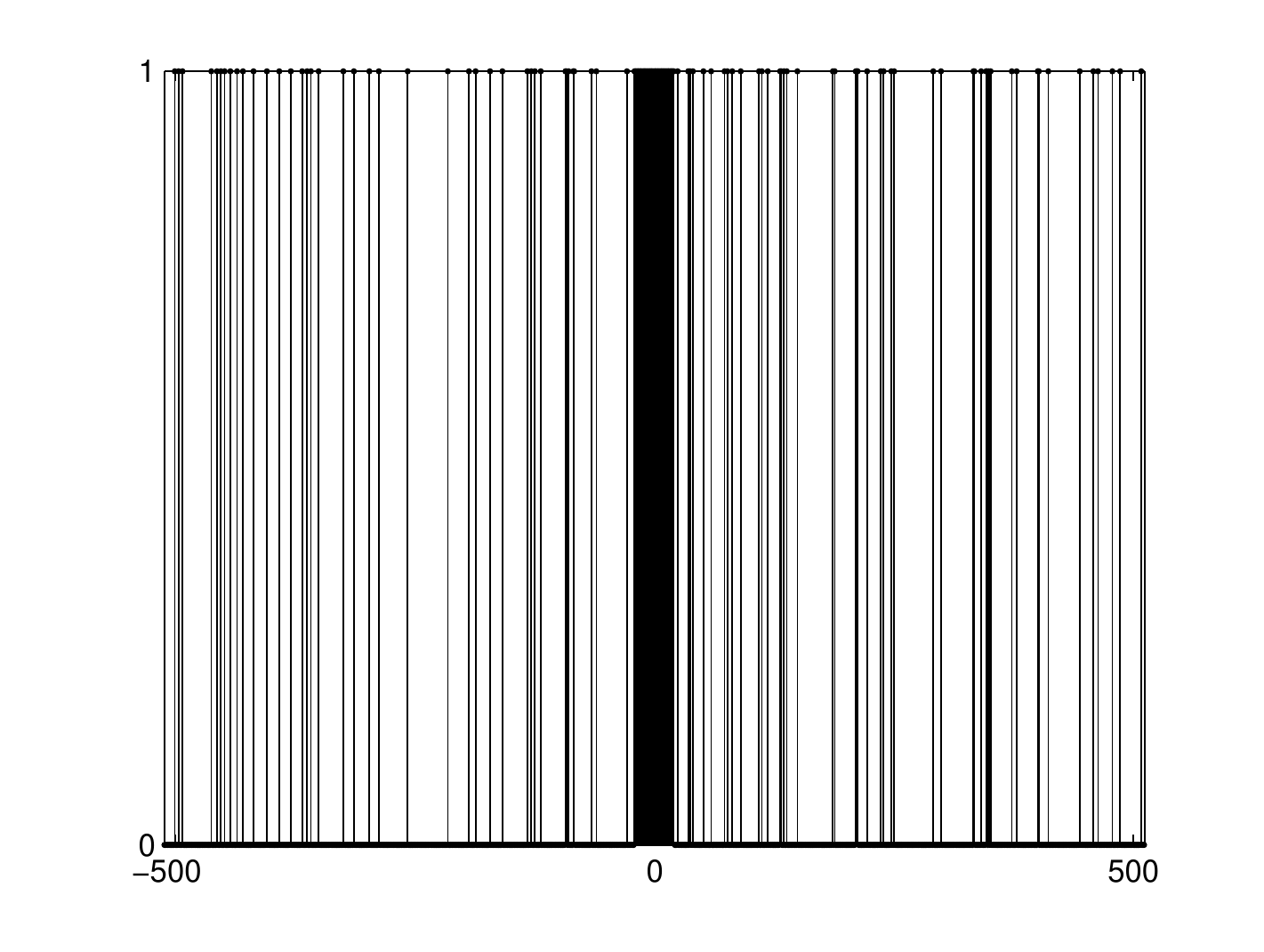} 
&\includegraphics[width = 0.32\textwidth]{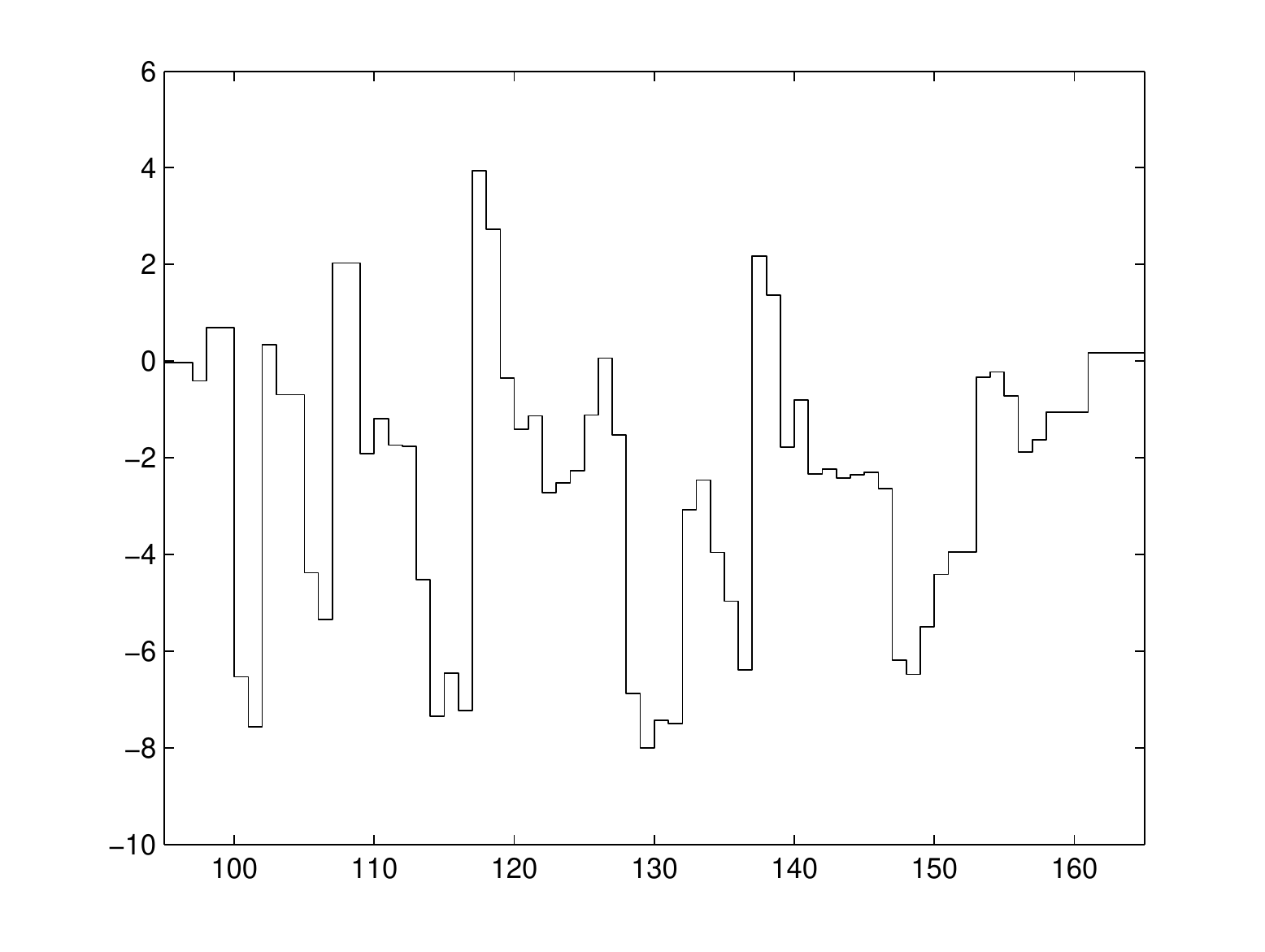}& \includegraphics[width = 0.32\textwidth]{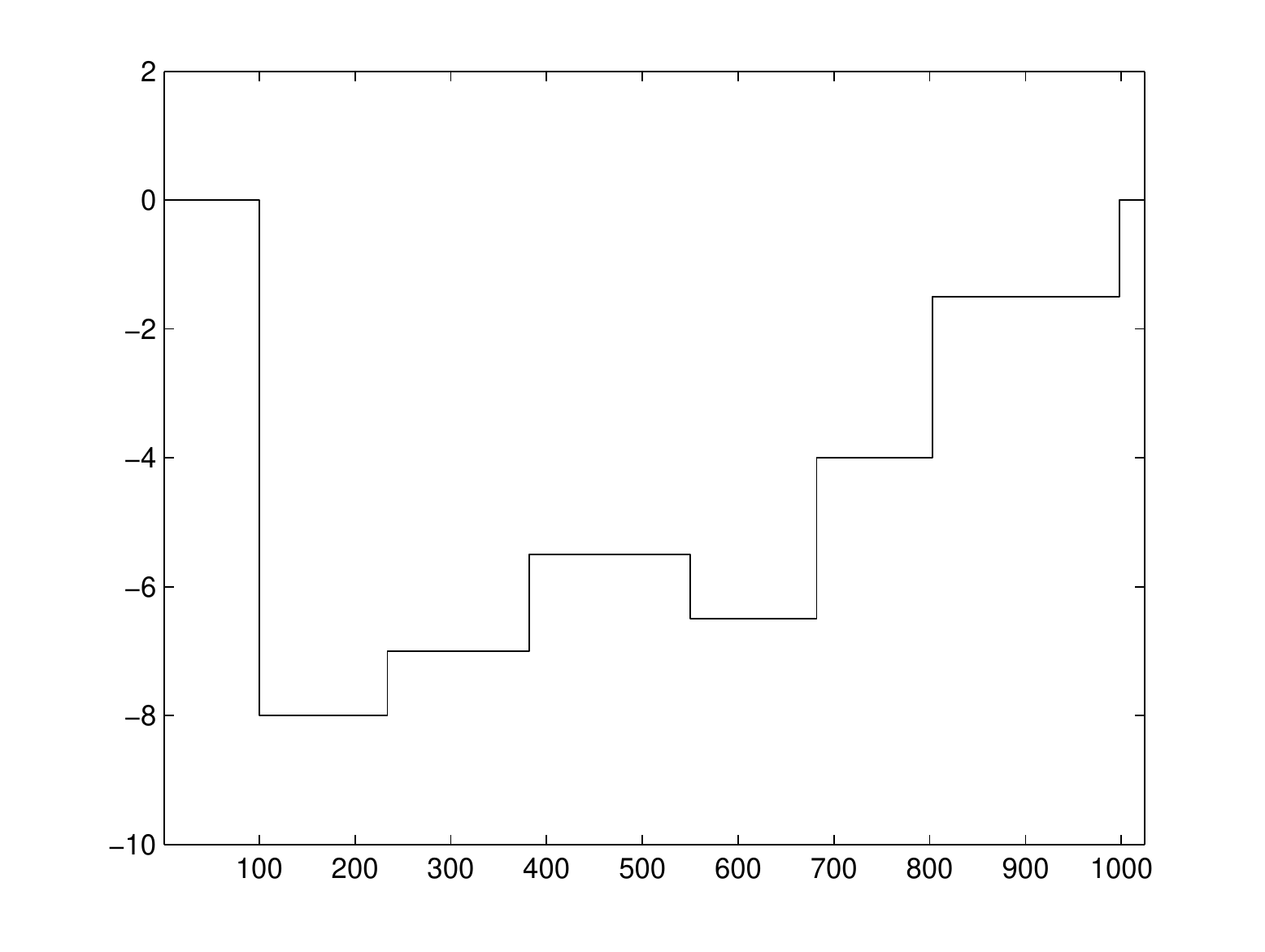}  \\
$\Omega_U$ (unif. rand.) & Zoom of $R(x_1,\Omega_U)$, Err = 28.7\%   & $R(x_2,\Omega_U)$, Err = 97.0\% \\
\includegraphics[width = 0.32\textwidth]{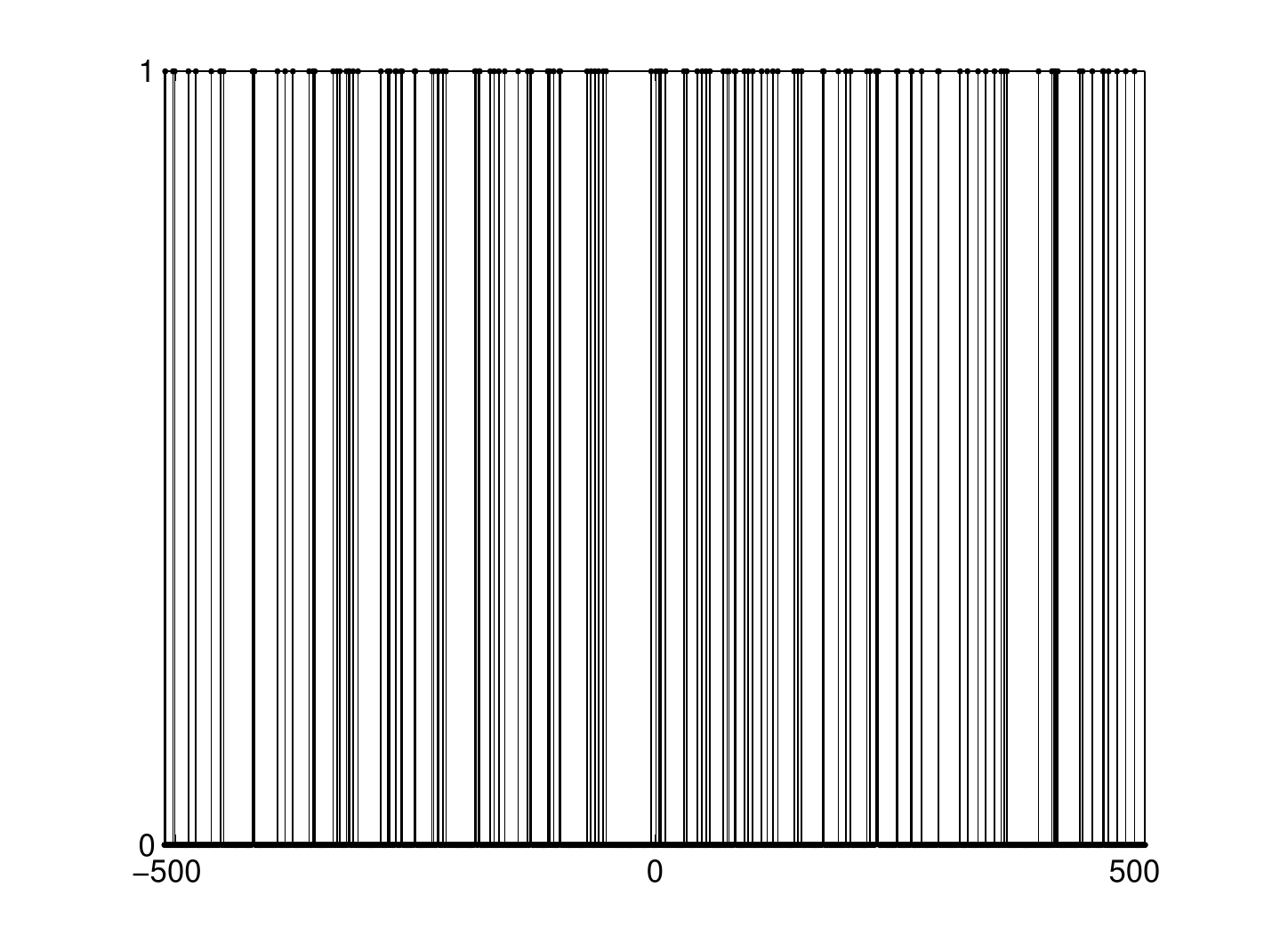} 
&\includegraphics[width = 0.32\textwidth]{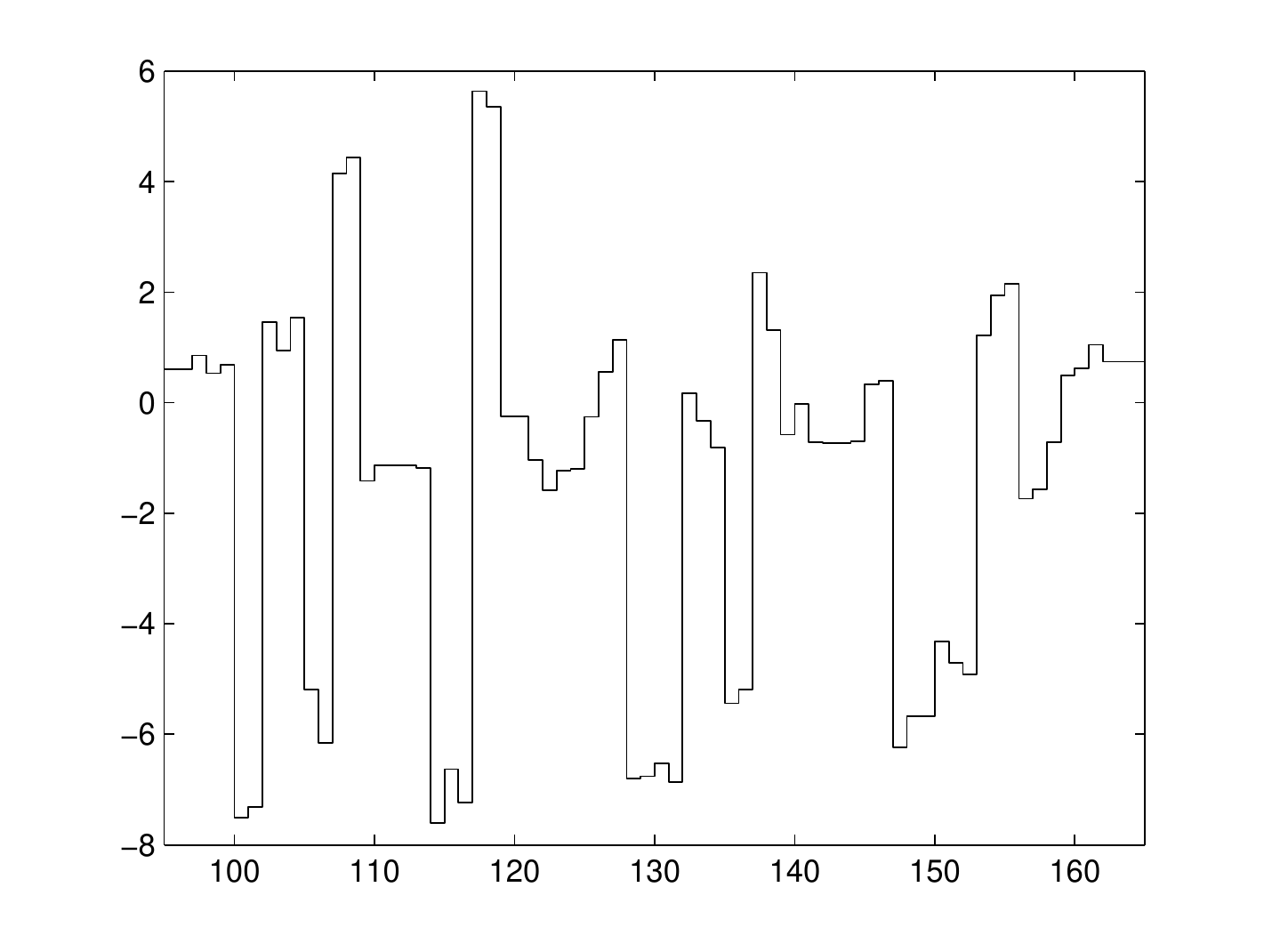}& \includegraphics[width = 0.32\textwidth]{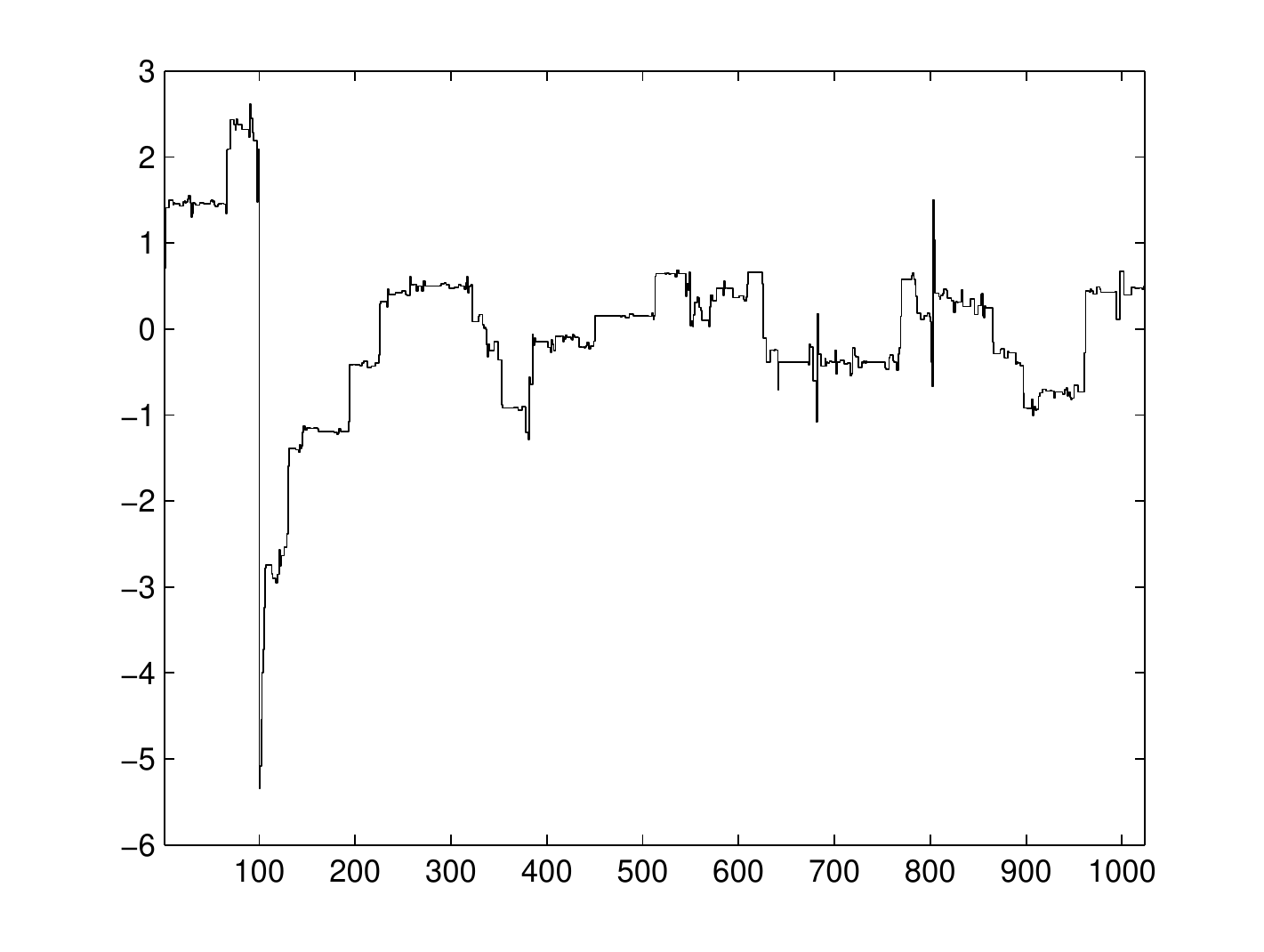}  \\
$\Omega_L$ (low freq.) & Zoom of $R(x_1,\Omega_L)$, Err = 74.8\%   & $R(x_2,\Omega_L)$,  Err = 5.0\% \\
\includegraphics[width = 0.32\textwidth]{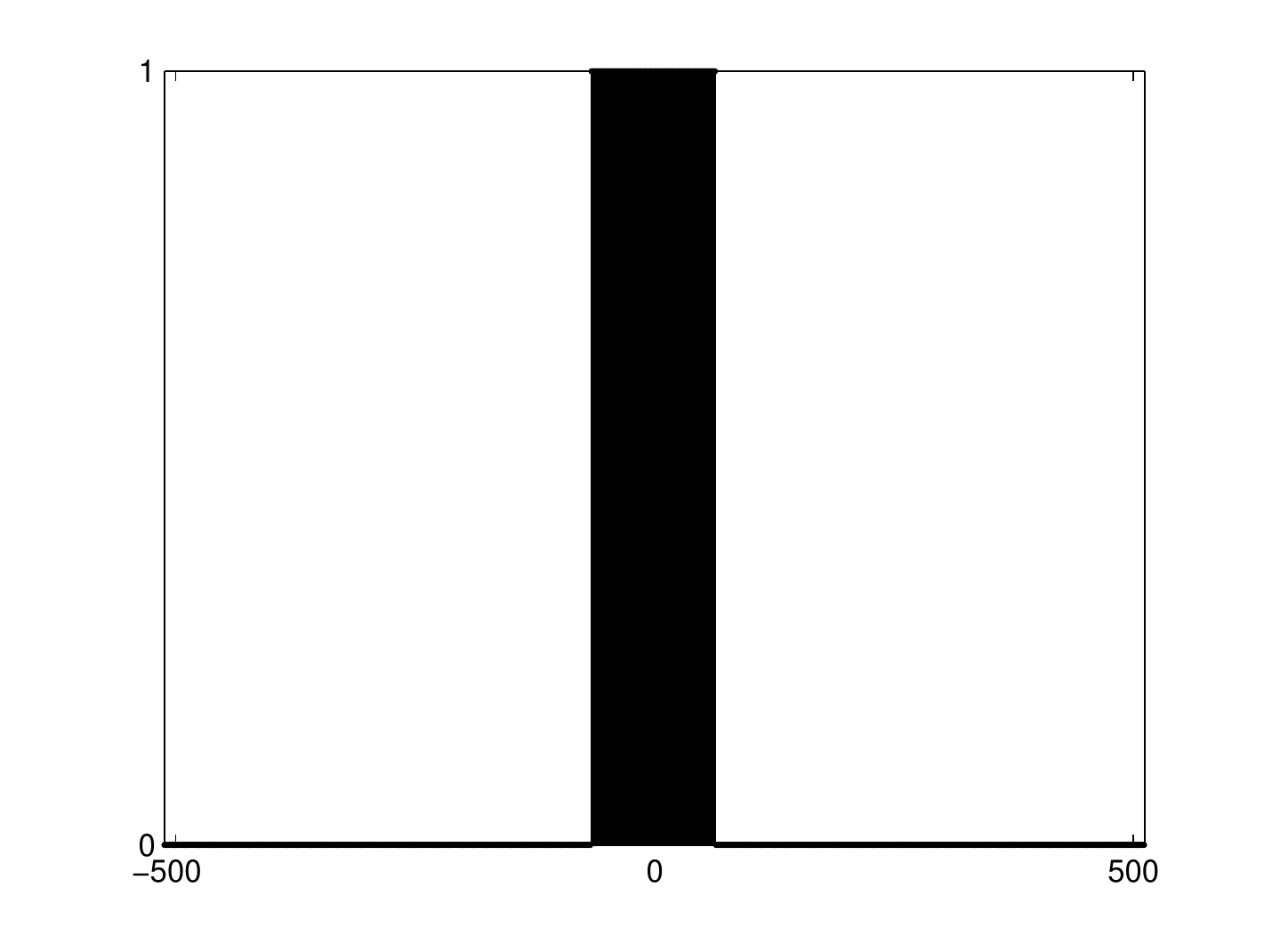} 
&\includegraphics[width = 0.32\textwidth]{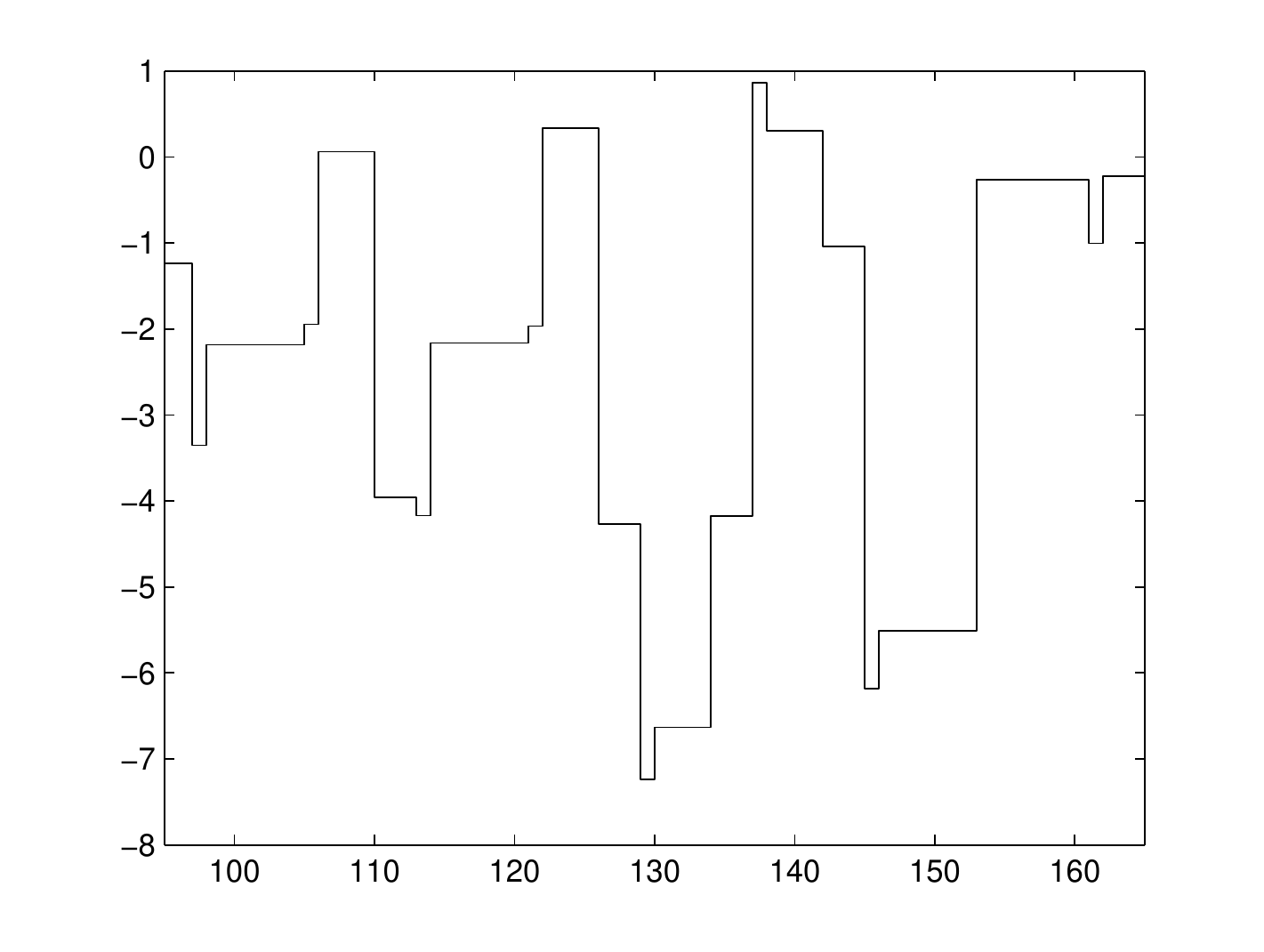}& \includegraphics[width = 0.32\textwidth]{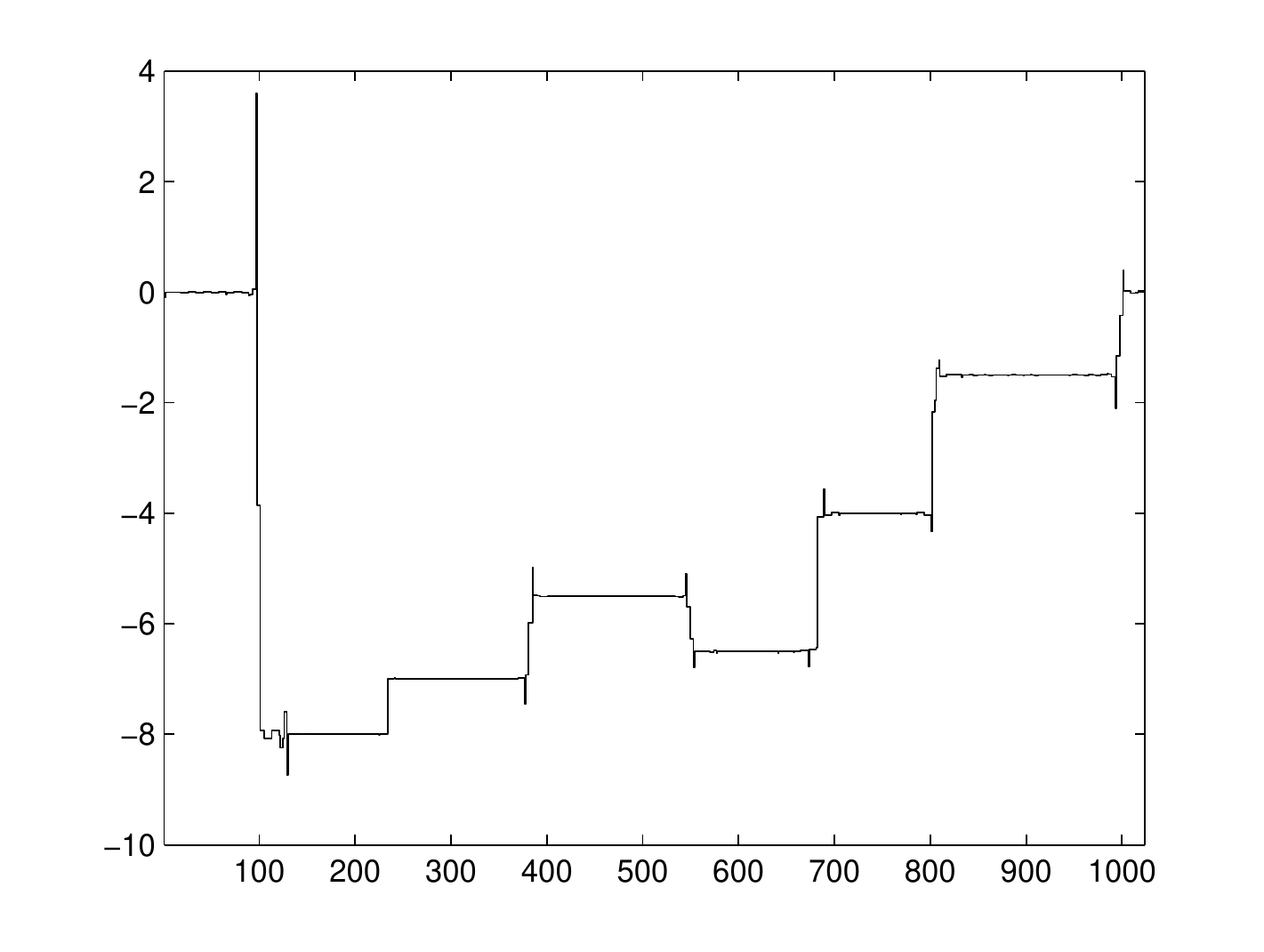}  
\end{tabular}
\end{center}
\caption{Reconstructions of $x_1$ and $x_2$ obtained by solving (\ref{eq:min_findim}) with different sampling maps $\Omega$ which index 130 of their Fourier coefficients (12.7\% subsampling).  $\Omega_V$ indexes the first 41  coefficients of lowest frequencies, plus 89 the remaining coefficients chosen uniformly at random. $\Omega_U$ indexes 130 of the  coefficients  uniformly at random. $\Omega_L$ indexes the 130 coefficients of lowest frequencies.
  \label{fig:signals}}
\end{figure}

\section{Structured sampling with orthonormal systems}\label{sec:onb}
The main result of this paper will be an extension of the abstract result of \cite{adcockbreaking} to the case where the sparsifying transform is a tight frame. This section recalls the key concepts introduced in \cite{adcockbreaking} to analyse the use of variable density sampling schemes for orthonormal sparsifying bases. We first remark that although compressed sensing originally considered only finite dimensional vector spaces, the applications in which variable density sampling tend to be of interest are more naturally modelled on infinite dimensional Hilbert spaces. For this purpose, a Hilbert space framework for compressed sensing was introduced in \cite{BAACHGSCS} and \cite{adcockbreaking}.

For a Hilbert space $\cH$, and given orthonormal bases $\br{\psi_j}_{j\in\bbN}$ (the sampling vectors) and $\br{\varphi_j}_{j\in\bbN}$ (the sparsifying vectors), define the operators
\be{\label{eq:VD_def}
V: \cH \to \ell^2(\bbN), \quad  f \mapsto (\ip{f}{\psi_j})_{j\in\bbN},
\qquad D: \cH \to \ell^2(\bbN), \quad  f \mapsto (\ip{f}{\varphi_j})_{j\in\bbN}.
}
Suppose we wish to recover some $f\in\cH$ from samples of the form $y = (\ip{f}{\psi_j})_{j\in\Omega} + \eta = P_\Omega V f + \eta$ for some $\Omega \subset \bbN$ and noise vector $\eta$ of $\ell^2$-norm at most $\delta$.
 A key question in compressed sensing is how solutions to the following minimization problem allows one to exploit the sparsity of some $f\in\cH$ with respect to $D$ to obtain accurate recovery from a minimal number of samples. 
 \be{\label{eq:min_orth}
\inf_{g\in\cH, D g\in\ell^1(\bbN)} \nm{D g}_{\ell^1} \text{ subject to } \nm{y -  P_\Omega V g}_{\ell^2}\leq \delta.
}
The coherence (defined below) of the operator $VD^*$ has been recognized to be an important factor in determining the minimal cardinality of the sampling set $\Omega$. Note that this can be seen as a measure of the correlation between the sampling system associated with $V$ and the sparsifying system associated with $D$.
\begin{definition}[Coherence]
Let $U$ be a bounded linear operator on  $\ell^2(\bbN)$ (or let $U \in\bbC^{N\times N}$ for some $N\in\bbN$) be such that $\nm{U e_j}_{\ell^2}=1$ for all $j\in\bbN$ (or $j=1,\ldots, N$). Let $\br{e_j:j\in\bbN}$ be the canonical basis of $\ell^2(\bbN)$ (or $\bbC^N$). The coherence of $U$ is defined as $\mu(U)= \sup_{k,j} \abs{\ip{U e_j}{e_k}} $.
\end{definition}
For the case where $VD^* \in\bbC^{N\times N}$ is a  finite dimensional isometry, the main result of \cite{candes2011probabilistic} showed that if $\Omega\subset \br{1,\ldots, N}$ consists of $\ord{s \, \mu^2(VD^*) \, N \, \log N}$ samples drawn uniformly at random, where $f$ is  $s$-sparse, then any solution $\hat f$ to  (\ref{eq:min_orth}) satisfies $\nmu{\hat f - f}\leq C \delta$ for some universal constant $C>0$. Furthermore, one cannot improve upon the estimate of $\ord{s \, \mu^2(VD^*) \, N \, \log N}$. Thus, for the recovery of sparse signals, the minimal sampling cardinality is completely determined by this coherence quantity.

 Unfortunately, when $\mu(VD^*)\approx 1$, this result merely concludes that $\Omega$ must index all available samples. This is especially problematic because when $VD^*$ is a bounded linear operator defined on the infinite dimensional Hilbert space $\ell^2(\bbN)$ -- it is necessarily the case that $\mu(VD^*) \geq c >0$ for some constant $c$ and one cannot expect the coherence of any finite dimensional discretization of $VD^*$ to be of order $\ord{N^{-1/2}}$ (see \cite{adcockbreaking} for a detailed explanation of this phenomenon). In the case where $V$ is associated with a Fourier basis and $D$ is associated with a wavelet basis, it is necessarily the case that $\mu(VD^*) = 1$.

The key idea of \cite{adcockbreaking} is to recognize that by placing additional assumptions on the sparsity or compressibility structure of the underlying signal, one can make non trivial statements on how $\Omega$ can be chosen in accordance to the underlying sparsity.  Thus, to consider how one should draw samples from the first $M$ samples in order to accurately recover $f\in\cH$, with $\nm{P_\Delta D  f}_{\ell^1} \ll \nm{f}_{\cH}$ for some $\Delta \subset \br{1,\ldots, N}$ with $\abs{\Delta}=s$, one approach is to divide the sampling and sparsifying vectors into levels then analyse the correspondence between the different sampling and sparsifying levels. The main theoretical result from \cite{adcockbreaking} is based on three principles:
\begin{itemize}
\item Multilevel sampling - instead of considering sampling uniformly at random across all available samples, partition the samples into levels and consider sampling uniformly at random with different densities at each level. This model was introduced to analyse the effects of nonuniform sampling patterns.
\item Local coherence - the coherence of partial sections of $VD^*$.
\item Sparsity in levels - instead of considering sparsity across all available coefficients, partition the coefficients into levels and consider the sparsity within each level.
\end{itemize}

We define each of these concepts below. 
\defn{[Multilevel sampling]
\label{multi_level_dfn}
Let $r \in \bbN$, $\mathbf{M} = (M_1,\ldots,M_r) \in \bbN^r$ with $0=M_0< M_1 
< \ldots < M_r$, $\mathbf{m} = (m_1,\ldots,m_r) \in \bbN^r$, with $m_k \leq 
M_k-M_{k-1}$, $k=1,\ldots,r$, and suppose that
$$
\Omega_k \subseteq \{ M_{k-1}+1,\ldots,M_{k} \},\quad | \Omega_k | = m_k,\quad 
k=1,\ldots,r,
$$
are chosen uniformly at random.  We refer to the set
$
\Omega = \Omega_{\mathbf{M},\mathbf{m}} = \Omega_1 \cup \cdots \cup \Omega_r
$
as an $(\mathbf{M},\mathbf{m})$-multilevel sampling scheme.
}

\subsection{Sparsity in levels}

The notion of sparsity in levels is defined as follows. As explained below, this notion is particularly important when considering wavelet sparsity for imaging purposes.
\defn{[Sparsity in levels]
\label{def:Asy_Sparse}
Let $x$ be an element of either $\bbC^N$ or $\ell^2(\bbN)$. For $r \in \bbN$ let 
$\mathbf{N} = (N_1,\ldots,N_r) \in \bbN^r$ with $0=N_0< N_1 < \ldots < N_r$ 
and $\mathbf{s} = (s_1,\ldots,s_r) \in \bbN^r$, with $s_k \leq N_k - N_{k-1}$, 
$k=1,\ldots,r$.  We say that $x$ is 
$(\mathbf{s},\mathbf{N})$-sparse if, for each $k=1,\ldots,r$,
$
\Delta_k : = \mathrm{supp}(x) \cap \{ N_{k-1}+1,\ldots,N_{k} \},
$
satisfies $| \Delta_k | \leq s_k$.  We denote the set of 
$(\mathbf{s},\mathbf{N})$-sparse vectors by $\Sigma_{\mathbf{s},\mathbf{N}}$.
}

\defn{[$(\mathbf{s},\mathbf{N}$)-term approximation]\label{def:sparsity_level}
Let   $x = (x_j )$ be
an element of either $\bbC^N$ or $\ell^2(\bbN)$.
We define the ($\mathbf{s},\mathbf{N}$)-term approximation
\be{
\label{sigma_s_m}
\sigma_{\mathbf{s},\mathbf{N}}(x) = \min_{\eta \in 
\Sigma_{\mathbf{s},\mathbf{N}} } \| x - \eta \|_{\ell^1}.
}
}

As well as the level sparsities $s_k$ defined in Definition \ref{def:sparsity_level}, we shall also require the notion of a 
relative sparsity, which takes into account the sampling operator $V$ and will account for how different levels interfere with each other.

\defn{[Relative sparsity]
\label{def:rel_sparsity}
Let $V, D\in\cB(\cH, \cH')$ where $\cH$ is a Hilbert space and $\cH'$ is either $\mathbb{C}^{N}$ or $\ell^2(\bbN)$.  Let $\mathbf{s} = (s_j)_{j=1}^r \in \bbN^r$,
$\mathbf{N} = (N_j)_{j=1}^r  \in \bbN^r$ and $\mathbf{M} = (M_j)_{j=1}^r \in 
\bbN^r$ with $0=N_0< N_1 < \cdots < N_r$ and $0=M_0< M_1 < \cdots < M_r$. 
 For $1 \leq k \leq r$, the 
$k^{\rth}$ relative sparsity is given by
$$
\hat \kappa_k = \hat \kappa_k(\mathbf{N},\mathbf{M},\mathbf{s}) =  \max_{g \in 
\Theta}\|P_{\Gamma_k} V g\|^2,
$$
where $\Gamma_k = (M_{k-1}, M_k]\cap \bbN$ and $\Theta$ is the set
\bes{
\Theta = \{g\in\cH :  , \|Dg\|_{\ell^{\infty}} \leq 1, \, 
|\mathrm{supp}( P_{\Lambda_l}Dg)| = s_l, \, l=1,\hdots, r\}.
}
where $\Lambda_l = (N_{k-1}, N_k]\cap \bbN$.

}

\subsubsection*{The Fourier/wavelets case}

\paragraph{On level sparsities}

It has been established that natural images are not simply sparse in their wavelet coefficients, but exhibit a distinctive `tree-structure' in their coefficients \cite{CrouseEtAlWaveleTree}. Given a wavelet basis $\br{\varphi_j}_{j\in\bbN}$, it is often the case that a typical image with sparse approximation $\sum_{j\in\Delta} \alpha_j \varphi_j$ will actually not be sparse with respect to the wavelets of low scales, but will become increasingly sparse with respect to the wavelets of higher scales. In particular, if $\br{N_k}_{k\in\bbN}$ corresponds to the wavelet scales so that $\br{\varphi_j}_{j\leq N_k}$ consists of all wavelets up to the $k^{\rth}$ scale, and $s_k = \abs{\Delta \cap (N_{k-1}, N_k]}$ is the sparsity at the $k^{\rth}$ wavelet scale, then one typically observes that although $s_1/N_1\approx 1$, one has \emph{asymptotic sparsity} with $s_k/(N_k-N_{k-1}) \to 0$ as $k$ increases. This phenomenon is illustrated in Figure \ref{fig:bus2}.

 Thus, for the purpose of reconstructing natural images, it is perhaps too general to consider the recovery of all sparse wavelet coefficients and it  suffices to consider the recovery of images whose sparse representations exhibit asymptotic sparsity. This is the motivation behind the concept of sparsity in levels.

\paragraph{On relative sparsities}
In the case where $V$ is the Fourier sampling operator and $D$ is the analysis operator associated with an orthonormal basis,  one can in fact show that the change of basis matrix $VD^*\in\cB(\ell^2(\bbN))$ is near block diagonal and by letting $\mathbf{M}$ and $\mathbf{N}$ correspond to wavelet scales, 
$$
\hat \kappa_k(\mathbf{N}, \mathbf{M}, \mathbf{s}) \lesssim \sum_{j=1}^r s_j A^{-\abs{j-k}},
$$
for some $A > 1$ which depends only on the given wavelet basis. So, the dependence of the $k^{\rth}$ relative sparsity on each $s_j$ decays  exponentially in $\abs{j-k}$ and moreover, it follows that $\sum_{k}\hat \kappa_k \lesssim \sum_{k}s_k$. The reader is referred to \cite{adcockbreaking} for a proof of this.

\begin{figure}
\begin{center}
\begin{tabular}{@{\hspace{0pt}}c@{\hspace{3pt}}c@{\hspace{3pt}}c@{\hspace{0pt}}}
\includegraphics[ width=0.27\textwidth]{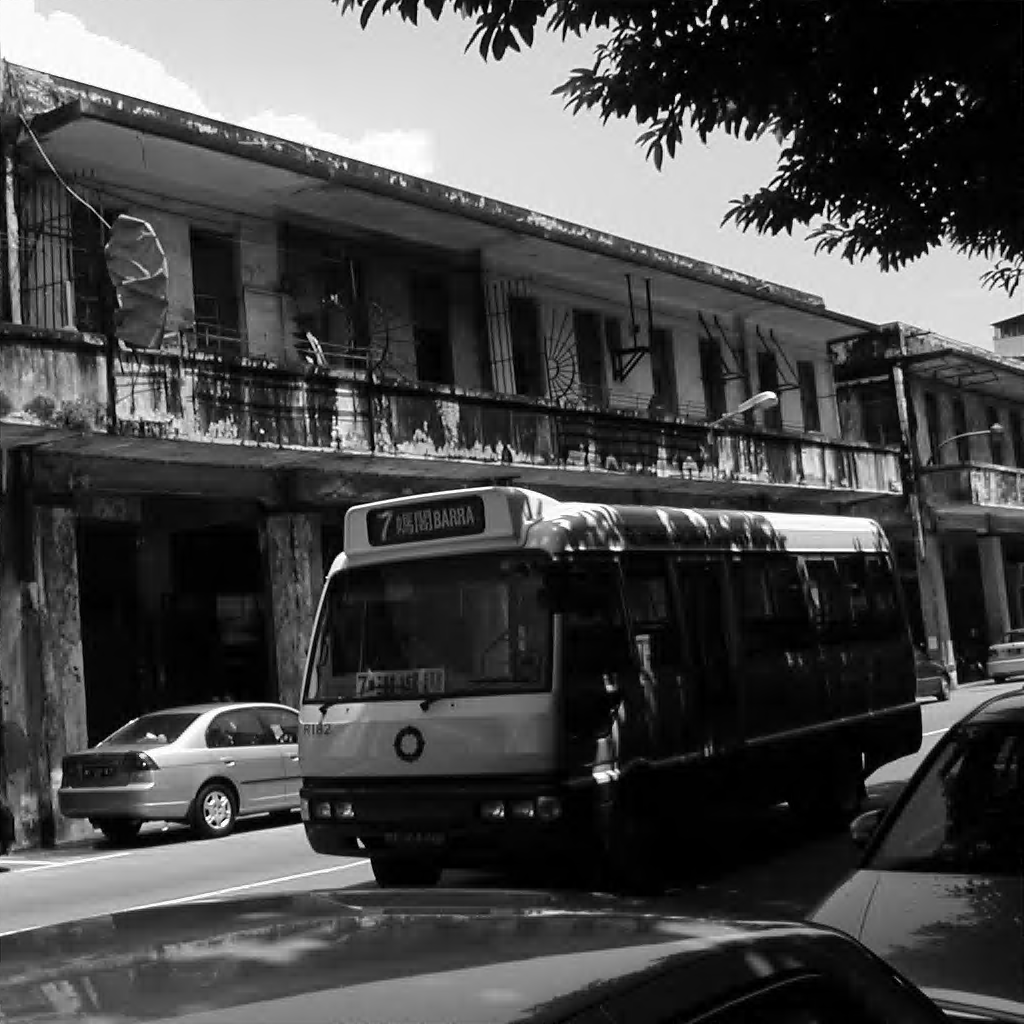} &
\includegraphics[width=0.27\textwidth]{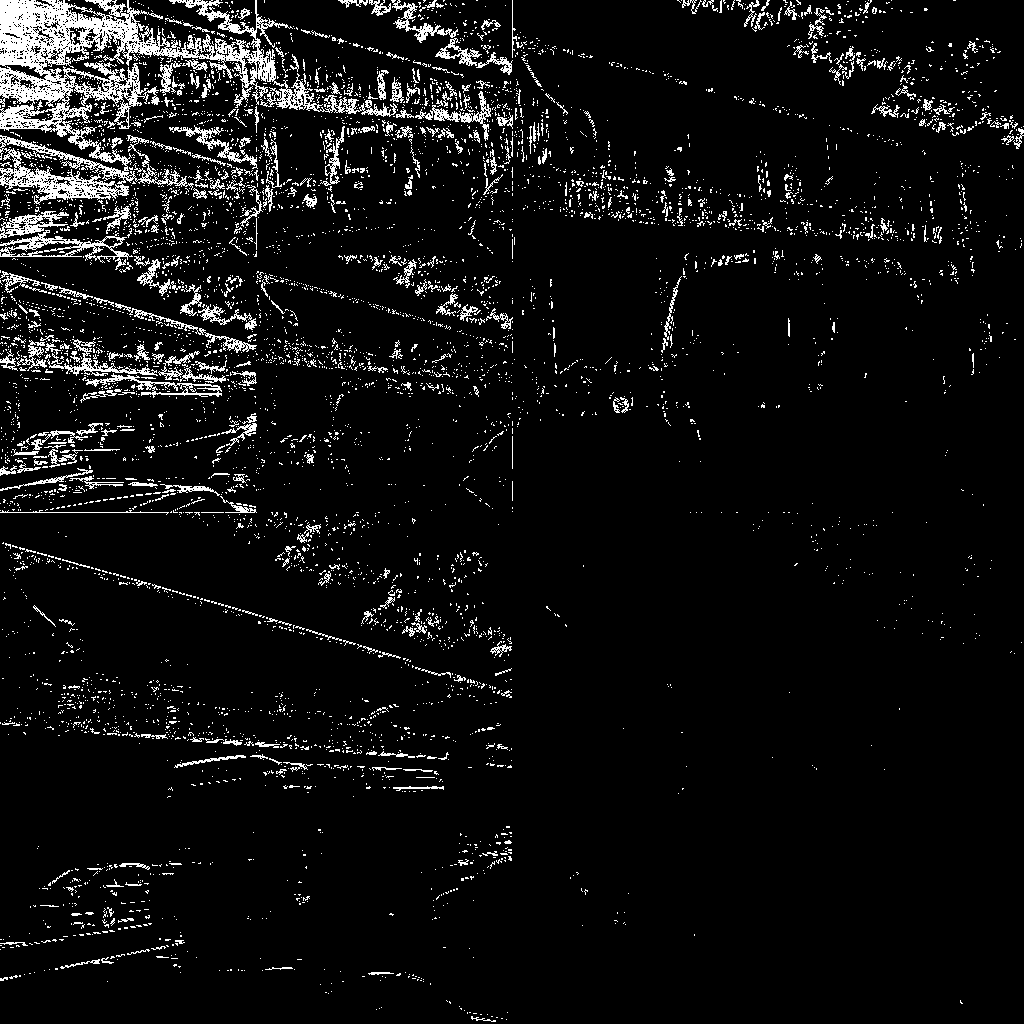}&
\includegraphics[width=0.43\textwidth]{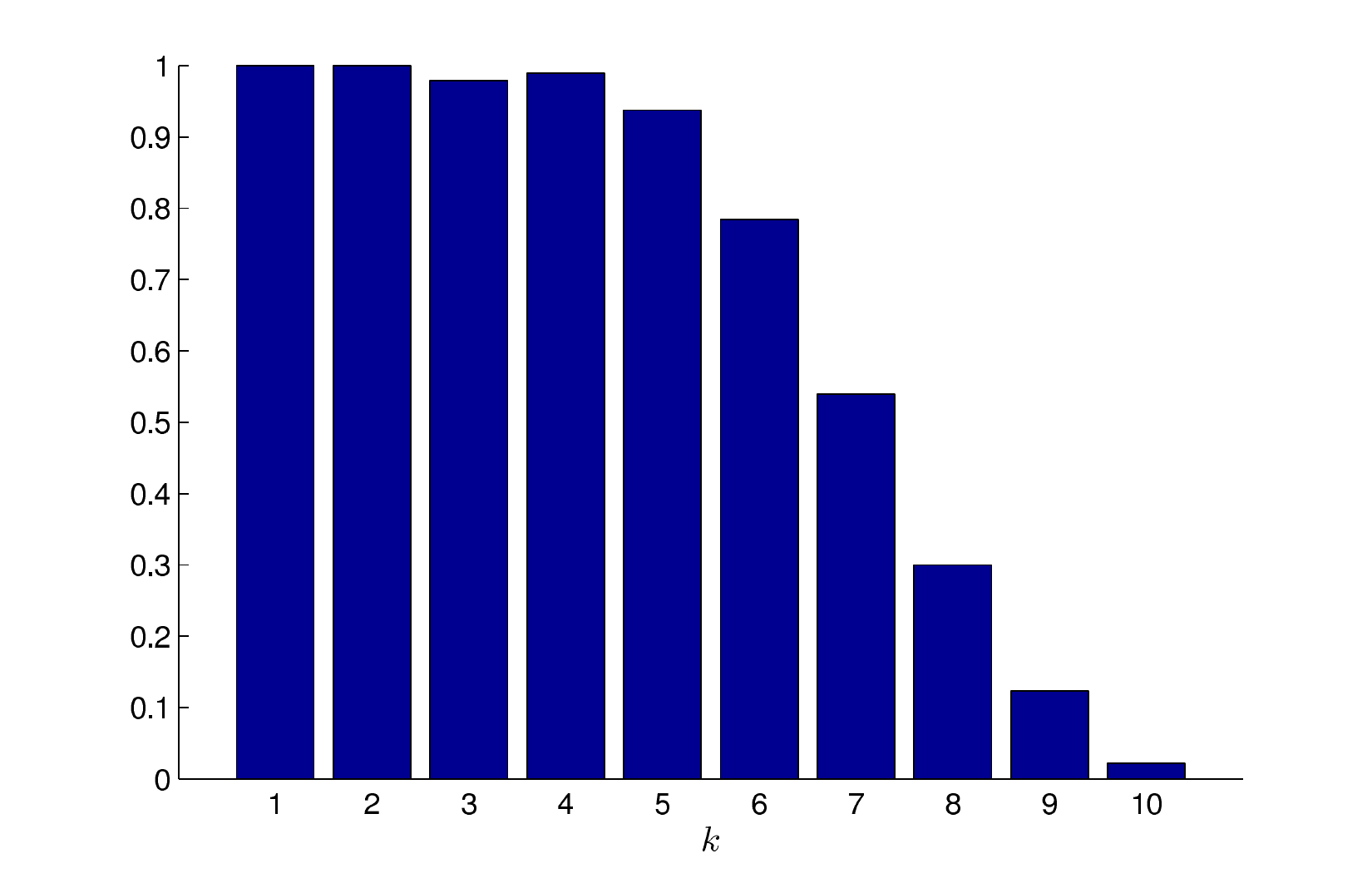}
\end{tabular}
\end{center}
\caption{ Left: reconstruction from the largest $6\%$ of the Daubechies-4 wavelet coefficients of a $1024\times 1024$ image. Centre: location of coefficients in the sparse  representation - coefficients are ordered in increasing wavelet scales away from the top left corner. Right: fraction of  coefficients at each wavelet scale $k$ which contribute to the sparse representation.
\label{fig:bus2}}
\end{figure}

\subsection{Local coherence}

Although the coherence between the sampling and sparsifying systems is a crucial concept in the understanding of the minimal sampling cardinality required for the recovery of sparse signals, there are important systems of interest in applications where it is simply too crude to consider coherence alone. Instead, we require the more refined notion of local coherence.
\defn{[Local coherence]\label{def:loc_coherence}
Let $V,D\in\cB(\cH,\cH')$ where $\cH$ is a Hilbert space and $\cH'$ is either $\bbC^{N}$ or $\ell^2(\bbN)$.
Let $\mathbf{N} = (N_1,\ldots,N_r) \in \bbN^r$ and $\mathbf{M} = 
(M_1,\ldots,M_r) \in \bbN^r$ with $0=N_0<  N_1 < \cdots < N_r $ and $0=M_0< M_1 < 
\cdots < M_r $.
For $k=1,\ldots, r$, let
$
\Gamma_k = \br{M_{k-1} +1 , \ldots, M_k}.
$
For $k=1,\ldots, r-1$, let $\Lambda_k = \br{N_{k-1} +1 , \ldots, N_k}$ and let $\Lambda_r = \br{n\in\bbN: n>N_r}$.
The $(k,l)^{\rth}$ local coherence between $V$ and $D$ with respect to 
$\mathbf{N}$ and $\mathbf{M}$ is given by
\eas{
\mu_{\mathbf{N},\mathbf{M}}(k,l) &= 
\sqrt{\mu(P_{\Gamma_k} VD^* P_{\Lambda_l}) \,  
\mu(P_{\Gamma_k} VD^*)},\quad k=1,\ldots, r,\quad l=1,\ldots,r.
}
}

\subsubsection*{The Fourier/wavelets case}
If $VD^*$ is constructed from  any orthonormal wavelet basis with  Fourier sampling, then it is necessarily the case that $\mu(VD^*) = 1$. However, it is only the initial section of $VD^*$ associated with low Fourier frequencies and low wavelet scales that has high coherence. In particular, one can show that 
$$
\mu( P_{[N]}^\perp VD^*), \mu(VD^*  P_{[N]}^\perp) = \ord{N^{-1/2}}.
$$
Finally, we remark that this property of asymptotic incoherence (decay in the coherence away from initial finite sections) is not unique to the  Fourier/wavelets case, but can also be observed for other representation systems such as Fourier/Legendre polynomial systems.  In the Fourier/wavelets case, it is this decay in the local coherences that makes it possible to exploit sparsity to subsample the  Fourier coefficients.

\subsection{Recovery guarantees in the case of orthonormal sparsifying transforms }
When we are considering the recovery of an infinite dimensional object by drawing finitely many samples, one may ask the following question: What is the range of the samples, $M$, that we should sample from in order to recover a sparse representation with respect to the first $N$ sparsifying elements? This question is addressed by the balancing property.

\begin{definition}[Balancing property \cite{adcockbreaking}]\label{balancing_property}
Let $VD^* \in \mathcal{B}(\ell^2(\mathbb{N}))$ be an isometry.  Then $M \in \bbN$ 
and $K \geq 1$ satisfy the  balancing property with respect to 
$V$, $D$, $N \in \bbN$ and $s \in \bbN$ if 
\begin{equation}\label{conditions1}
\begin{split}
\| P_{[N]}D V^* P_{[M]}^\perp VD^* P_{[N]} \|_{\ell^{\infty} \rightarrow \ell^{\infty}} \leq  
\frac{1}{8}\left(\log_2^{1/2}\left(4 \sqrt{s}KM\right)\right)^{-1},
\end{split}
\end{equation}
and
\begin{equation}\label{conditions34}
\begin{split}
\|P_{[N]}^{\perp}D V^*  P_{[M]} VD^* P_{[N]}\|_{\ell^{\infty} \rightarrow \ell^{\infty}}
\leq \frac{1}{8},
\end{split}
\end{equation}
where $\nm{\cdot}_{\ell^\infty \rightarrow \ell^{\infty}}$ is the norm on 
$\cB(\ell^{\infty}(\bbN))$.

\end{definition}

We now recall the main result of \cite{adcockbreaking} which informs on how multilevel sampling will depend on local coherences and the underlying sparsity structure.  For this, we require the following notation:
$$
\tilde{M} = \min\{i\in\bbN: \max_{k\geq i}\|  P_{[M]}  U e_k \| \leq 
1/(32 q^{-1} \sqrt{s})\},
$$
where $M$, $s$ and $q$ are as defined below.

\thm{ \cite{adcockbreaking}
\label{thm:main_bc}
Let $VD^*$ be an isometry either on $\ell^2(\bbN)$ or $\bbC^N$. Let $f \in 
\cH$.  Suppose that $\Omega = \Omega_{\mathbf{M},\mathbf{m}}$ is a 
multilevel sampling scheme, where $\mathbf{M} = (M_1,\ldots,M_r) \in \bbN^r$ 
and $\mathbf{m} = (m_1,\ldots,m_r) \in \bbN^r$.  Let 
$(\mathbf{s},\mathbf{N})$, where $\mathbf{N} = (N_1,\ldots,N_r) \in \bbN^r$, 
$N_1 < \ldots < N_r$, and $\mathbf{s} = (s_1,\ldots,s_r) \in \bbN^r$, be any 
pair such that the following holds: 
\begin{enumerate}
\item[(i)] the parameters
$
M=M_r, q^{-1} = \max_{k=1,\ldots,r} \left \{ \frac{M_{k}-M_{k-1}}{m_k} 
\right \},
$
satisfy the balancing property with respect to $V$, $D$, $N:= N_r$ and $s : 
= s_1+\ldots + s_r$;
\item[(ii)] for $\epsilon \in (0,e^{-1}]$,
\bes{
1 \gtrsim \frac{M_k-M_{k-1}}{m_k} \, \log(s\epsilon^{-1})\, 
\log\left(q^{-1} \tilde M \sqrt{s}\right)  \, \left(
\sum_{l=1}^r \mu^2_{\mathbf{N},\mathbf{M}}(k,l) \, s_l\right) , \qquad k=1,\ldots, r,
 }
  and 
$
m_k \gtrsim \hat m_k \,  \log(s\epsilon^{-1}) \, \log\left(q^{-1} \tilde M 
\sqrt{s}\right),
$
where $\hat{m}_k$ is such that
\be{
\label{conditions_on_hatm}
 1 \gtrsim \sum_{k=1}^r \left(\frac{M_k-M_{k-1}}{\hat m_k} - 1\right) \, 
 \mu^2_{\mathbf{N},\mathbf{M}}(k,l)\, \hat s_k, \qquad l=1,\ldots, r,
 }
for all $(\hat{s}_1,\ldots,\hat{s}_r )  \in\bbR_+^r$ with $ \hat s_1+\cdots + \hat s_r \leq s_1+\cdots + s_r$, and $ \hat s_k \leq \hat \kappa_k(\mathbf{N},\mathbf{M},\mathbf{s})$ for each $ k=1,\ldots, r$.

\end{enumerate}
Suppose that $\hat f$ is a minimizer of (\ref{eq:min_orth}) with $y =  P_\Omega V f + \eta$ and $\nm{\eta}_{\ell^2}\leq \delta$.  Then, with probability exceeding $1- \epsilon$, 
\bes{
\|\hat f - f \|\leq  C  \, \left(q^{-1/2}\, \delta \, \left(1+L 
\,\sqrt{s}\right) +\sigma_{\mathbf{s},\mathbf{N}}(D f) \right),
}
for some constant $C$, where $\sigma_{\mathbf{s},\mathbf{M}}$ is as in 
\R{sigma_s_m}, and $
L= C \, \left(1+ 
\frac{\sqrt{\log_2\left(6\epsilon^{-1}\right)}}{\log_2(4q^{-1} M\sqrt{s})}\right).
$
 If $m_k = M_{k}-M_{k-1}$ for $1 \leq k \leq r$ then this holds with 
 probability $1$.
}

Notice that the number of samples at each level is dependent on the local coherences between $V$ and $D$, the level sparsities $\br{s_k}$ and the relative level sparsities $\br{\hat s_k}$.  As discussed in \cite{adcockbreaking}, the relative level sparsities accounts for the interference between the different sampling and sparsifying levels and cannot be removed from the estimates. However, recall that in the case of Fourier sampling with wavelet sparsity where the levels correspond to the wavelet scales, one can essentially show that the dependence of $\hat s_k$ on each $s_j$ becomes exponentially small as $\abs{k-j}$ increases.  

 This result firstly suggests that  even in cases where incoherence is missing, subsampling in accordance to sparsity is still possible provided that the sampling and sparsifying bases are not uniformly coherent -- subsampling is possible when local coherence is small. Note also that this result suggests that  a change in the sparsity structure, i.e. the distribution of $\br{s_k}$ and $\br{\hat s_k}$, should result in a change in the sampling strategy.

\section{Main result}\label{sec:main}

The work of \cite{adcockbreaking} provides an initial understanding on how one can  structure  sampling in accordance to underlying sparsity structures so that the number of  samples require is (up to $\log$ factors) linear with sparsity. A natural extension of this work would be to consider this question when $D$ is an analysis operators associated with a tight frame instead of an orthonormal basis. This is of particular interest due to the recent development of sparse representations with respect to multiscale systems such as wavelet, curvelet and shearlet frames.
  In this paper, we will consider the case where $V :\cH \to \ell^2(\bbN)$ and $D:\cH \to \ell^2(\bbN)$ are isometries. This assumption simply states that $V$ and $D$ are the analysis operators of Parseval frames, i.e. $\br{\psi_j}_{j\in\bbN}$ and  $ \br{\varphi_j}_{j\in\bbN}$ are both Parseval frames of $\cH$ in (\ref{eq:VD_def}). 
  
  Note that if $D$ is associated with an orthonormal basis instead of a Parseval frame (i.e. $D$ is unitary), then (\ref{eq:min_orth}) is equivalent to
  \be{\label{eq:min_ortho_synth} 
\inf_{x\in\ell^1(\bbN)} \nm{  x}_{\ell^1} \text{ subject to } \nm{y -  P_\Omega VD^* x}_{\ell^2}\leq \delta.
}  
This minimization problem is referred to as synthesis regularization. On the other hand, in the case of non-orthonormal systems, (\ref{eq:min_orth}) (often referred to as analysis regularization) and (\ref{eq:min_ortho_synth}) are no longer equivalent. Some of the differences between synthesis and analysis regularization were investigated in \cite{elad2007analysis} and while the majority of theoretical works in compressed sensing has focussed on synthesis regularization, the theory behind the solutions of the analysis regularization problem (\ref{eq:min_orth}) is less comprehensive.

\subsection{Sparsity}  
In this section, we introduce concepts for describing sparsity under an analysis operator. In considering the solutions of (\ref{eq:min_orth}), it is intuitive that this minimization problem will favour signals $f$ for which the entries in $Df$ have fast decay or are mainly zero entries.  Note also that if there exists an index set $\Delta$ such that $P_{\Delta}^\perp D f = 0$, then $f\in\cN(P_{\Delta}^\perp D) \subset \cR(D^*P_\Delta)$ whenever $D^*D= I$. In the works of \cite{candes2011compressed,KrahmerNW15}, the signal space considered is, for each sparsity level $s$, the union of subspaces spanned by $s$ columns of $D^*$, $ \cW=\cup_{\abs{\Delta} = s} \cR(D^*P_\Delta)$.

As discussed in \cite{KrahmerNW15}, to understand the impact of  sparsity on the recovery of such a model, it is natural to  consider the effects of the analysis operator $D$ on any given $f\in \cW$ and in particular, the approximate sparsity of $Df$. For this purpose, \cite{KrahmerNW15} introduced  the localization factor $\eta$, which we previously recalled in (\ref{eq:KNW_factor}),
and their recovery estimates were given in terms of $\eta^2 s$. Moreover, as observed in \cite{candes2011compressed}, a standard measure of sparsity or compression in a vector is the quasi $\ell^p$ norm with $p\leq 1$.  With this in mind, we introduce that concept of \textit{localized sparsity} below.

\begin{definition} \label{def:analysis_sparsity}
Let $r \in\bbN$ and let $\mathbf{N}= (N_j)_{j=1}^r\in\bbN^r$, $\mathbf{s} = (s_j)_{j=1}^r\in\bbN^r$. Assume that $N_1<N_2<\cdots <N_r =:N$. Let $\Lambda_j = \bbN \cap (N_{j-1}, N_j]$ for $j=1,\ldots, r-1$ and $\Lambda_r = \bbN \cap (N_{r-1}, \infty)$. 
Let $p=2^{-J}$ for some $J\in\bbN\cup\br{0}$.  Let   $\kappa>0$ be the smallest number such that
\be{\label{eq:locali_sparsity}
 \kappa^{1-p/q} \geq \sup\br{ \nm{Dg}_p^p: g = D^*x, \nm{Dg}_{\ell^q}=1, x\in\Sigma_{\mathbf{s},\mathbf{N}}}, \qquad  q\in\br{2,\infty},
}
where we let $p/\infty = 0$.
Then, $\kappa(\mathbf{N},\mathbf{s},p) = \kappa$ is said to be the localized sparsity with respect to $p$, $\mathbf{N}$ and $\mathbf{s}$.
For each $j=1,\ldots, r$, let $\kappa_j>0$ be the smallest number such that
\bes{
\kappa_j^{1-p/q} \geq \sup\br{ \nm{P_{\Lambda_j}Dg}_p^p: g = D^*x, \nm{P_{\Lambda_j}Dg}_{\ell^q}=1, x\in\Sigma_{\mathbf{s},\mathbf{N}}}, \qquad  q\in\br{2,\infty},
}
Then $\kappa_j(\mathbf{N},\mathbf{s},p)  = \kappa_j$ is said to be the $j^{\rth}$ localized level sparsity with respect to $p$, $\mathbf{N}$ and $\mathbf{s}$.
\end{definition}
 
\begin{remark}
Observe that the localized sparsity is related to the localization factor in (\ref{eq:KNW_factor}): if $p=1$ and $q=2$ in (\ref{eq:locali_sparsity}), then it suffices to let $\kappa =\eta^2 s$.

One can consider $\kappa(\mathbf{s},\mathbf{N})$ to be a measure of the analysis sparsity of an element $f$ (i.e. sparsity of $Df$) given that it is synthesis sparse with respect to the frame $\br{f_j}$ associated with $D$ (i.e. $f=\sum_{j\in\Delta} x_j \varphi_j$ with $\abs{\Delta} = s$ and $x\in\bbC^\Delta$). 
Note that if $D$ is associated with an orthonormal basis, then $DD^*$ is the identity and it suffices to let $\kappa = s_1+\cdots + s_r$.

 The localized level sparsities $\kappa_j(\mathbf{s},\mathbf{N})$ describe the sparsity structure of $Df$ given that $f$ is synthesis sparse with a $(\mathbf{s}, \mathbf{N})$-sparsity pattern. Again, if $D$ is associated with an orthonormal basis, then these localized level sparsities are simply the level sparsities $\br{s_j}_{j=1}^r$. 
 
\end{remark}

We also require the definition of relative sparsity, note that the only difference to Definition \ref{def:rel_sparsity} is that the set $\Theta$ is defined in terms of $\nm{\cdot}_{\ell^2}$ instead of $\nm{\cdot}_{\ell^\infty}$.

\defn{[Relative sparsity]
\label{def:rel_sparsity2}
Let $V, D\in\cB(\cH, \cH')$ where $\cH$ is a Hilbert space and $\cH'$ is either $\mathbb{C}^{N}$ or $\ell^2(\bbN)$.  Let $\boldsymbol{\kappa} = (\kappa_j)_{j=1}^r \in \bbN^r$,
$\mathbf{N} = (N_j)_{j=1}^r  \in \bbN^r$ and $\mathbf{M} = (M_j)_{j=1}^r \in 
\bbN^r$ with $0=N_0< N_1 < \cdots < N_r$ and $0=M_0< M_1 < \cdots < M_r$. 
 For $1 \leq k \leq r$, the 
$k^{\rth}$ relative sparsity is given by
$$
\hat \kappa_k = \hat \kappa_k(\mathbf{N},\mathbf{M},\boldsymbol{\kappa}) =  \max_{g \in 
\Theta}\|P_{\Gamma_k} V g\|^2,
$$
where $\Gamma_k = (M_{k-1}, M_k]\cap \bbN$ and $\Theta$ is the set

\bes{
\Theta = \{g\in\cH : g = D^*\eta, \, \|P_{\Lambda_l}Dg\|_{\ell^{2}}^2 \leq \kappa_l,\, l=1,\hdots, r\}.
}
where $\Lambda_l = (N_{k-1}, N_k]\cap \bbN$.

}
%

\subsection{Main result}
  The main result of this paper describes how the reconstruction error of any solution of (\ref{eq:min_orth})  depends on the choice of samples. Note that 
the problem of considering the minimizers of (\ref{eq:min_orth}) is well posed since minimizers necessarily exist (see \ref{prop:existence})

In the case of orthonormal systems, the balancing property provides an indication of the range that one should sample from when recovering a sparse support set $\Delta \subset [N]$ for some $N\in\bbN$. This condition essentially describes how large $M$ must be such that $P_{[M]}V$ is close to an isometry  on $\cR(D^*P_\Delta)$ for all $\Delta \subset [N]$. In the case where $D$ is no longer constructed from an orthonormal basis, we define the balancing property as follows. 
  
\begin{definition} \label{balancing_property_new}
Let $V,D \in \mathcal{B}(\cH, \ell^2(\mathbb{N}))$ be  isometries.  Then $M \in \bbN$ 
and $K \geq 1$ satisfy the  balancing property with respect to 
$V$, $D$, $N$, $s\in \bbN$ and $\kappa_2\geq \kappa_1 >0$ if for all $\cW = \cR(D^*P_\Delta)$ where $\Delta\subset [N]$ is such that $\abs{\Delta} = s$,
\begin{equation}\label{condition_bp1}
\begin{split}
\| D Q_\cW V^* P_{[M]}^\perp V Q_\cW D^*  \|_{\ell^{2} \rightarrow \ell^{2}} \leq  
\frac{\sqrt{\kappa_1/\kappa_2}}{8}\left(\log_2^{1/2}\left(4 \sqrt{\kappa_2}KM\right)\right)^{-1},
\end{split}
\end{equation}
and
\begin{equation}\label{condition_bp2}
\begin{split}
\|D Q_\cW^\perp  V^*  P_{[M]} V Q_\cW D^*  \|_{\ell^{2} \rightarrow \ell^{\infty}}
\leq \frac{1}{8\sqrt{\kappa_2}},
\end{split}
\end{equation}

\end{definition}

Although this balancing property conceptually enforces the same isometry properties as the balancing property presented in the case of orthonormal systems, note that the conditions are stated in terms of the $\ell^2$ norm instead. This difference is due to a slightly  different \textit{dual certificate} construction in the proof of our main result, and this slightly stronger balancing property will allow us to derive sharper bounds on the number of samples required. We remark also that in the case where $\kappa_1=\kappa_2$, this balancing property in fact reduces to the original balancing property introduced in \cite{BAACHGSCS}. 

In the following theorem, for $r\in\bbN$, let   $\mathbf{M} = (M_1,\ldots,M_r) \in \bbN^r$ , $\mathbf{m} = (m_1,\ldots,m_r) \in \bbN^r$, $\mathbf{N} = (N_1,\ldots,N_r) \in \bbN^r$, 
$N_1 < \ldots < N_r$, and $\mathbf{s} = (s_1,\ldots,s_r) \in \bbN^r$. For  $p\in (0,1]$, let $\boldsymbol{\kappa} = (\kappa_j)_{j=1}^r$ with $\kappa_j = \kappa_j(\mathbf{N},\mathbf{s},p)$ and let $\hat \kappa_j = \hat \kappa_j(\mathbf{N},\mathbf{M},\boldsymbol\kappa)$.
Let
\be{\label{eq:tilde_M}
\tilde M = \nm{DD^*}_{\ell^\infty \to \ell^\infty} \min\br{i\in\bbN : 
\max_{j\geq i} \nm{ P_{[M]}  V  D^* e_j}_{\ell^2} \leq \frac{q}{8\sqrt{\kappa_{\max}}}, \quad \max_{j\geq i} \nm{Q_{\cR(D^*P_{[N]})} D^* e_j}\leq \frac{\sqrt{5q}}{4}},
}
and
 \be{\label{eq:B}
  B(\mathbf{s},\mathbf{N}) = \sup\br{ \tilde B(\Delta) :  \Delta \text{ is } (\mathbf{s}, \mathbf{N})\text{-sparse}},
}
where, given any $\Delta\subset\bbN$ and $\cW_\Delta = \cR(D^*P_\Delta)$, 
\be{\label{eq:B_Delta}
\tilde B(\Delta) = \max\br{\nm{D Q_{\cW_\Delta}^\perp D^*}_{\ell^\infty \to \ell^\infty} , \sqrt{\nm{D Q_{\cW_\Delta} D^*}_{\ell^\infty \to \ell^\infty}\cdot \max_{l=1}^r \br{ \sum_{t=1}^r \nm{P_{\Lambda_l}D Q_{\cW_\Delta} D^* P_{\Lambda_t}}_{\ell^\infty \to \ell^\infty}}}}.
}
The key notations used in Theorem \ref{thm:main} are summarized in Table \ref{t:notation}.

 \begin{table} 
\begin{center}
{\small{
\begin{tabular}{| c| c|}

\hline 
 Notation & Description \\
 \hline
 $V$ & Sampling operator\\
 $D$ & Sparsifying operator\\
 $r$ & Number of levels\\
 $\mathbf{N} = (N_k)_{k=1}^r$ & Divides the sparsifying coefficients into levels \\
  $\mathbf{M} = (M_k)_{k=1}^r$ & Divides the sampling coefficients into levels \\
  $\mathbf{m} = (m_k)_{k=1}^r$ & Number of samples at each level\\
 $\mathbf{s} = (s_k)_{k=1}^r$ & Level sparsities\\
 $\mu_{\mathbf{N},\mathbf{M}}(k,l)$ & $(k,l)^{\rth}$ localized coherence, see Definition \ref{def:loc_coherence}\\
 $\sigma_{\mathbf{s},\mathbf{N}}$ & See Definition \ref{def:sparsity_level}\\
  $\kappa_j $ &$j^\rth$ localized sparsity, see Definition \ref{def:analysis_sparsity}
 \\
 $\hat \kappa_j$ & $j^{\rth}$ relative sparsity, see Definition \ref{def:rel_sparsity2}\\
 $\tilde M$ & See (\ref{eq:tilde_M})\\
 $B(\mathbf{s},\mathbf{N})$ & See (\ref{eq:B}) 
 \\\hline
\end{tabular}
   }}
  \end{center}
   \caption{Summary of the key notations for Theorem \ref{thm:main}.}
\label{t:notation}
 \end{table}
 
\thm{ 
\label{thm:main}
Let $\cH$ be a Hilbert space and let $V,D\in \cB(\cH , \ell^2(\bbN))$ be isometric linear operators. Let $f \in 
\cH$.  Suppose that $\Omega = \Omega_{\mathbf{M},\mathbf{m}}$ is a 
multilevel sampling scheme.  Let 
$(\mathbf{s},\mathbf{N})$ be such that the following holds: 
\begin{enumerate}

\item[(i)] the parameters
$
M=M_r, q^{-1} = \max_{k=1,\ldots,r} \left \{ \frac{M_{k}-M_{k-1}}{m_k} 
\right \},
$
satisfy the balancing property with respect to $V$, $D$, $N:= N_r$, $\kappa_{\min}= r\min\br{\kappa_j}$ and $\kappa_{\max} = r\max\br{\kappa_j}$;
\item[(ii)] 
For $\epsilon \in (0,e^{-1}]$,
\bes{
1 \gtrsim   \sqrt{r}\, \log(\epsilon^{-1}) \log\left(q^{-1} \tilde M \sqrt{\kappa_{\max}}\right) \, B(\mathbf{s},\mathbf{N})\, \frac{M_k-M_{k-1}}{m_k}  \, \left(
\sum_{l=1}^r \mu^2_{\mathbf{N},\mathbf{M}}(k,l) \, \kappa_l\right), \qquad k=1,\ldots, r,
 }
   and 
$
m_k \gtrsim r\, \hat m_k \, B(\mathbf{s},\mathbf{N})^2 \, \log(\epsilon^{-1}) \, \log\left(q^{-1} \tilde M 
\sqrt{\kappa_{\max}}\right),
$
where $\hat{m}_k$ is such that
\bes{
 1 \gtrsim  \sum_{k=1}^r \left(\frac{M_k-M_{k-1}}{\hat m_k} - 1\right) \, 
 \mu^2_{\mathbf{N},\mathbf{M}}(k,l)\,  \hat \kappa_k,  \qquad l=1,\ldots, r.
 }
 \end{enumerate}
Suppose that $\hat f$ is a minimizer of (\ref{eq:min_orth}) with $y =  P_\Omega V f + \eta$ and $\nm{\eta}_{\ell^2}\leq \delta$.  Then, with probability exceeding $1- \epsilon $, 
\bes{
\|\hat f - f \|\leq  C  \, \left(q^{-1/2}\, \delta \, \left(1+L 
\,\sqrt{\kappa_{\max}}\right) +\sigma_{\mathbf{s},\mathbf{N}}(D f) \right),
}
for some constant $C$, where $\sigma_{\mathbf{s},\mathbf{N}}$ is as in 
\R{sigma_s_m}, and $
L=   1+ 
\frac{\sqrt{\log_2\left(6\epsilon^{-1}\right)}}{\log_2(4q^{-1} M\sqrt{\kappa_{\max}})}.
$
 If $m_k = M_{k}-M_{k-1}$ for $1 \leq k \leq r$ then this holds with 
 probability $1$.
}

\subsubsection{The unconstrained minimization problem}
Instead of solving the constrained minimizaton problem in Theorem \ref{thm:main}, for computational reasons, it is often of interest to solve instead an unconstrained minimization problem for some $\alpha >0$,
\be{\label{eq:unconstr}
\inf_{g\in\cH} \alpha \nm{ D g}_{\ell^1} + \nm{ P_\Omega  V g - y}_{\ell^2}^2.
}
The following result presents a recovery guarantee for this unconstrained problem.
\begin{corollary}\label{cor:unconstr}
Consider the setting of Theorem \ref{thm:main}, and let $\alpha = \sqrt{q}\delta$. Then, with probability exceeding $(1-\epsilon)$, any minimizer $\hat f$ of (\ref{eq:unconstr}) satisfies
$$
\nm{f- \hat f}_\cH \leq   \delta\left( q^{-1/2} + L\sqrt{\kappa_{\max}} + L^2 \kappa_{\max}\right) +  \sigma_{\mathbf{N}, \mathbf{s}}( D f).
$$

\end{corollary}
\begin{remark}
Note that by choosing $\alpha = \sqrt{q} \delta$, the guaranteed error bound is, up to $\sqrt{q}L^2 s$, the same as the guaranteed error bound of solutions to the constrained problem
$$
\inf_{g\in\cH}  \nm{ D g}_{\ell^1} \text{ subject to } \nm{ P_\Omega  V g - y}_{\ell^2}\leq \delta.
$$
This affirms the finding in \cite[Figure 7]{benning2014phase}, which numerically demonstrates that there exists a linear relation between the regularization parameter $\alpha$ and noise level of the measurements $\delta$. Moreover, this linear scaling increases as $q$ increases.
\end{remark}

\subsection{Remarks on the main result}

\subsubsection{On the factor $r$}
In the case where $D$ is associated with an orthonormal basis, the key difference between our main result and Theorem \ref{thm:main_bc} is that bounds on the number of samples in Theorem \ref{thm:main_bc} has a factor of $\log(s)$ while the bounds in Theorem \ref{thm:main} have a factor of $r$ (the number of levels) instead. In general, the sparsity $s$ may grow as the ambient dimension $N_r$ increases, whilst the number of levels $r$ can be thought of as simply a constant; for example, $r=2$ in the case of the half-half schemes presented in \S \ref{sec:discreteHaar} (see also \cite{studer2012compressive} for the application of a half-half scheme in fluorescence microscopy). Therefore, Theorem \ref{thm:main} may be considered to provide slightly sharper bounds than Theorem \ref{thm:main_bc} and is in fact optimal in the case where $r=1$ (since the optimal sampling cardinality is $\ord{s\log N}$ \cite{candes2011probabilistic}). We remark however, that by utilizing the construction of the dual certificate from \cite{adcockbreaking}, one can replace the factor of $r$ with $\log(\tilde M)$.

\subsubsection{Remarks on reconstructing -lets from Fourier samples}

We begin with a corollary of Theorem \ref{thm:main}. Its proof can be found in Section \ref{sec:loc_sp}.
\begin{corollary}\label{cor:rel_sp}
Let $V$ and $D$ be isometries. Let $\br{N_j}_{j=1}^r$,$\br{M_j}_{j=1}^r\in\bbN^r$ with $\Gamma_j = (M_{j-1},\ldots, M_j]\cap \bbN$ and $\Lambda_j =  (N_{j-1},\ldots, N_j]\cap \bbN$ where $M_0=N_0=0$. 
Let $\omega$ be a non-negative function defined on $\br{1,\ldots, r}^2$ with
\be{\label{eq:maincor1}
\sum_{l=1}^r \omega(k,l)\leq C, \quad k=1,\ldots, r, \qquad \sum_{k=1}^r \omega(k,l)\leq C, \quad l=1,\ldots, r,
}
for some $C>0$.
Suppose that
\be{\label{eq:maincor2}
 \nm{P_{\Gamma_k}VD^*P_{\Lambda_l}}\leq \omega(k,l), \qquad \mu^2_{\mathbf{N},\mathbf{M}}(k,l) \leq \omega(k,l) \cdot \min\br{\frac{1}{N_{l-1}}, \frac{1}{M_{k-1}}}.
 }
 Then  condition (ii) of Theorem \ref{thm:main} holds provided that
$$
m_k \gtrsim r C^2 B(\mathbf{s},\mathbf{N})^2 \cdot \log(\epsilon^{-1}) \log(q^{-1} \tilde M \sqrt{\kappa_{\max}}) \cdot \frac{M_k - M_{k-1}}{M_{k-1}}\cdot \left( \sum_{l=1}^r \omega(k,l) \kappa_l \right).
$$
In particular,   we have that
$$
m_1+\cdots +  m_r \gtrsim \tilde C \cdot \log(\epsilon^{-1}) \log(q^{-1} \tilde M \sqrt{\kappa_{\max}}) \cdot \left( \kappa_1 + \cdots + \kappa_r \right),
$$
where $\tilde C =r C^3 B(\mathbf{s},\mathbf{N})^2 \max_{k=1}^r\br{(M_k - M_{k-1})/M_{k-1}}$.
\end{corollary}

Note that the dependence of our main result on the localized coherence terms allows one to exploit both asymptotic incoherence \textit{and} the correspondences between the different sparsifying and sampling levels.
Conditions (\ref{eq:maincor1}) and (\ref{eq:maincor2}) essentially control the correspondence between the different sampling and sparsifying levels, whilst maintaining asymptotic incoherence in $VD^*$. Under these conditions, this result presents a direct link between the localized sparsities and the sampling strategy where the dependence of the number of samples in the $k^\rth$ level $m_k$ on the $l^\rth$ localized sparsity $\kappa_l$ is weighted by $\omega(k,l)$. Note also that the only other dimensional dependencies consist of one $\log$ factor and the factor of $B(\mathbf{s},\mathbf{N})$, which numerically does not seem to be significant (see Section \ref{sec:BsN}). 

Of course, further analysis of Corollary \ref{cor:rel_sp} would be necessary for a full comparison between our results and Theorem \ref{thm:knw}, however,  an advantage of Corollary \ref{cor:rel_sp} is that it makes explicit the dependence between how one should subsample and the sparsity structure, and provided that $B(\mathbf{s},\mathbf{N})$ remains bounded, Corollary \ref{cor:rel_sp} will provide for a sharper estimate on the sampling cardinality. In the case where $V$ is constructed from a Fourier basis and $D$ is constructed from a wavelet basis, it is in fact the analysis of (\ref{eq:maincor1}) and (\ref{eq:maincor2}) that enabled \cite{adcockbreaking} to derive sharp bounds on the sampling cardinality. We now explain this in more detail: 

\cite{adcockbreaking} considered the case where
$$
D:L^2[0,1)\to \ell^2(\bbN), \qquad f \mapsto (\ip{f}{\varphi_j})_{j\in\bbN}
$$
for some orthonormal wavelet basis $\br{\varphi_j}$ associated with a scaling function $\Phi$ and a mother wavelet $\Psi$ satisfying for all $\xi\in\bbR$,
\be{\label{eq:Fourier_decay}
\abs{\hat \Phi \left(\xi\right)} \lesssim (1+\abs{\xi})^{-\beta}, \qquad \abs{\hat \Psi \left(\xi\right)} \lesssim (1+\abs{\xi})^{-\beta}.
}
where $\hat \Phi $ and $\hat \Psi$ denote the  Fourier transforms of $\Phi$ and $\Psi$ respectively, and
the Fourier sampling operator is
$$
V:L^2[0,1)\to \ell^2(\bbN), \qquad f \mapsto (\ip{f}{e^{i2\pi \omega k \cdot }})_{k\in\bbZ}
$$
for some appropriate Fourier sampling density $\omega\in (0,1]$. Then,  if  we  let $(M_k)$ and $(N_k)$ correspond to wavelet scales, so that $M_k = \ord{2^{R_k}}$, $N_k = 2^{R_k}$ and $(R_k)_{k=1}^r\in\bbN^r$ is an increasing sequence of integers,  then we can let \be{\label{eq:omega_fw}
\omega(k,j) = \min\br{ A^{R_j - R_{k-1}}, B^{R_k - R_{j-1}}}} where $A, B>1$ are constants which depend on the Fourier decay exponent $\beta$ and the number of vanishing moments of the generating wavelet. Furthermore, since $D$ is constructed from an orthonormal basis, $\kappa_j = s_j$ for each $j=1,\ldots, r$, and $B(\mathbf{s},\mathbf{N}) = 1$.
 \cite{adcockbreaking} also analysed the balancing property in the Fourier/wavelets case and condition (i) can be shown to hold provided that 
$$
 M_r \gtrsim \omega^{-1} N_r^{1+1/(2\beta-1)}\log_2(\epsilon^{-1} N_r \sqrt{s} q^{-1})^{1/(2\beta -1)},
$$  
and $\log(\tilde M)\leq \log(M_r)$. So the number of samples needed on the $k^{\rth}$ level is
$$
m_k \gtrsim \cL\cdot \frac{M_k-M_{k-1}}{M_{k-1}}\cdot\left( s_k  + \sum_{j\leq k-1} A^{R_j - R_{k-1}} s_j + \sum_{j\geq k+1} B^{R_k - R_{j-1}} s_j\right),
$$
with $\cL =   r \log(M)$. Note that the total sampling cardinality is, up to one $\log$ factor and the ratio $\max_{k=1}^r  \br{ \frac{M_k-M_{k-1}}{M_{k-1}}}$, linear with the total sparsity.

It is likely that by carrying out a similar analysis in \cite{adcockbreaking}, one can apply Theorem \ref{thm:main} to derive sharp recovery results for the recovery of coefficients with other multiscale systems, such as shearlets and wavelet frames from Fourier samples. This work is beyond the scope of this paper, however, we simply highlight two aspects of any such analysis. The first  part of the analysis would include precise estimates on the correspondences between the different sampling and sparsifying levels (i.e. analysis of $\omega$ in Corollary \ref{cor:rel_sp}). In the case of orthonormal Fourier and wavelet bases, the choice of $\omega$ in (\ref{eq:omega_fw}) is simply due to the Fourier decay (\ref{eq:Fourier_decay}) and the number of vanishing moments in the underlying wavelet and not on orthogonality properties. Thus, such a choice of $\omega$ would also suffice in the case of wavelet frames with Fourier decay and vanishing moments properties. Furthermore, since similar Fourier decay estimates and vanishing moments properties also exist for multiscale systems such as curvelets and shearlets, it would be possible to derive similar estimates in the case of other multiscale systems. Secondly, the key difference between Theorem \ref{thm:main_bc} and Theorem \ref{thm:main} is the localized sparsity $\kappa(\mathbf{s},\mathbf{N})$ and the localized level sparsities with respect to $\rD$.  As mentioned, these terms are equal to the sparsity and level sparsity terms when $DD^*$ is the identity, furthermore, it is known that multiscale systems such as wavelet frames, shearlets and curvelets are \textit{intrinsically localized} with near-diagonal Gram matrices. It is therefore conceivable that this property can be exploited to show that localized sparsity $\kappa$ is close to the true sparsity $s$. This idea is further discussed in \S \ref{sec:loc}.

\subsubsection{On $B(\mathbf{s},\mathbf{N})$}\label{sec:BsN}
In the case where $D$ is associated with an orthonormal basis, $B(\mathbf{s},\mathbf{N})=1$. Further analysis of this quantity will be left as future work, however, we simply remark here that initial computations of this quantity suggest that the impact of $B(\mathbf{s},\mathbf{N})$ will not be significant: To test the behaviour of $B(\mathbf{s},\mathbf{N})$, consider the following experiment where we test the behaviour of this quantity when considering the support of piecewise constant vectors, under the redundant Haar transform $D$. Given $p\in\bbN$, let $N=2^p$,  $D$ be the discrete Haar wavelet frame transform and compute
\be{\label{eq:E}
E(p) = \max\br{\tilde B(\Delta): \Delta = \mathrm{supp}(Dx), x\in \Lambda_N}
} where  $\tilde B$ is as defined in (\ref{eq:B_Delta}),
and $\Lambda_N$ is a collection of 1000 randomly generated piecewise constant vectors, each of length $N$.
A plot of $E$ for $p=4,\ldots, 10$ is shown in Figure \ref{fig:Eminmax}.

\begin{figure}
\begin{center}
\includegraphics[width=0.8\textwidth]{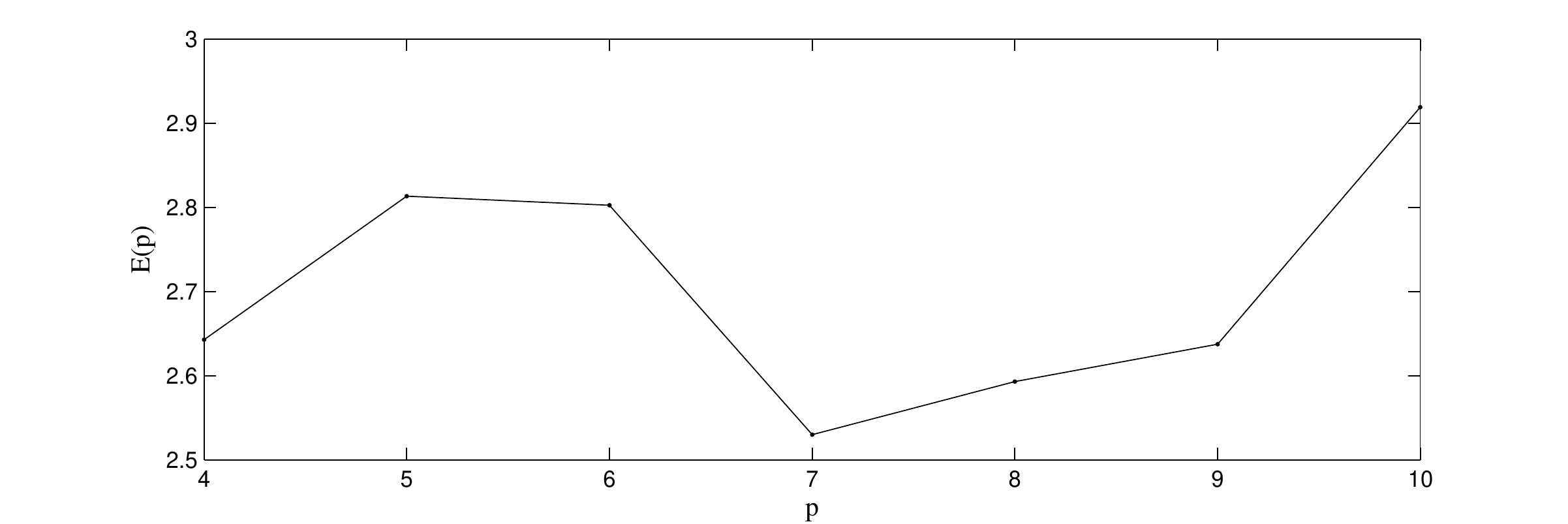}
\end{center}
\caption{Plot of $E$ from (\ref{eq:E}). \label{fig:Eminmax}}
\end{figure}

\section{Localized sparsity}\label{sec:loc_sp}

In this section, we present some basic properties of the localized sparsity defined in Definition \ref{def:analysis_sparsity}. 
The key findings which would be of use in the proof of Theorem \ref{thm:main} are summarized in Corollary \ref{cor:loc_sp}.

We first present Lemma \ref{lem:rel_sp} to show that provided that $VD^*$ satisfies some ``block diagonal" structure (so that $\omega(j,k)$ in Lemma \ref{lem:rel_sp} decays sufficiently as $\abs{j-k}$ increases), each relative sparsity term $\hat \kappa_j$ can be controlled in terms of $\br{\kappa_l}_{l=1}^r$ and the dependence on each $\kappa_l$ decays as $\abs{j-l}$ increases. So  Lemma \ref{lem:rel_sp} can then be applied to derive Corollary \ref{cor:rel_sp}, which shows that under an additional assumption on the structure of $VD^*$, the signal dependencies of Theorem \ref{thm:main} arise only in the localized level sparsities $\boldsymbol{\kappa}=\br{\kappa_j}_j$. 
Note that this block diagonal property can be shown to exist when $V$ is a Fourier sampling transform and $D$ is a wavelet transform \cite{adcockbreaking}.

\begin{lemma}\label{lem:rel_sp}
Let $V$ and $D$ be isometries. Let $\br{N_j}_{j=1}^r$,$\br{M_j}_{j=1}^r\in\bbN^r$ with $\Gamma_j = (M_{j-1},\ldots, M_j]\cap \bbN$ and $\Lambda_j =  (N_{j-1},\ldots, N_j]\cap \bbN$ where $M_0=N_0=0$. Suppose that
$ \nm{P_{\Gamma_k}VD^*P_{\Lambda_l}}\leq \omega(k,l)$ and
$$
\sum_{l=1}^r \omega(k,l)\leq C, \qquad k=1,\ldots, r.
$$ Then 
$$
\hat \kappa_j \leq C \sum_{l=1}^r \omega(j,l) \kappa_l.
$$
\end{lemma}

\begin{proof}
Since $D^*D = I$,
\eas{
\hat\kappa_j & = \max_{g\in\Theta} \nm{P_{\Gamma_j}VD^*D g}^2 \leq  \max_{g\in\Theta}\left(\sum_{l=1}^r \nm{P_{\Gamma_j}VD^*P_{\Lambda_l}D g}\right)^2\\
&\leq \max_{g\in\Theta}\left(\sum_{l=1}^r \nm{P_{\Gamma_j}VD^*P_{\Lambda_l}}\nm{D g}\right)^2
\leq \left(\sum_{l=1}^r \omega(j,l) \sqrt{\kappa_l} \right)^2\\
&\leq \left(\sum_{l=1}^r\omega(j,l)\right) \left(\sum_{l=1}^r \omega(j,l) \kappa_l \right)
\leq C \sum_{l=1}^r \omega(j,l) \kappa_l.
}
 
\end{proof}
\begin{proof}[Proof of Corollary \ref{cor:rel_sp}]
Since $\mu^2_{\mathbf{N},\mathbf{M}}(k,l) \leq \omega(k,l) \min\br{1/N_{l-1}, 1/M_{k-1}}$, condition (ii) of Theorem \ref{thm:main} is implied by
\be{\label{cor_pf:eq2}
m_k\gtrsim \cL \cdot \sqrt{\tilde B} \cdot \frac{M_k-M_{k-1}}{M_{k-1}} \cdot \left(\sum_{l=1}^r \omega(k,l)\cdot \kappa_l \right), \qquad k=1,\ldots, r,
}
and
\be{\label{cor_pf:eq1}
1\gtrsim \cL \cdot\tilde B\cdot  \sum_{k=1}^r  \frac{M_k-M_{k-1}}{m_k} \cdot \frac{\omega(k,l)}{M_{k-1}} \cdot \hat \kappa_k, \qquad l=1,\ldots, r,
}
where $\tilde B =r B(\mathbf{s},\mathbf{N})^2$ and $\cL = \log(\epsilon^{-1}) \log(q^{-1} \tilde M \sqrt{\kappa_{\max}}) $.

Since $\sum_{k=1}^r \omega(k,l) \leq C$ for each $l=1,\ldots, r$, (\ref{cor_pf:eq1}) is true provided that
$$
m_k \gtrsim \frac{C \tilde B \cL (M_k-M_{k-1})}{M_{k-1}} \hat \kappa_k, \qquad k=1,\ldots, r.
$$
By Lemma \ref{lem:rel_sp}, this is true provided that
\be{\label{cor_pf:eq3}
m_k \gtrsim \frac{C^2 \tilde B \cL (M_k-M_{k-1})}{M_{k-1}}  \kappa_k, \qquad k=1,\ldots, r.
}
 Note that (\ref{cor_pf:eq3}) also implies (\ref{cor_pf:eq2}).

Finally, the last statement of Lemma \ref{lem:rel_sp} follows by summing up the $m_k$'s and using $\sum_{k}\omega(k,l) \leq C$ for each $l=1,\ldots,r$.
\end{proof}

\begin{lemma}\label{lem:loc_sp_eq_sp}
Let $p\leq 1$.
Suppose that $x$ has at most $s$ non-zero entries and $\nm{x}_{\ell^q} = 1$ for some $q\geq 1$. Then, 
$\nm{x}^p_{\ell^p} \leq s^{1-p/q}$.
\end{lemma}
\prf{
Let $\Delta$ denote the support of $x$. By H\"{o}lder's inequality,
$$
\sum_{j\in\Delta}\abs{x_j}^{p} \leq \left(\sum_{j\in\Delta}\abs{x_j}^{q}\right)^{p/q} \left(\abs{\Delta}\right)^{1-p/q} = s^{1-p/q}.
$$
}

\begin{lemma}\label{lem:loc_sp_2}
Let $p\in (0,1]$ and let $q\in [1,\infty]$.
Suppose $\kappa>0$ is such that $\nm{x}_{\ell^q}\leq 1$ implies that  $\nm{x}_{\ell^p}^p \leq \kappa^{1-p/q}$.
\begin{itemize}
\item[(i)] $\nm{x}_{\ell^1}\leq \kappa\nm{x}_{\ell^\infty}$.
\item[(ii)] $\nm{x}_{\ell^2}^2\leq \kappa\nm{x}_{\ell^\infty}^2$.
\item[(iii)]  If $p=2^{-L}$ for some $L\in\bbN\cup\br{0}$, then $\nm{x}_{\ell^1}^2 \leq \kappa\nm{x}_{\ell^2}^2$.
\end{itemize}

\end{lemma}
\prf{
Without loss of generality, first assume that $\nm{x}_{\ell^\infty}=1$. Then, $\nm{x}_{\ell^p}^p\leq \kappa$. So, (i) follows because
$$
\sum_{j}\abs{x_j} = \sum_{j}\abs{x_j}^p \abs{x_j}^{1-p} \leq \sum_{j}\abs{x_j}^p \leq \kappa,
$$
and (ii) follows because 
$$
\sum_{j}\abs{x_j}^2 = \sum_{j}\abs{x_j}^p \abs{x_j}^{2-p} \leq \sum_{j}\abs{x_j}^p \leq \kappa.
$$
To show (iii), assume (without loss of generality) that $\nm{x}_{\ell^2}=1$ and recall that $p=2^{-L}$. Then, by repeatedly applying the Cauchy-Schwarz inequality,
\eas{
&\nm{x}_{\ell^1}^2  = \left(\sum_j \abs{x_j}\right)^2 \leq \sum_j \abs{x_j}^p \sum_j \abs{x_j}^{2(1-p/2)}
 \leq \kappa^{1-p/2} \nm{x}_{\ell^2} \left(\sum_j \abs{x_j}^{2-2p}\right)^{1/2}\\
&\leq \kappa^{1-p/2} \nm{x}_{\ell^2}^{1+1/2}  \left(\sum_j \abs{x_j}^{2-4p}\right)^{1/4}  
\leq \kappa^{1-p/2} \nm{x}_{\ell^2}^{1+1/2}  \left(\sum_j \abs{x_j}^{2-2^{L}p}\right)^{1/2^L}
= \kappa^{1-p/2} \nm{x}_{\ell^1}^{p},
}
and by dividing both sides of the inequality by $\nm{x}_{\ell^1}^{p}$, we obtain
$$
\nm{x}_{\ell^1}^{2(1-p/2)} \leq \kappa^{1-p/2} \implies \nm{x}_{\ell^1}^2\leq \kappa. 
$$
}

\begin{corollary}\label{cor:loc_sp}
In the notation of Definition \ref{def:analysis_sparsity}, a direct application of the two lemmas presented above (with $x:=P_{\Lambda_l}DD^*y$, $l=1,\ldots, r$) would yield the following results:
\begin{enumerate}
\item 
 Lemma \ref{lem:loc_sp_eq_sp}: if $D$ was the analysis operator of an orthonormal basis (so that $DD^*=I$) in Definition \ref{def:analysis_sparsity}, then $\kappa(\mathbf{s},\mathbf{N}) = s_1+\cdots + s_r$.
\item Lemma \ref{lem:loc_sp_2} (i) : If $y\in\Sigma_{\mathbf{s},\mathbf{N}}$ and $\nm{P_{\Lambda_j} DD^* y}_{\ell^\infty}=1$, then $$  \nm{P_{\Lambda_j} D D^* y}_{\ell^1}\leq \kappa_j(p,\mathbf{N},\mathbf{s}), \quad j=1,\ldots, r.$$
\item Lemma \ref{lem:loc_sp_2} (ii) : If $y\in\Sigma_{\mathbf{s},\mathbf{N}}$ and $\nm{P_{\Lambda_j} DD^* y}_{\ell^\infty}=1$, then $$  \nm{P_{\Lambda_j} D D^* y}_{\ell^2}^2\leq \kappa_j(\mathbf{N},\mathbf{s},p), \quad j=1,\ldots, r.
$$

\item Lemma \ref{lem:loc_sp_2} (ii) : If $y\in\Sigma_{\mathbf{s},\mathbf{N}}$ and $\nm{P_{\Lambda_j} DD^* y}_{\ell^2}=1$, then $$ \nm{P_{\Lambda_j} D D^*y}_{\ell^1}^2\leq \kappa_j(\mathbf{N},\mathbf{s},p), \quad j=1,\ldots, r.
$$
\end{enumerate}
\end{corollary}

\subsubsection*{Numerical example: the Haar frame}
 In the case where $D$ is associated with a Haar frame on $\bbC^N$, it can be shown that if $\abs{\mathrm{supp}(x)}=s$ then $DD^* x$ has at most $\ord{s\log N}$ nonzero entries (see \cite{KrahmerNW15}). Therefore, from Lemma \ref{lem:loc_sp_eq_sp}, $\kappa(\mathbf{N},\mathbf{s}) \lesssim s\log(N)$ where $\mathbf{s} = (s_j)_{j=1}^r$ and  $s=s_1+\ldots+s_r$. In the case of a Haar frame, experimental results suggest that the localized level sparsities  $\br{\kappa_j}_j$ tend to follow a similar pattern to the level sparsities $\br{s_j}$:  Let $D$ be the discrete Haar frame, and let $f\in\bbR^{1024}$ be as shown on the right of Figure \ref{fig:approx_kappa}. Let $\Delta$ be the support of $D f$. Let $\cS$ consist of 1000 randomly generated vectors, each supported on $\Delta$. For each $j=0,\ldots, 10$, let $\Lambda_j$ index all Haar framelet coefficients in the $j^{\rth}$ scale and let \be{\label{eq:approx_kappa}
\tilde \kappa_j = \sup \br{  \nm{P_{\Lambda_j} \eta_\infty}_{\ell^1}, \nm{P_{\Lambda_j} \eta_2}_{\ell^1}^2:   \eta_\infty = DD^* x/\nm{DD^* x}_{\ell^\infty}, \eta_2 = DD^* x/\nm{DD^* x}_{\ell^\infty}, x\in\cS}.
}
We also let $s_j = \abs{\Delta\cap \Lambda_j}$. The bar charts in Figure \ref{fig:approx_kappa} show for each $j=0,\ldots, 10$, $\br{s_j/\abs{\Lambda_j}}$ (centre plot) and $\br{\tilde \kappa_j/\abs{\Lambda_j}}$ (left plot).
Note that $\br{\tilde \kappa_j}$ merely approximate the localized level sparsities $\br{\kappa_j}$, because otherwise, we would need to consider all $(\mathbf{s},\mathbf{N})$-sparse support sets instead of just one support set $\Delta$ and we would also need to maximize over all vectors supported on $\Delta$, instead of just 1000 randomly generated vectors. Nonetheless, Figure \ref{fig:approx_kappa} provides some indication of the behaviour of the localized level sparsities.
\begin{figure}
\begin{center}
\begin{tabular}{@{\hspace{0pt}}c@{\hspace{0pt}}c@{\hspace{0pt}}c@{\hspace{0pt}}}
$f$ & Sparsity in levels& Localized sparsity in levels \\
\includegraphics[width = 0.33\textwidth]{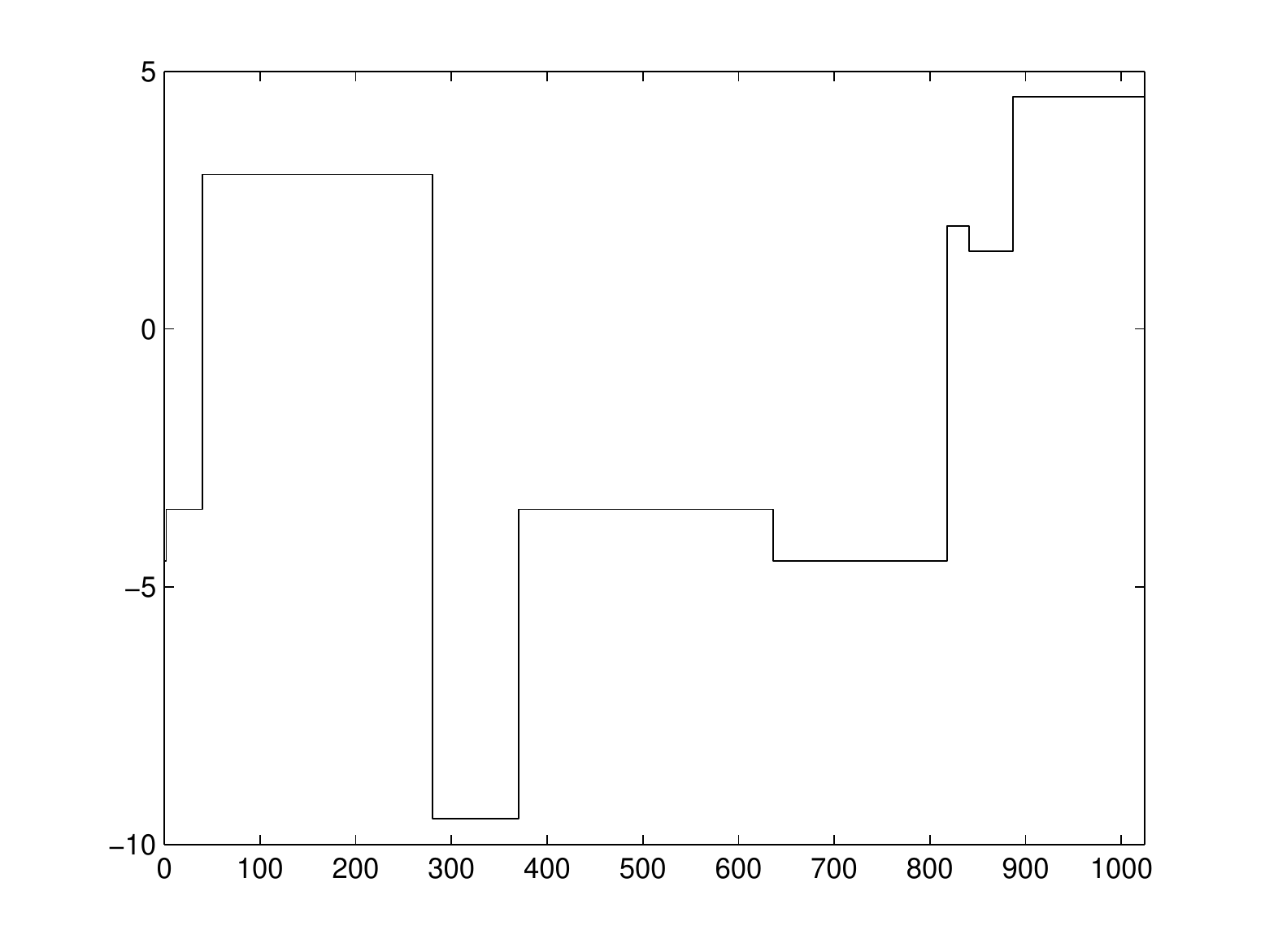}
& \includegraphics[width = 0.33\textwidth]{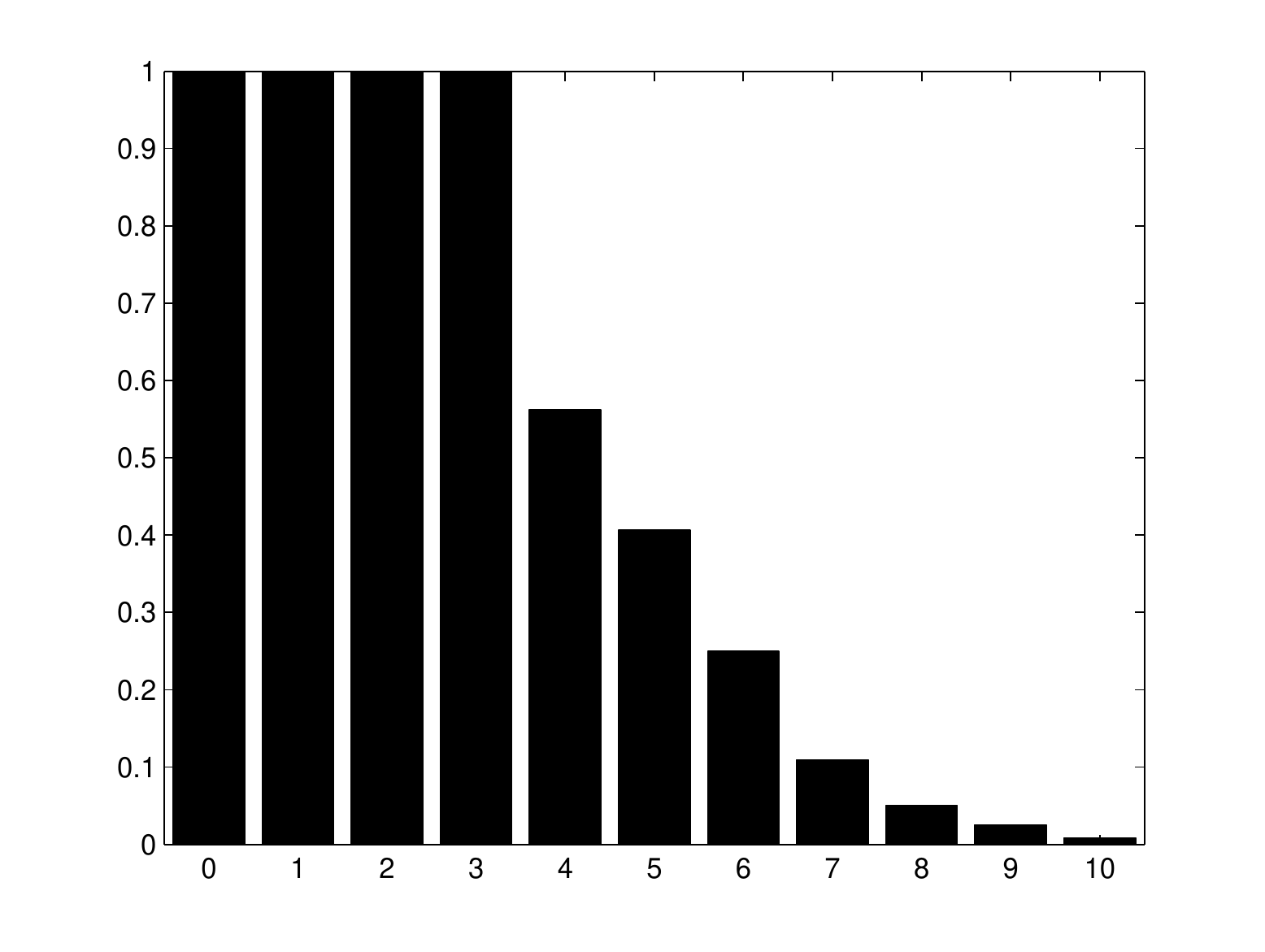}  &
\includegraphics[width = 0.33\textwidth]{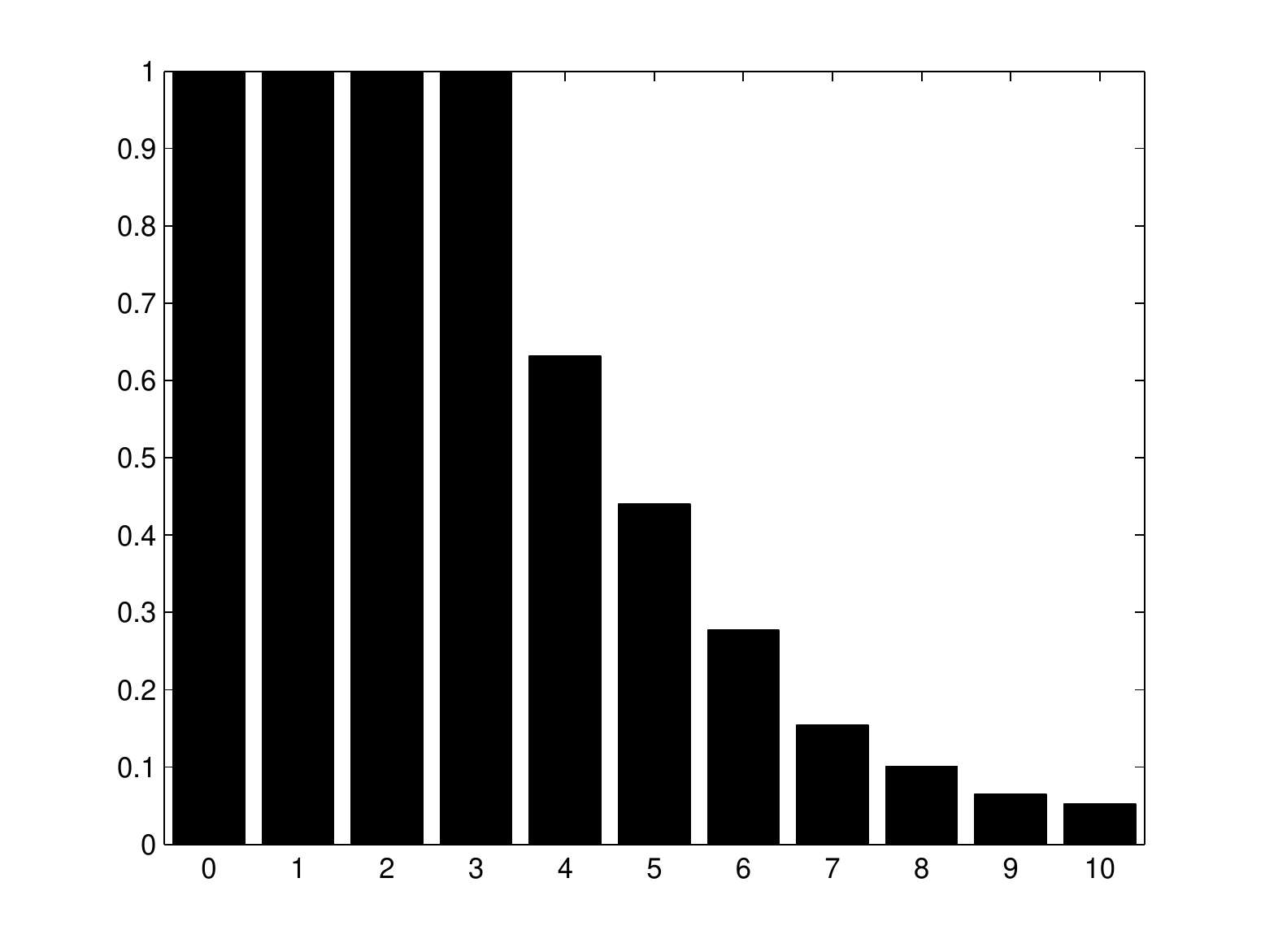} 
\end{tabular}
\end{center}
\caption{Left: the original vector. Centre: the level sparsities of $Df$ at each scale. Right: the approximate localized level sparsities at each scale, as defined in (\ref{eq:approx_kappa}). \label{fig:approx_kappa}
}
\end{figure}

\subsection{Intrinsic localization}\label{sec:loc}

As mentioned previously, many of the popular frames such as curvelets, shearlets and wavelet frames are intrinsically localized so that their Gram matrices are near diagonal. This property has been studied for wavelet frames in  \cite{jaffard1990proprietes,cordero2004localization,fornasier2005intrinsic} and more recently for anisotropic systems such as shearlets and curvelets in \cite{grohs2013intrinsic,grohs2014intrinsic}.
In this section, we will show how the property of intrinsic localization can yield estimates on the localized sparsity term, $\kappa$,  and the localized level sparsity terms, $\kappa_j$'s. We first recall the notion of intrinsic localization \cite{grochenig2004localization,grochenig2003localized}.

\begin{definition}
Let $\cH$ be a Hilbert space and let $\Psi =\br{\psi_j}_{j\in\bbN}$ be a frame for $\cH$. $\Psi$ is said to be intrinsically localized with respect to $c>0$ and $L \geq 1$ if
$$
\abs{\ip{\psi_j}{\psi_k}} \leq \frac{c}{(1+\abs{j-k})^L}, \qquad \forall j,k\in\bbN.
$$
Given $\Delta, \Lambda \subseteq \bbN$ and $p\in(L^{-1},1]$, let
$$
I_p(\Delta, \Lambda) = \max_{j\in\Delta} \sum_{k\in \Lambda}\abs{ \ip{\psi_k}{\psi_j} }^p .
$$
\end{definition}
\begin{remark}
Under this definition, wavelet frames have been shown to be intrinsically localized \cite{cordero2004localization} with the parameter $L$ being dependent on the regularity of the generating wavelets. For the anisotropic systems studied in \cite{grohs2013intrinsic} and \cite{grohs2014intrinsic},  
the definition of intrinsic localization used is more complex than the definition presented above. However, the key idea of how to exploit this property to obtain bounds on the localized sparsity values should still be applicable.
\end{remark}

\begin{remark}\label{rem:I_p}
Given any $\Delta, \Lambda \subseteq \bbN$ and $p\in(L^{-1},1]$, note that $Lp>1$ and
$$
I_p(\Delta, \Lambda) \leq \max_{j\in\Delta} \sum_{k\in\Lambda}\frac{c}{(1+\abs{j-k})^{Lp}} \leq 1+ \int_{1}^\infty \frac{c}{\abs{x}^{Lp}} \mathrm{d}x \leq 1+ \frac{c}{Lp-1}.
$$
So
$I_p(\Delta,\Lambda)$ is finite.
Moreover, if we let
$$d=\mathrm{dist}(\Delta, \Lambda):= \min_{k\in\Delta, j\in\Lambda}\abs{k-j}\geq 1,$$ then
$$
I_p(\Delta, \Lambda)  
\leq \int_d^\infty \frac{c}{\abs{x}^{Lp}} \mathrm{d}x = \frac{c}{(Lp-1)\cdot d^{Lp-1}}
$$
\end{remark}

The main result of this section is as follows. 
\begin{proposition}
Let $\Psi$ be a Parseval frame which is intrinsically localized with respect to $c=1$ (to simplify the amount of notation only)  and $L\geq 1$ and let $D$ be the associated analysis operator. Given any $\mathbf{N}=(N_k)_{k=1}^r\in\bbN^r$ and $\mathbf{s}=(s_k)_{k=1}^r\in\bbN^r$, let $s=s_1+\cdots+s_r$ and $N_0=0$,
$$
B = \sup\br{\nm{(P_\Delta D)^\dagger}: \Delta\subset [N_r], \abs{\Delta\cap (N_k-N_{k-1}]}=s_k}<\infty, $$
and
$$B' =  \sup\br{\nm{(P_\Delta DD^*P_\Delta)^\dagger}_{\ell^\infty \to \ell^\infty}: \Delta\subset [N_r], \abs{\Delta\cap (N_k-N_{k-1}]}=s_k}<\infty.
$$
Let
$
d_{j,k} = \mathrm{dist}(\Lambda_j, \Delta_k),
$
and recall the definition of localized sparsity and localized level sparsities from Definition \ref{def:analysis_sparsity}. Let $p\in (L^{-1},1]$.
Then,
\begin{itemize}
\item[(i)] $$\kappa(\mathbf{N}, \mathbf{s},p) \leq s\cdot  \max\br{ \left( B^{p }(1+1/(Lp-1))\right)^{1/(1-p/2)},B^{p }(1+1/(Lp-1)) } $$
\item[(ii)] For $j=1,\ldots, r$, $$\kappa_j(\mathbf{N}, \mathbf{s},p) \leq s_j \cdot \max\br{
 B^{p/(1-p/2)} \left(\sum_{k=1}^r   \frac{(s_k/s_j)^{1-p/2}}{d_{j,k}^{ Lp-1}}\right)^{1/(1-p/2)},\quad 
 \tilde B^p \sum_{k=1}^r   \frac{s_k/s_j}{d_{j,k}^{ Lp-1}}}.$$

\end{itemize}

\end{proposition}
\prf{
Note that $B$ and $B'$ are both finite, since there are finitely many subsets of $[N_r]$ and for each subset $\Delta$, $\br{\psi_j}_{j\in\Delta}$ is necessarily a frame for its span with a strictly positive lower frame bound. (i) follows from taking the maximum of (i) and (ii) of Proposition \ref{prop:intrinsic_loc} over all $\Delta$ subsets with an $(\mathbf{s},\mathbf{N})$-sparsity pattern and plugging in the estimate of $I_p$ from Remark \ref{rem:I_p}. (ii) follows from (iii) and (iv) of Proposition \ref{prop:intrinsic_loc}. 
}

\begin{proposition}\label{prop:intrinsic_loc}
Let $p\in (0,1]$ and let $\Delta\subset \bbN$ with $\abs{\Delta} = s$. Then, for all $g\in\cR(D^* P_\Delta)$, 
\begin{itemize}

\item[(i)] if $\nm{Dg}_{\ell^2} = 1$, then $\nm{D g}_{\ell^p}^p \leq \nm{(P_\Delta D)^\dagger}^p I_p(\Delta, \bbN) s^{1-p/2}$;
\item[(ii)] 
if $\nm{Dg}_{\ell^\infty }=1$, then $\nm{D g}_{\ell^p}^p\leq \nm{(P_\Delta D)^\dagger}^p I_p(\Delta, \bbN) s$.
\end{itemize}
 Let $\br{\Lambda_j}_{j=1}^r$ be a partition for $\bbN$, and let $\Delta_j = \Lambda_j \cap \Delta$ and $s_j = \abs{\Delta_j}$.  Then, for all $g\in\cR(D^* P_\Delta)$,
\begin{itemize}
\item[(iii)] if $\nm{Dg}_{\ell^2} = 1$, then $$\nm{P_{\Lambda_n}D g}_{\ell^p}^p \leq \nm{(P_\Delta D)^\dagger}^p \sum_{m=1}^r I_p(\Delta_m, \Lambda_n) s_m^{1-p/2}, \qquad n=1,\ldots, r;$$
\item[(iv)] 
if $\nm{Dg}_{\ell^\infty }=1$, then
$$
\nm{P_{\Lambda_n} D g}_{\ell^p}^p\leq \nm{(P_\Delta DD^*P_\Delta)^\dagger}_{\ell^\infty\to \ell^\infty}^p \sum_{m=1}^r I_p(\Delta_m, \Lambda_n) s_m, \qquad n=1,\ldots, r.
$$
\end{itemize}
\end{proposition}
\begin{proof}
For (i), suppose that $\nm{DD^*P_\Delta x}_{\ell^2}=1$.
\spl{\label{eq:intrinsic1}
&\nm{DD^*P_\Delta x}_{\ell^p}^p = \sum_{k\in\bbN}\abs{\sum_{j\in\Delta} x_j \ip{\psi_k}{\psi_j} }^p
\leq  \sum_{k\in \bbN}\sum_{j\in\Delta} \abs{ x_j}^p\abs{ \ip{\psi_k}{\psi_j} }^p\\
&= \sum_{j\in\Delta }\abs{ x_j}^p\sum_{k\in\Lambda}\abs{ \ip{\psi_k}{\psi_j} }^p
\leq \left(\max_{j\in\Delta} \sum_{k\in \bbN}\abs{ \ip{\psi_k}{\psi_j} }^p \right)\cdot \sum_{j\in\Delta }\abs{ x_j}^p\\
&\leq \left(\max_{j\in\Delta} \sum_{k\in \bbN}\abs{ \ip{\psi_k}{\psi_j} }^p \right)\cdot \nm{x}_{\ell^2}^p \cdot \abs{\Delta}^{1-p/2} \\
&\leq
 \left(\max_{j\in\Delta} \sum_{k\in \bbN}\abs{ \ip{\psi_k}{\psi_j} }^p \right)\cdot \nm{(P_\Delta D)^\dagger}^p \cdot \nm{D^* P_\Delta x}_{\ell^2}^p \cdot s^{1-p/2},
}
where we have applied the Cauchy-Schwarz inequality in the penultimate line. Therefore,
$$
\nm{DD^*P_\Delta x}_{\ell^p}^p \leq I(\Delta, \bbN) \nm{(P_\Delta D)^\dagger}^p s^{1-p/2}.
$$
To show (ii), if we instead assume that $\nm{DD^*P_\Delta x}_{\ell^\infty}=1$, then since
$$
\nm{DD^*P_\Delta x}_{\ell^2}^p = \nm{DD^*}_{\ell^2}^p \nm{x}_{\ell^2}^p \leq \nm{DD^*}_{\ell^2}^p s^{p/2},
$$
by plugging this into the last line of (\ref{eq:intrinsic1}), we obtain
$$
\nm{DD^*P_\Delta x}_{\ell^p}^p \leq I(\Delta, \Lambda) \nm{(P_\Delta D)^\dagger}^p s.
$$

The proof of (iii) is  similar to the above: if $\nm{D^*P_\Delta x}=1$, then for each $n=1,\ldots, r$,
\spl{\label{eq:intrinsic2}
&\nm{P_{\Lambda_n}DD^*P_\Delta x}_{\ell^p}^p 
\leq \sum_{m=1}^r \left(\max_{j\in\Delta_m} \sum_{k\in\Lambda}\abs{ \ip{\psi_k}{\psi_j} }^p \right)\cdot \sum_{j\in\Delta_m }\abs{ x_j}^p\\
&\leq \sum_{m=1}^r \sum_{j\in\Delta_m } \left(\max_{j\in\Delta_m} \sum_{k\in\Lambda}\abs{ \ip{\psi_k}{\psi_j} }^p \right)\cdot \nm{P_{\Delta_m}x}_{\ell^2}^p \cdot \abs{\Delta_m}^{1-p/2} \\
&\leq \nm{(P_\Delta D)^\dagger}^p \cdot \sum_{m=1}^r I(\Delta_m, \Lambda_n)  \cdot s_m^{1-p/2},
}
Finally, to show (iv), if $\nm{D^*P_\Delta x}_{\ell^\infty}=1$, then
for each $n=1,\ldots, r$,
\spl{\label{eq:intrinsic3}
&\nm{P_{\Lambda_n}DD^*P_\Delta x}_{\ell^p}^p \leq \sum_{m=1}^r \left(\max_{j\in\Delta_m} \sum_{k\in\Lambda}\abs{ \ip{\psi_k}{\psi_j} }^p \right)\cdot \sum_{j\in\Delta_m }\abs{ x_j}^p\\
&\leq \sum_{m=1}^r \sum_{j\in\Delta_m } \left(\max_{j\in\Delta_m} \sum_{k\in\Lambda}\abs{ \ip{\psi_k}{\psi_j} }^p \right)\cdot \nm{P_{\Delta_m}x}_{\ell^\infty}^p \cdot \abs{\Delta_m} \\
&\leq \nm{(P_\Delta D D^*P_\Delta)^\dagger}_{\ell^\infty \to \ell^\infty}^p \cdot \sum_{m=1}^r I(\Delta_m, \Lambda_n)\cdot  s_m,
}
 
\end{proof}

\section{Conditions for stable and robust recovery}\label{sec:duals}

Given $\Delta\subset \bbN$ and some $f\in\cH$, the following proposition presents conditions under which one is guaranteed robust and stable recovery up to $\nm{ P_\Delta^\perp D f}_{\ell^1}$.
\begin{proposition}[Dual certificate]\label{prop:dual_certificate}
Let $f\in \cH$.
Let $\Delta \subset \bbN$ be such that $\abs{\Delta} =s$. Let $\cW = \cR( D^* P_{\Delta})$. For $r\in\bbN$, let $\mathbf{q}=\br{q_j}_{j=1}^r\in (0,1]^r$ and let $\br{\Omega_k}_{k=1}^r$ be disjoint subsets of $\bbN$. Let $\Omega = \Omega_1 \cup \cdots \cup \Omega_r$.
Suppose that
\begin{itemize}
\item[(i)] $
\nm{ Q_{\cW}  V^* (q_1^{-1} P_{\Omega_1 }\oplus \cdots \oplus q_r^{-1}  P_{\Omega_r}) V Q_{\cW}  -  Q_{\cW}} < \frac{1}{4},
$
\item[(ii)] $\sup_{j\in\bbN} \nm{ P_{\br{j}}  D  Q_{\cW}^\perp V^* (q_1^{-1} P_{\Omega_1} \oplus \cdots \oplus q_r^{-1} P_{\Omega_r}) V  Q_{\cW}^\perp D^* P_{\br{j}}}  < \frac{5}{4},$
\end{itemize}
and that there exists $\rho = V^* P_\Omega w$ and $L>0$ with the following properties.
\begin{itemize}
\item[(iii)] $\nm{ D^* P_{\Delta} \sgn( P_{\Delta} D x) -  Q_\cW \rho} \leq   q^{1/2}/8$,
\item[(iv)] $\nm{ P_\Delta^\perp  D  Q_{\cW}^\perp \rho}_{\ell^\infty} \leq 1/2 $,
\item[(v)] $\nm{w}_{\ell^2} \leq L \, \sqrt{\kappa}$.
\end{itemize}
Let $y \in\ell^2(\bbN)$ be such that $\nm{ P_\Omega V f - y}  \leq \delta$.
Then,  any minimizer $\hat f \in \cH$ of (\ref{eq:min_orth})
satisfies
\be{\label{eq:error_bd_dual}
\nm{f-\hat f}_\cH \lesssim \delta  \, \left( q^{-1/2}  +  L\sqrt{\kappa}\right) +     \nm{P_\Delta^\perp D f}_{\ell^1}.
}
\end{proposition}
\begin{proof}
Since $D$ is an isometry, $D^* D$ is the identity on $\cH$. Given any $g\in\cH$, 
\be{\label{eq:Q_W}
Q_\cW^\perp g= Q_\cW^\perp D^*  D g =  Q_\cW^\perp D^* P_\Delta^\perp D g,
}
since $Q_\cW$ is the orthogonal projection onto $\cR(D^* P_{\Delta})$.
So, using the assumption that $\nm{D}_{\cH \to \ell^2}\leq 1$, for any $g\in\cH$,
$$
\nm{g} \leq \nm{Q_\cW g} + \nm{Q_\cW^\perp g} \leq \nm{ Q_\cW g}_\cH + \nm{  P_\Delta^\perp   D g}_{\ell^1}.
$$ Now, let $h = f-\hat f$. To bound $\nm{h}_\cH$, it suffices to derive bounds for $\nm{Q_\cW h}$ and $\nm{ P_\Delta^\perp  D h}_{\ell^1}$.

Let $V_{\Omega, \mathbf{q}} =  Q_{\cW} V^* (q_1^{-1} P_{\Omega_1 }\oplus \cdots \oplus q_r^{-1} P_{\Omega_r}) V Q_{\cW}$. We first observe that  (i) implies that $V_{\Omega, \mathbf{q}}$ has a bounded inverse on $ Q_{\cW}(\cH)$, with $$
\nm{
(Q_{\cW} V^* (q_1^{-1} P_{\Omega_1 }\oplus \cdots \oplus q_r^{-1} P_{\Omega_r}) V  Q_{\cW})^{-1}}_{\cH \to \cH} \leq \frac{4}{3}, $$
and 
$$
\nm{
 (q_1^{-1/2} P_{\Omega_1 }\oplus \cdots \oplus q_r^{-1/2}  P_{\Omega_r}) V  Q_{\cW}}  \leq \sqrt{\frac{5}{4}}.
$$
Observe also that
$$
\nm{P_\Omega V h}  \leq \nm{ P_\Omega V f - y}  + \nmu{ P_\Omega V \hat f - y} \leq 2\delta.
$$
By applying the above observations, we have that
\spl{\label{prop:dual:1}
\nm{Q_\cW h} &= \nm{V_{\Omega, \mathbf{q}}^{-1} V_{\Omega, \mathbf{q}}  h}\\&
\leq \nm{V_{\Omega, \mathbf{q}}^{-1}}_{\cH} \nm{ Q_{\cW} V^*(q_1^{-1}P_{\Omega_1} \oplus \cdots \oplus q_r^{-1} P_{\Omega_r}) V (h- Q_\cW^\perp)h}\\
&\leq \frac{4 }{3 }\,  \left( \sqrt{5}\, q^{-1/2}\, \delta + \sqrt{\frac{5}{4}}  \, \nm{(q_1^{-1/2} P_{\Omega_1 }\oplus \cdots \oplus q_r^{-1/2}  P_{\Omega_r}) V  Q_{\cW}^\perp h}  \right).
}
Also, by (ii),
\eas{
&\nm{(q_1^{-1/2} P_{\Omega_1 }\oplus \cdots \oplus q_r^{-1/2}  P_{\Omega_r}) V  Q_{\cW}^\perp h}  = \nm{(q_1^{-1/2} P_{\Omega_1 }\oplus \cdots \oplus q_r^{-1/2}  P_{\Omega_r}) V  Q_{\cW}^\perp  D^* P_\Delta^\perp   D h}_{\ell^2}\\
&\leq \sup_{j\in\Delta^c} \nm{(q_1^{-1/2} P_{\Omega_1 }\oplus \cdots \oplus q_r^{-1/2}  P_{\Omega_r})  V  Q_{\cW}^\perp  D^*e_j}_{\ell^2} \nm{ P_\Delta^\perp   D h}_{\ell^1} \leq \sqrt{\frac{5}{4}} \nm{ P_\Delta^\perp   D h}_{\ell^1}.
}
Plugging this into (\ref{prop:dual:1}) yields
\be{\label{eq:prop:dual:1}
\nm{ Q_\cW h}_\cH
\leq 2  \left(  q^{-1/2}\, \delta +  \nm{ P_\Delta^\perp   D h}_{\ell^1}\right).
}
So, to bound $\nm{h}_\cH$, it suffices to bound $\nm{ P_\Delta^\perp   D h}_{\ell^1}$.

Observe that
\eas{
\nmu{ D \hat f}_{\ell^1} &= \nmu{ P_\Delta^\perp   D (f+h)}_{\ell^1} + \nmu{  P_{\Delta} D (f+h)}_{\ell^1}\\ &
\geq \nm{ P_\Delta^\perp   D h}_{\ell^1} - \nm{ P_\Delta^\perp   D f}_{\ell^1} + \nm{ P_{\Delta} D f}_1 + \Re \ip{ P_{\Delta} D  h}{\sgn( P_{\Delta} D f)}\\ &
=  \nm{ P_\Delta^\perp   D h}_{\ell^1} - 2\nm{ P_\Delta^\perp   D f}_{\ell^1} + \nm{ D f}_1 + \Re \ip{ P_{\Delta} D  h}{\sgn( P_{\Delta} D f)}.
}
Since $\hat f$ is a minimizer, we can deduce that
$$
\nm{ P_\Delta^\perp   D h}_{\ell^1} \leq \abs{\ip{ P_{\Delta} D  h}{\sgn( P_{\Delta} D f)}} + 2\nm{ P_\Delta^\perp   D f}_{\ell^1}.
$$
Using the existence of $\rho =  V^*  P_\Omega w$ and recalling that $ Q_\cW^\perp =  Q_{\cW}^\perp  D^*  P_\Delta^\perp   D$ from (\ref{eq:Q_W}), we have that
\eas{
&\abs{\ip{ P_{\Delta} D  h}{\sgn( P_{\Delta} D f)}}
= \abs{\ip{ h}{ D^*\sgn( P_{\Delta} D f)}}\\&
\leq \abs{\ip{ h}{ D^*\sgn( P_{\Delta} D f) -  Q_\cW \rho}} + \abs{\ip{h}{\rho}} +\abs{\ip{h}{ Q_\cW^\perp\rho}}\\
&\leq \nm{ Q_{\cW} h}_\cH \nm{ D^*\sgn( P_{\Delta} D f) -  Q_\cW \rho}_\cH + \nm{ P_\Omega  V h}_{\ell^2} \nm{w}_{\ell^2} + \abs{\ip{ D^*  P_\Delta^\perp   D h}{ Q_\cW^\perp\rho}}\\
&\leq \frac{\sqrt{q}}{8  }\nm{ Q_{\cW} h}_\cH  + 2\delta \, L \, \sqrt{\kappa} + \frac{1}{4} \nm{ P_\Delta^\perp   D h}_{\ell^1} \\
&\leq \delta \, \left(\frac{1}{4}  + 2L\sqrt{\kappa}\right) + \frac{3}{4}\nm{ P_\Delta^\perp   D h}_{\ell^1},
}
where we have used the bound on $\nm{ Q_\cW h}_\cH$ from (\ref{eq:prop:dual:1}) to obtain the last inequality.
Therefore,
$$
\nm{ P_\Delta^\perp   D h}_{\ell^1} \leq 8\nm{ P_\Delta^\perp   D f}_{\ell^1} + \delta \, \left(1 + 8 L\sqrt{\kappa}\right).
$$
Finally, combining this with (\ref{eq:prop:dual:1}) yields
$$
\nm{h}_\cH \leq \delta \, \left(2  q^{-1/2}  + 3 \left(1  + 8 L\sqrt{\kappa}\right)\right) + 16\,\nm{ P_\Delta^\perp   D f}_{\ell^1}.
$$
\end{proof}

\begin{proposition}[Dual certificate for the unconstrained problem]\label{prop:dual_certificate_unconstr}
Consider the setting of Proposition \ref{prop:dual_certificate} and assume that conditions (i)-(v) are satisfied. Let $\alpha>0$ and suppose that $\hat f\in\cH$ is a minimizer of
$$
\inf_{g\in\cH} \alpha \nm{ D g}_{\ell^1} + \nm{ P_\Omega  V g - y}_{\ell^2}^2,
$$
where $y \in\ell^2(\bbN)$ is such that $\nm{ P_\Omega V f - y}  \leq \delta$.
Then,
$$
\nm{f-\hat f}_\cH \lesssim \frac{\delta^2}{\alpha} +  \alpha \left(\frac{1}{\sqrt{q}} +  L \sqrt{\kappa}\right)^2 + \delta \left(\frac{1}{\sqrt{q}} +  L \sqrt{\kappa}\right)+ \nm{ P_\Delta^\perp  D f}_{\ell^1}.
$$
Thus, if $\alpha = \sqrt{q}\delta$, then
$$
\nm{f-\hat f}_\cH \lesssim \delta\left(\frac{1}{\sqrt{q}} + L\sqrt{\kappa} + \sqrt{q}L^2 \kappa\right) + \nm{ P_\Delta^\perp  D f}_{\ell^1}.
$$
\end{proposition}

\begin{proof}
Let $h=  f- \hat f$, just as in Proposition \ref{prop:dual_certificate}, 
$$
\nm{h}_\cH \leq \nm{ Q_\cW h}_\cH + \nm{ P_\Delta^\perp  D h}_{\ell^1}
$$
and it suffices to bound the two terms on the right side of the inequality. We first consider $\nm{ Q_\cW h}_\cH$. By applying assumptions (i) and (ii), we can proceed as in the proof of Proposition \ref{prop:dual_certificate} to derive
\be{
\nm{ Q_\cW h}_\cH \leq \frac{4 }{3}\left( \sqrt{\frac{3}{2}}q^{-1/2} \nm{ P_\Omega  V h}_{\ell^2} + \frac{5}{4}\nm{ P_\Delta^\perp  D h}_{\ell^1}\right).
}
Then, by letting $\lambda = \nm{y -  P_\Omega  V \hat f}_{\ell^2}$ and observing that
$$\nm{ P_\Omega  V h}_{\ell^2}  \leq \nm{ P_\Omega  V f - y}_{\ell^2} + \nm{ P_\Omega  V \hat f - y}_{\ell^2} \leq \delta + \lambda,
$$
 we have that
\spl{\label{eq:uc1}
\nm{ Q_\cW h}_\cH &\leq \frac{4 }{3}\left( \sqrt{\frac{3}{2}}q^{-1/2} (\delta+\lambda) + \frac{5}{4}\nm{ P_\Delta^\perp  D h}_{\ell^1}\right)\\
&\leq 2  \left(\frac{(\delta+\lambda)}{\sqrt{q}} + \nm{ P_\Delta^\perp  D h}_{\ell^1}\right).
}

To bound $ \nm{ P_\Delta^\perp  D h}_{\ell^1}$, first observe that
\eas{
\alpha \nm{ D \hat f}_{\ell^1} \geq &\alpha \nm{ P_\Delta^\perp  D h}_{\ell^1} - 2\alpha \nm{ P_\Delta^\perp  D f}_{\ell^1} + \alpha \nm{ D f}_{\ell^1} + \alpha \Re\ip{ P_\Delta  D h}{\sgn( P_\Delta  D f)} \\&+ \lambda^2-\lambda^2 + \nm{y- P_\Omega  V f}_{\ell^2}^2 -\delta^2.
}
Since $\hat f$ is a minimizer, it follows that $
\alpha \nm{ D \hat f}_{\ell^1}+\lambda^2 \leq \alpha \nm{ D f}_{\ell^1} + \nm{y- P_\Omega  V f}_{\ell^2}^2$ and therefore,
\be{\label{eq:uc2}
\alpha \nm{ P_\Delta^\perp  D h}_{\ell^1} +  \lambda^2 \leq 2\alpha \nm{ P_\Delta^\perp  D f}_{\ell^1}  + \alpha \abs{\ip{ P_\Delta  D h}{\sgn( P_\Delta  D f)}}  +\delta^2.
}
In the same way as in the proof of Proposition \ref{prop:dual_certificate}, we may apply the properties of the dual certificate $\rho =  V^*  P_\Omega w$ to bound $\abs{\ip{ P_\Delta  D h}{\sgn( P_\Delta  D f)}}$, so that the following holds.
\eas{
\abs{\ip{ P_\Delta  D h}{\sgn( P_\Delta  D f)}} &\leq \frac{\sqrt{q}}{8 } \nm{ Q_\cW h}_\cH + \nm{ P_\Omega  V h}_{\ell^2}\nm{w}_{\ell^2} + \frac{1}{4}\nm{ P_\Delta^\perp  D h}_{\ell^1}.
}
By inserting the bound from (\ref{eq:uc1}), recalling that $\nm{ P_\Omega  V h}_{\ell^2}\leq \delta + \lambda$ and that $\nm{w}_{\ell^2}\leq L \sqrt{\kappa}$, it follows that
\bes{
\abs{\ip{ P_\Delta  D h}{\sgn( P_\Delta  D f)}} \leq(\delta+\lambda)\left(\frac{1}{4}+ L\sqrt{\kappa}\right) + \frac{1}{2}\nm{ P_\Delta^\perp  D h}_{\ell^1}.
}
Plugging this bound into (\ref{eq:uc2}) now yields
\be{\label{eq:uc3}
\lambda^2 + \frac{\alpha}{2}\nm{ P_\Delta^\perp  D h}_{\ell^1} \leq 2\alpha \nm{ P_\Delta^\perp  D f}_{\ell^1}+ \alpha(\delta+\lambda)\left(\frac{1}{4}+ L\sqrt{\kappa}\right)+\delta^2.
}
This implies that
$$\lambda^2-\alpha\left(\frac{1}{4}+L\sqrt{\kappa}\right)\lambda -\delta^2-\alpha\delta\left(\frac{1}{4}+L\sqrt{\kappa}\right)\leq 0
$$
and by applying the quadratic formula and observing that $\lambda\geq 0$, it follows that
$$
\lambda\leq \frac{\alpha/4+L\alpha\sqrt{\kappa} + \sqrt{(\alpha/4+L\alpha\sqrt{\kappa})^2+ 4\left(\delta^2+\alpha\delta (1/4+L\sqrt{\kappa})+2\alpha \nm{ P_\Delta^\perp  D f}_{\ell^1}\right)}}{2}.
$$
Note that $\nm{\cdot}_{\ell^2}\leq \nm{\cdot}_{\ell^1}$, and so,
\spl{\label{eq:uc4}
\lambda&\leq \alpha/4+L\alpha\sqrt{\kappa} + \delta + \sqrt{\alpha \delta(1/4+L\sqrt{\kappa})} + \sqrt{2\alpha \nm{ P_\Delta^\perp  D f}_{\ell^1}}\\
&\leq \alpha/4+L\alpha\sqrt{\kappa} + \delta + \alpha(1/8+L\sqrt{\kappa}/2)+\delta/2 + \sqrt{2\alpha \nm{ P_\Delta^\perp  D f}_{\ell^1}}\\
&=\frac{3\alpha}{2}(1/4+L\sqrt{\kappa}) + \frac{3\delta}{2} + \sqrt{2\alpha \nm{ P_\Delta^\perp  D f}}_{\ell^1},
}
where the second inequality comes from the fact that $\sqrt{ab}\leq (a+b)/2$ for any $a,b\geq 0$. We also know from (\ref{eq:uc3}) that
$$
\nm{ P_\Delta^\perp  D h}_{\ell^1} \leq 4 \nm{ P_\Delta^\perp  D f}_{\ell^1} + (\delta+\lambda)\left(\frac{1}{2}+ 2L\sqrt{\kappa}\right)+ \frac{2\delta^2}{\alpha}.
$$
By combining this with the bound from (\ref{eq:uc1}), 
\eas{
\nm{h}_{\cH} &\leq \nm{ Q_\cW h}_\cH + \nm{ P_\Delta^\perp  D h}_{\ell^1}\\
&\leq 3 \left( 4\nm{ P_\Delta^\perp  D f}_{\ell^1} + \delta\left(\frac{1}{2}+2L\sqrt{\kappa}\right) + \frac{2\delta^2}{\alpha}\right)\\& + \frac{2 \delta}{\sqrt{q}}
+ \left(\frac{2 }{\sqrt{q}} + 3\left(\frac{1}{2}+2L\sqrt{\kappa}\right)\right)\lambda\\
&\lesssim   \nm{ P_\Delta^\perp  D f}_{\ell^1} +\delta (1+L\sqrt{\kappa}) + \frac{\delta^2}{\alpha} + \frac{\delta}{\sqrt{q}} + \left(\frac{1}{\sqrt{q}} + 1+L \sqrt{\kappa}\right)\lambda.
}
Recalling (\ref{eq:uc4}),
\eas{
\left(\frac{1}{\sqrt{q}} + 1+L \sqrt{\kappa}\right)\lambda &\lesssim
\left(\frac{1}{\sqrt{q}} + 1+L \sqrt{\kappa}\right)\left(\alpha(1 +L\sqrt{\kappa}) + \delta + \sqrt{\alpha \nm{ P_\Delta^\perp  D f}}_{\ell^1}\right)\\
&\lesssim\left(\frac{1}{\sqrt{q}} + 1+L \sqrt{\kappa}\right)\left(\alpha(1 +L\sqrt{\kappa}) + \delta\right)
+ \frac{\alpha}{q} + \nm{ P_\Delta^\perp  D f}_{\ell^1} + \alpha (1+L \sqrt{\kappa})\\
&\leq \alpha \left(\frac{1}{\sqrt{q}} + 1+L \sqrt{\kappa}\right)^2 + \delta \left(\frac{1}{\sqrt{q}} + 1+L \sqrt{\kappa}\right)+ \nm{ P_\Delta^\perp  D f}_{\ell^1}.
}
Therefore,
$$
\nm{h}_\cH \leq   \frac{\delta^2}{\alpha} +  \alpha \left(\frac{1}{\sqrt{q}} + 1+L \sqrt{\kappa}\right)^2 + \delta \left(\frac{1}{\sqrt{q}} + 1+L \sqrt{\kappa}\right)+ \nm{ P_\Delta^\perp  D f}_{\ell^1}.
$$
\end{proof}

\section{Overview of the proof of Theorem \ref{thm:main}}
The remainder of this paper is focussed on the proof of Theorem \ref{thm:main}, and we begin by setting some notation which will be used throughout.

Let
 $ V,  D \in\cB(\cH, \ell^2(\bbN))$ be isometries. Let $f\in\cH$.  
 For $r\in\bbN$, $M\in\bbN$ and $N\in\bbN$, let $\mathbf{N} = \br{N_k}_{k=1}^r\in\bbN^r$, $\mathbf{M} = \br{M_k}_{k=1}^r\in\bbN^r$, $\mathbf{s} = \br{s_k}_{k=1}^r\in\bbN^r$, $\mathbf{m} = \br{m_k}_{k=1}^r\in\bbN^r$, with
 \begin{itemize}
 \item $0= M_0< M_1<\cdots<M_r=:M$, and let  $\Gamma_k = (M_{k-1}, M_k] \cap \bbN$ and $\Omega_k \sim \mathrm{Ber}(q_k, \Gamma_k)$.
  \item $0= N_0< N_1<\cdots<N_r=:N$, and let $\Lambda_k = (N_{k-1}, N_k] \cap \bbN$ for $k=1,\ldots,r-1$ and $\Lambda_r = (N_{r-1}, \infty)\cap \bbN$.

 \item $m_k\leq M_k-M_{k-1}$ and let $\mathbf{q}=\br{q_j}_{j=1}^r\in (0,1]^r$ with $q_j = m_j/(M_j-M_{j-1})$.
 \item $s_k\leq N_k-N_{k-1}$ and let  $\Delta \subset [N]$ be such that $\abs{\Delta} =s_1+\cdots+s_r =: s$ and $\Delta_k =\Lambda_k \cap \Delta$, $\abs{\Delta_k} = s_k$. Let $\cW = \cR( D^*  P_{\Delta})$.
 \end{itemize}  
For some $p\in (0,1]$, we will write $\boldsymbol{\kappa} = \br{\kappa_j}_{j=1}^r$ with $\kappa_l = \kappa_l(\mathbf{N}, \mathbf{s},p)$ and $\hat \kappa_l = \hat \kappa_l(\mathbf{N}, \mathbf{M}, \boldsymbol{\kappa})$ for each $l=1,\ldots, r$. Let $\kappa_{\min} = r\min_{l=1}^r \kappa_l$ and $\kappa_{\max} = r\max_{l=1}^r \kappa_l$. Note that $\kappa_{\min}\leq \sum_{l=1}^r \kappa_j \leq \kappa_{\max}$.
 
 We also define $T\in\cB(\ell^2(\bbN), \ell^2(\bbN))$ such that given any $x=(x_j)_{j\in\bbN}\in\ell^2(\bbN)$,
$$
Tx = y, \qquad y_j = \frac{x_j}{\max\br{1,\sqrt{r\kappa_k}}}, \quad j\in\Lambda_k, \quad k=1,\ldots, r.
$$ 
Observe that $T$ is an invertible operator, $\nm{T} \leq 1/\sqrt{ \kappa_{\min}}$ and $\nm{T^{-1}}\leq \sqrt{\kappa_{\max}}$.

\subsection{Outline of the proof}

To prove Theorem \ref{thm:main}, it suffices to show that conditions (i)-(v) of Proposition \ref{prop:dual_certificate} are satisfied with high probability whenever the sampling scheme is the multilevel sampling scheme $\Omega = \Omega_{\mathbf{M}, \mathbf{m}}$ described in Theorem \ref{thm:main}. To this end, we first remark that it has become customary in compressed sensing theory to deduce recovery statements for uniform sampling models by first proving statements based on some alternative sampling model which is easier to analyse. One approach, considered in \cite{candes2006robust, adcockbreaking, BAACHGSCS} is to consider a Bernoulli sampling model, defined below. 
\begin{definition}
Let $\mathbf{M} = \br{M_k}_{k=1}^r$ with $0=M_0<M_1<\cdots<M_r$. Let $\Omega^{\mathrm{Ber}}_{\mathbf{M}, \mathbf{m}} := \Omega_1^{\mathrm{Ber}}\cup \cdots \cup \Omega_r^{\mathrm{Ber}}$, where  $\Omega_k^{\mathrm{Ber}} := \br{\delta_j \cdot j : j\in \Gamma_k}$ with $\Gamma_k = \bbN \cap (M_{k-1}, \cdots, M_k]$, and $\delta_j$ are independent random variables such that $\bbP(\delta_j = 1) = m_k/(M_k - M_{k-1})$ and $\bbP(\delta_j = 0) = 1-m_k/(M_k - M_{k-1})$. 
The Bernoulli sampling set $\Omega_k^{\mathrm{Ber}} $ described will be denoted by $\Omega_k^{\mathrm{Ber}}  \sim \mathrm{Ber}(m_k/(M_k - M_{k-1}), \Gamma_k)$ and we say that $\Omega^{\mathrm{Ber}}_{\mathbf{M}, \mathbf{m}}= \Omega_1^{\mathrm{Ber}}\cup \cdots \cup \Omega_r^{\mathrm{Ber}}$ is a Bernoulli $(\mathbf{M}, \mathbf{m})$-sampling scheme.

\end{definition}
As explained in \cite[II.C]{candes2006robust} (see also \cite{foucart2013mathematical}), the probability that one of the conditions of Proposition \ref{prop:dual_certificate} fails for $\Omega = \Omega_{\mathbf{M}, \mathbf{m}}$ chosen uniformly at random is up to a constant bounded from above by the probability that one of these conditions fails under the Bernoulli $(\mathbf{M}, \mathbf{m})$-sampling scheme, $\Omega= \Omega^{\mathrm{Ber}}_{\mathbf{M}, \mathbf{m}}$.

So, to prove Theorem \ref{thm:main}, it suffices to show that conditions (i) to (v) of Proposition \ref{prop:dual_certificate} hold with probability exceeding $(1-\epsilon)$ with $\Omega = \Omega_{\mathbf{M}, \mathbf{m}}^{\mathrm{Ber}}$ satisfying the following assumption.

\begin{assumption}\label{dual_assumptions} 
Let $\cL = \left(\log( \epsilon^{-1}) +1\right)\, \log(\tilde M q^{-1} \kappa_{\max}^{1/2} \nm{ D D^*}_{\ell^\infty\to \ell^\infty}  )$ and 
 $$ 
B   = \max \br{\nm{DQ_{\cW}^\perp D^*}_{\ell^\infty \to \ell^\infty} , \sqrt{\nm{DQ_{\cW}D^*}_{\ell^\infty \to \ell^\infty} \max_{l=1}^r\br{\sum_{t=1}^r \nm{P_{\Lambda_l}D Q_{\cR(D^*P_\Delta)} D^* P_{\Lambda_t}}_{\ell^\infty \to \ell^\infty}}}}.
$$
Let $$
\tilde M = \min\br{i\in\bbN : 
\max_{j\geq i} \nm{ P_{[M]}  V  D^* e_j}_{\ell^2} \leq \frac{q}{8\sqrt{\kappa_{\max}}}, \qquad \max_{j\geq i} \nm{Q_{\cR(D^*P_{[N]})} D^* e_j}\leq \frac{\sqrt{5q}}{4}},
$$

Suppose that
\begin{itemize}
\item[(a)] $
\nm{ D  Q_\cW V^*  P_{[M]}^\perp  V Q_\cW D   }_{\ell^2 \to \ell^2} \leq \sqrt{\frac{\kappa_{\min}}{2\kappa_{\max}}}\log^{-1/2}(4q^{-1}\sqrt{\kappa_{\max}}\tilde M).
$
\item[(b)] $
\nm{    D  Q_{\cW}^\perp  V^*  P_{[M]}  V Q_\cW D^* T^{-1}}_{\ell^2 \to \ell^\infty} \leq \frac{1}{8\sqrt{\kappa_{\max}}}.
$
\item[(c)] For $k=1,\ldots, r$, $$
 q_k \gtrsim \sqrt{r}\, \cL  \, B \, \sum_{l=1}^r \mu_{\mathbf{N}, \mathbf{M}}^2(k,l) \, \kappa_l.
$$
\item[(d)]For $k=1,\ldots, r,$
$q_k \geq \cL \, \hat q_k$ with $\br{\hat q_k}_{k=1}^r$ satisfying
$$
1\gtrsim r\, B^2 \,  \sum_{k=1}^r (\hat q_k^{-1}-1)  
\,\mu_{\mathbf{N}, \mathbf{M}}^2(k,j) \, \hat  \kappa_k, \qquad j=1,\ldots, r.
$$
\end{itemize}
\end{assumption}
Note that this assumption is strictly weaker than the assumptions of Theorem \ref{thm:main}, and by showing that conditions (i) to (v) of Proposition \ref{prop:dual_certificate}, we will prove that the error bound (\ref{eq:error_bd_dual}) holds for one support set $\Delta$. So, by ensuring that the conditions of this assumption hold over all $\Delta$ sets which are $(\mathbf{N},\mathbf{s})$ sparse patterns (as required by Theorem \ref{thm:main}), we can conclude that (\ref{eq:error_bd_dual}) holds for any  $(\mathbf{N},\mathbf{s})$ sparse support sets.

Under this assumption, 
\begin{itemize}
\item \S \ref{sec:dual_constr} will show that  conditions (iii) to (v) of Proposition \ref{prop:dual_certificate} are satisfied with probability exceeding $1-5\epsilon/6$;
\item  \S \ref{sec:isometr_props} will show that conditions (i) and (ii) of Proposition \ref{prop:dual_certificate} are satisfied with probability exceeding $1-\epsilon/6$;
\item \S \ref{sec:prelim} will present some preliminary results for use in \S \ref{sec:dual_constr} and \S \ref{sec:isometr_props}.
\end{itemize}

\paragraph{The proof of Corollary \ref{cor:unconstr}}
Once we have shown that the conditions of Proposition \ref{prop:dual_certificate} hold with probability exceeding $1-\epsilon$, the conclusion of Proposition \ref{prop:dual_certificate_unconstr} automatically follows and we may conclude Corollary \ref{cor:unconstr}.

\section{Preliminary results}\label{sec:prelim}\label{sec:prelims}
In this section, we present four propositions which will be applied to show that the conditions of Proposition \ref{prop:dual_certificate} are satisfied with high probability under the conditions of Theorem \ref{thm:main} with a Bernoulli sampling scheme. The results in this section are derived using Talagrand's inequality and Bernstein inequalities (for random variables and random matrices) which we state below.

\begin{theorem}\label{Talagrand}\textnormal{ (Talagrand \cite[Cor.\ 
7.8]{Ledoux})}
There exists a number $K$ with the following property. Consider $n$
independent random variables $X_i$ valued in a measurable space
$\Omega $ and let $\mathcal{F}$ be a (countable) class  of measurable
functions on $\Omega.$ Let $Z$ be the random variable $Z = \sup_{f \in
  \mathcal{F}}\sum_{i \leq n} f(X_i)$ and define
$$
S = \sup_{f \in \mathcal{F}}\|f\|_{\infty}, \qquad V = \sup_{f \in
  \mathcal{F}} \mathbb{E}\left( \sum_{i \leq n}f(X_i)^2  
 \right).
$$
If $ \mathbb{E}(f(X_i)) = 0$ for all $f \in \mathcal{F}$ and $i\leq n$, then, 
for each $t > 0$, 
we have
$$
\mathbb{P}(|Z - \mathbb{E}(Z)| \geq t) \leq 3 \exp \left(
-\frac{1}{K}\frac{t}{S} \log\left( 1 + \frac{tS}{V+S\mathbb{E}(\overline{Z})}  
\right)\right),
$$
where $\overline{Z} = \sup_{f\in\mathcal{F}}|\sum_{i \leq n} f(X_i)|$.
\end{theorem}

\begin{theorem}[Bernstein inequality for random variables \cite{foucart2013mathematical}]\label{thm:bernstein}
Let $Z_1,\ldots, Z_M \in\bbC$ be independent random variables with zero mean such that
$
\abs{Z_j} \leq K
$ almost surely for  all $l=1,\ldots, M$ and some constant $K>0$. Assume also that $\sum_{j=1}^M\bbE\abs{Z_j}^2 \leq \sigma^2$ for some constant $\sigma^2 >0$. Then for $t>0$,
$$
\bbP\left(\abs{\sum_{j=1}^M Z_j} \geq t\right) \leq 4\exp\left(-\frac{t^2/4}{\sigma^2 + Kt/(3\sqrt{2})}\right).
$$
If $Z_1,\ldots, Z_M \in\bbR$ are real instead of complex random variables, then 
$$
\bbP\left(\abs{\sum_{j=1}^M Z_j} \geq t\right) \leq 2\exp\left(-\frac{t^2/2}{\sigma^2 + Kt/3}\right).
$$
\end{theorem}

\begin{theorem}[Bernstein inequality for rectangular matrices \cite{tropp2012user}]\label{thm:matrixBernstein}
Let $Z_1,\ldots, Z_M \in\bbC^{d_1\times d_2}$ be independent random matrices  such that $\bbE Z_j = 0$ for each  $j=1,\ldots, M$ and $\nm{Z_j}_{2\to 2} \leq K$ almost surely for each  $j=1,\ldots, M$ and some constant $K>0$. Let $$\sigma^2 := \max \br{\nm{\sum_{j=1}^M \bbE(Z_j Z_j^*)}_{\ell^2\to \ell^2}, \nm{\sum_{j=1}^M \bbE(Z_j^* Z_j)}_{\ell^2 \to \ell^2}}.$$
Then, for $t>0$,
$$
\bbP\left(\nm{\sum_{j=1}^M Z_j}_{\ell^2\to \ell^2}\geq t\right) \leq 2(d_1+d_2)\exp\left(\frac{-t^2/2}{\sigma^2 + Kt/3}\right)
$$
\end{theorem}

\begin{proposition}\label{prop1}
Let $g\in\cW$ and let $\alpha>0$ and $\gamma \in [0,1]$.   Suppose that
\be{\label{eq:prop1_assp}
\nm{T D  (Q_{\cW}  V^*  P_{[M]}  V  Q_{\cW} -    Q_{\cW}  ) D^*T^{-1}}_{\ell^2\to \ell^2} \leq \alpha/2.
}
Then
$$
\bbP\left(\nm{ T  D ( Q_{\cW}  V^* (q_1^{-1}  P_{\Omega_1}\oplus \cdots \oplus q_r^{-1}  P_{\Omega_r})  V  Q_{\cW} -  Q_{\cW}) g}_{\ell^2} \geq \alpha \nm{T  D  g}_{\ell^2}\right) \leq \gamma
$$
provided that
\be{\label{eq:prop1_cond1}
  \sqrt{ r}\tilde B \log\left(\frac{3}{\gamma}\right) \sum_{l=1}^r \mu^2_{\mathbf{N},\mathbf{M}}(k,l) \kappa_l  \lesssim  \alpha, \qquad k=1,\ldots, r,
} and 
\be{\label{eq:prop1_cond2}
r \tilde B^2 \log\left(\frac{3}{\gamma}\right)  \sum_{k=1}^r (  q_k^{-1}-1) \, \mu_{\mathbf{N},\mathbf{M}}^2(k,l) \, \hat \kappa_k
\lesssim \alpha^2, \qquad l=1,\ldots, r,
}
where
$$\tilde B^2 =\nm{DQ_\cW D^*}_{\ell^\infty \to \ell^\infty} \max_{l=1}^r  \sum_{t=1}^r \nm{P_{\Lambda_l}DQ_\cW D^* P_{\Lambda_t} }_{\ell^\infty \to \ell^\infty}.$$

\end{proposition}
\begin{proof}
Without loss of generality, assume that $\nm{T D g}_{\ell^2} =1$.
Let $\br{\delta_j}_{j=1}^M$ be random Bernoulli variables such that $\bbP(\delta_j = 1) = \tilde q_j$ where $\tilde q_j = q_k$ for $j=M_{k-1}+1,\ldots, M_k$. Then,
\eas{
&   T  D (  Q_{\cW}  V^* (q_1^{-1}  P_{\Omega_1}\oplus \cdots \oplus q_r^{-1}  P_{\Omega_r})  V  Q_{\cW}  - Q_{\cW} ) g\\
&= \sum_{j=1}^M (\tilde q_j^{-1} \delta_j -1)  T  D  Q_{\cW}  V^* (e_j \otimes \overline{e}_j)  V   Q_{\cW}g +  T D  (Q_{\cW}  V^*  P_{[M]}  V  Q_{\cW} -    Q_{\cW}  ) g,
}
where $\otimes$ is the Kronecker product.
Since $D^*D=I$ and since (\ref{eq:prop1_assp}) holds by assumption,
\eas{
&\nm{T D  (Q_{\cW}  V^*  P_{[M]}  V  Q_{\cW} -    Q_{\cW}  ) g}_{\ell^2}
= \nm{T D  (Q_{\cW}  V^*  P_{[M]}  V  Q_{\cW} -    Q_{\cW}  ) D^*T^{-1}TD g}_{\ell^2}\\
&\leq \nm{ T D  (Q_{\cW}  V^*  P_{[M]}  V  Q_{\cW} -    Q_{\cW}  ) D^*T^{-1}}_{\ell^2 \to \ell^2}
\leq \alpha/2.
}
So, it suffices to show that 
$$
\bbP\left( \nm{\sum_{j=1}^M Y_j}_{\ell^2} \geq \alpha/2 \right)\leq \gamma,
$$
where for each $j=1,\ldots, M$,  $Y_j =(\tilde q_j^{-1} \delta_j -1)      D  Q_{\cW}  V^* (e_j \otimes \overline{e}_j)  V  g$. We will aim to apply Talagrand's inequality (Theorem \ref{Talagrand}) to obtain this probability bound.

Let $\cG$ be a countable set of vectors in the unit ball of $\ell^2(\bbN)$ and for each $\zeta \in\cG$, define the linear functionals $\hat \zeta_1, \hat \zeta_2: \ell^2(\bbN)\to \bbR$ by
$$
\hat \zeta_1(y) = \Re \ip{y}{\zeta}, \quad \hat \zeta_2(y) = - \Re \ip{y}{\zeta}, \qquad \forall y\in\ell^2(\bbN).
$$
Let $\cF = \br{\hat \zeta_1, \hat \zeta_2: \zeta \in\cG}$. Then, $Z := \sup_{f\in\cF}\sum_{j=1}^M f(Y_j) = \nm{\sum_{j=1}^M Y_j}_{\ell^2}$.

\begin{itemize}
\item To bound $\cS = \max_{j}\nm{Y_j}_{\ell^2}$: \eas{
&\nm{Y_j}_{\ell^2} \leq \tilde q_j^{-1} \nm{T D  Q_{\cW}  V^* (e_j \otimes \overline{e}_j)  V  g}_{\ell^2}
=\tilde q_j^{-1} \sup_{\nm{x}_{\ell^2}=1} \abs{\ip{ TD  Q_{\cW}  V^* (e_j \otimes \overline{e}_j)  V  g}{x}}\\
&= \tilde q_j^{-1} \sup_{\nm{x}_{\ell^2}=1} \abs{\ip{  V  g}{e_j} \ip{TD  Q_{\cW}  V^* e_j}{x}}
}
For each $j\in\Gamma_k$, 
$$\abs{\ip{  V  g}{e_j}} = \abs{\ip{  VD^*D  g}{e_j}} \leq \sum_{l=1}^r \abs{\ip{VD^*P_{\Lambda_l}Dg}{e_j}} \leq \sum_{l=1}^r \mu(P_{\Gamma_k} VD^* P_{\Lambda_l}) \nm{P_{\Lambda_l}Dg}_{\ell^1}.
$$
Observe that $\nm{TDg}_{\ell^2}=1$ implies $\nm{P_{\Lambda_l}Dg}_{\ell^2}\leq \sqrt{r \kappa_l}$ for each $l=1,\ldots, r$. Furthermore, by (4) of Corollary \ref{cor:loc_sp}, this implies that $\nm{P_{\Lambda_l} D g}_{\ell^1}\leq \kappa_l\sqrt{r}$. So, it follows that $$\abs{\ip{  V  g}{e_j}}  \leq \sqrt{r}\sum_{l=1}^r \mu(P_{\Gamma_k} VD^* P_{\Lambda_l}) \kappa_l.$$ Also, if $\nm{x}_{\ell^2}=1$,
\spl{\label{eq:prop1_bd}
&\abs{\ip{TD  Q_{\cW}  V^* e_j}{x}}^2 \leq \sum_{l=1}^r \nm{P_{\Lambda_l}TD  Q_{\cW}  V^* e_j}_{\ell^2}^2
= \sum_{l=1}^r \frac{1}{r\kappa_l} \nm{P_{\Lambda_l} D  Q_{\cW}  D^*DV^* e_j}_{\ell^2}^2
\\
&\leq \sum_{l=1}^r \frac{1}{r\kappa_l} \nm{P_{\Lambda_l} D  Q_{\cW}  D^*}_{\ell^\infty \to \ell^2}^2 \nm{DV^* e_j}_{\ell^\infty}^2 \leq \nm{D  Q_{\cW}  D^*}_{\ell^\infty \to \ell^\infty}^2\mu^2(P_{\Gamma_k}VD),
}
where the last inequality follows because  $\nm{P_{\Lambda_l}D f}_{\ell^2}^2 \leq \kappa_l$ for all $f\in\cW$ with $\nm{Df}_{\ell^\infty}\leq 1$. Therefore,
$$
\max_{j}\nm{Y_j}_{\ell^2} \leq \max_{k=1}^r \nm{D  Q_{\cW}  D^*}_{\ell^\infty \to \ell^\infty} \sqrt{r}\sum_{l=1}^r \mu^2_{\mathbf{N},\mathbf{M}}(k,l) \kappa_l.
$$

\item To bound $\cV= \sup_{f\in\cF} \bbE \sum_{j=1}^M f(Y_j)^2 = \sup_{\zeta \in\cG} \bbE \sum_{j=1}^M \abs{\ip{\zeta}{Y_j}}^2$:
\spl{\label{eq:prop1_V}
\cV &\leq \sup_{\zeta \in\cG} \sum_{j=1}^M (\tilde q_j^{-1}-1) \abs{\ip{e_j}{Vg}}^2 \abs{\ip{V^* e_j}{Q_\cW D^* T \zeta}}^2\\
&\leq   \sup_{\zeta \in\cG} \sum_{k=1}^r (  q_k^{-1}-1) \nm{P_{\Gamma_k}Vg}_{\ell^2}^2 \max_{j\in\Gamma_k} \abs{\ip{V^* e_j}{Q_\cW D^* T \zeta}}^2\\
&\leq  \sum_{k=1}^r (  q_k^{-1}-1) \nm{P_{\Gamma_k}Vg}_{\ell^2}^2  \max_{j\in\Gamma_k} \nm{TDQ_\cW V^* e_j}_{\ell^2}^2 .
}
By combining $\nm{P_{\Lambda_l}D g}_{\ell^2} \leq \sqrt{r \kappa_l}$ (which follows from $\nm{TDg}_{\ell^2}=1$) with the definition of $\hat \kappa_k$, we obtain $\nm{P_{\Gamma_k}V g}_{\ell^2}^2 \leq r\hat \kappa_k$.  Also, for each $j\in\Gamma_k$,
\eas{
\nm{TDQ_\cW V^* e_j}_{\ell^2}^2 &= \sum_{t=1}^r \frac{1}{r\kappa_t}\nm{P_{\Lambda_t}DQ_\cW V^* e_j}_{  \ell^2}^2 \\
&= \sum_{t=1}^r \frac{1}{r\kappa_t} \nm{P_{\Lambda_t}DQ_\cW D^*}_{\ell^\infty \to \ell^2} \nm{DV^* e_j}_{\ell^\infty} \nm{P_{\Lambda_t}DQ_\cW V^* e_j}_{\ell^2}\\
&\leq \sum_{t=1}^r \frac{1}{r\sqrt{\kappa_t}} \nm{DQ_\cW D^*}_{\ell^\infty \to \ell^\infty}  \nm{DV^* e_j}_{\ell^\infty} \sum_{l=1}^r \nm{P_{\Lambda_t}DQ_\cW D^*{\Lambda_l}}_{\ell^\infty \to \ell^2} \nm{P_{\Lambda_l} V^* e_j}_{\ell^\infty}\\
&\leq \sum_{t=1}^r \frac{1}{r } \nm{DQ_\cW D^*}_{\ell^\infty \to \ell^\infty}    \sum_{l=1}^r \nm{P_{\Lambda_t}DQ_\cW D^*{\Lambda_l}}_{\ell^\infty \to \ell^\infty} \mu_{\mathbf{N},\mathbf{M}}(k,l)
}
where we have used, from the definition of the $\kappa_j$'s,  $\nm{P_{\Lambda_t}D f}_{  \ell^2} \leq \sqrt{\kappa_t}$ whenever $t=1,\ldots, r$, $f\in\cW$ and $\nm{Df}_{\ell^\infty}\leq 1$.
Therefore,
\eas{
\cV & \leq  r\nm{DQ_\cW D^*}_{\ell^\infty \to \ell^\infty}   \sum_{l=1}^r\sum_{k=1}^r (  q_k^{-1}-1) \hat \kappa_k  \mu_{\mathbf{N},\mathbf{M}}^2(k,l)   \sum_{t=1}^r \frac{1}{r }\nm{P_{\Lambda_t}DQ_\cW D^* P_{\Lambda_l} }_{\ell^\infty \to \ell^\infty}\\
&\leq r \tilde B^2 \max_{l=1}^r \sum_{k=1}^r (  q_k^{-1}-1) \, \mu_{\mathbf{N},\mathbf{M}}^2(k,l) \, \hat \kappa_k,
}
where we have applied
\eas{&\frac{1}{r}\nm{DQ_\cW D^*}_{\ell^\infty \to \ell^\infty} \sum_{l=1}^r  \sum_{t=1}^r \nm{P_{\Lambda_t}DQ_\cW D^* P_{\Lambda_l} }_{\ell^\infty \to \ell^\infty}
\\&
\leq \nm{DQ_\cW D^*}_{\ell^\infty \to \ell^\infty} \max_{l=1}^r \sum_{t=1}^r\nm{P_{\Lambda_l}DQ_\cW D^*P_{\Lambda_t}}_{\ell^\infty \to \ell^\infty}
=: \tilde B^2.
}

\item To bound $\bbE(Z) = \bbE \nm{\sum_{j=1}^M Y_j}$:
\eas{
&\bbE\left(\nm{\sum_{j=1}^M Y_j}_{\ell^2}^2\right) = \sum_{j=1}^M \bbE \nm{Y_j}_{\ell^2}^2
= \sum_{j=1}^M (\tilde q_j^{-1} -1) \nm{TDQ_\cW V^* e_j}_{\ell^2}^2 \abs{\ip{Vg}{e_j}}^2\\
&\leq \sum_{k=1}^r (q_k^{-1}-1) \nm{P_{\Gamma_k}V g}_{\ell^2}^2 \max_{j\in\Gamma_k} \nm{TDQ_\cW V^* e_j}_{\ell^2}^2.
}
This is the same upper bound as obtained in (\ref{eq:prop1_V}), so from the bound on $\cV$, we have that
$$
\bbE\left(\nm{\sum_{j=1}^M Y_j}_{\ell^2}^2\right)\leq \tilde B^2 \max_{l=1}^r \sum_{k=1}^r (  q_k^{-1}-1) \, \mu_{\mathbf{N},\mathbf{M}}^2(k,l) \, \hat \kappa_k.
$$
Finally, by Jensen's inequality,
$$
\bbE \nm{\sum_{j=1}^M Y_j} \leq  \sqrt{\bbE\left(\nm{\sum_{j=1}^M Y_j}_{\ell^2}^2\right) } \leq \sqrt{\tilde B^2 \max_{l=1}^r \sum_{k=1}^r (  q_k^{-1}-1) \, \mu_{\mathbf{N},\mathbf{M}}^2(k,l) \, \hat \kappa_k}.
$$
\end{itemize}

Let $C = \max\br{C_1, 4 C_2/\alpha  }$, where
$$
C_1 = \nm{D  Q_{\cW}  D^*}_{\ell^\infty \to \ell^\infty} \sqrt{r}\max_{k=1}^r  \sum_{l=1}^r \mu_{\mathbf{N},\mathbf{M}}(k,l) \kappa_l,
$$
and
$$
C_2 = r\tilde B^2 \max_{l=1}^r \sum_{k=1}^r (  q_k^{-1}-1) \, \mu_{\mathbf{N},\mathbf{M}}^2(k,l) \, \hat \kappa_k.
$$
Note that $C  \geq \cS$, $\alpha C/4 \geq \cV$. Suppose that $C_2\leq \alpha^2/16$, then $\bbE(Z)\leq  \alpha /4$, since $  \left(\bbE(Z)\right)^2 \leq C_2$. By applying Talagrand's inequality with the upper bound of $C\geq \cS$ and $\bbE(Z)\leq  \alpha /4$,
\eas{
&\bbP\left(Z \geq \frac{\alpha}{2}\right) \leq \bbP\left(Z \geq \frac{\alpha}{4} + \bbE(Z)\right)
\geq \bbP\left(\abs{Z-\bbE(Z)} \leq \frac{\alpha}{4}\right)\\
&\leq 3\exp\left(-\frac{\alpha}{4KC} \log \left(1+ \frac{C\alpha/4}{\cV + C\alpha/4}\right)\right)
\leq 3\exp\left(-\frac{\alpha}{4KC} \log \left(1+ \frac{C\alpha/4}{C\alpha/4 + C\alpha/4}\right)\right)\\
&\leq 3\exp\left(-\frac{\alpha}{4KC} \log \left(\frac{3}{2}\right)\right),
}
where $K$ is the constant from Talagrand's inequality. So, $\bbP\left(Z \geq \frac{\alpha}{2}\right) \leq \gamma$ provided that
$$
C\log\left(\frac{3}{\gamma}\right) \leq \frac{\alpha}{4K} \log\left(\frac{3}{2}\right).
$$
as well as $C_2 \leq \alpha^2/16$.
Therefore, the require result would follow provided that
$$
\nm{D  Q_{\cW}  D^*}_{\ell^\infty \to \ell^\infty} \sqrt{r}\log\left(\frac{3}{\gamma}\right) \sum_{l=1}^r \mu_{\mathbf{N},\mathbf{M}}(k,l) \kappa_l  \leq \frac{\alpha}{4K} \log\left(\frac{3}{2}\right), \qquad k=1,\ldots, r,
$$ and 
$$
r\tilde B^2  \log\left(\frac{3}{\gamma}\right) \sum_{k=1}^r (  q_k^{-1}-1) \, \mu_{\mathbf{N},\mathbf{M}}^2(k,l) \, \hat \kappa_k
\leq \frac{\alpha^2}{16}\min\br{1, \,\frac{1}{K} \log\left(\frac{3}{2}\right)}, \qquad l=1,\ldots, r.
$$

\end{proof}

\begin{proposition}\label{prop2}
Fix $g\in\cH$ and let $\alpha>0$ and $\gamma \in [0,1]$. Suppose that
\be{\label{eq:prop2assp}
\nm{   D  Q_{\cW}^\perp  V^*  P_{[M]}  V Q_\cW D^* T^{-1}}_{\ell^2 \to \ell^\infty} \leq \frac{\alpha}{2 }.
}
Let $B = \nm{ D Q_\cW^\perp D^* }_{\ell^\infty \to \ell^\infty}$.
Let $$
\tilde M =\min\br{i\in\bbN : \max_{j\geq i} 2
\sqrt{\kappa_{\max}} \cdot \max_{k=1}^r q_k^{-1} \cdot \nm{ P_{[M]}  V  D^* e_j}_{\ell^2} \leq  \alpha}, \qquad \kappa_{\max}= r\max\br{\kappa_j}_{j=1}^r.
$$
Then $\tilde M$ is finite and
$$
\bbP\left( \nm{ P_\Delta^\perp   D  Q_{\cW}^\perp  V^* (q_1^{-1} P_{\Omega_1} \oplus \cdots \oplus q_r^{-1} P_{ \Omega_r})  V   Q_{\cW} g}_{\ell^\infty} \geq \alpha \nm{ T D g}_{\ell^2}\right) \leq \gamma,
$$
provided that
$$
\sqrt{r} \,B\, \log\left(\frac{4 \tilde M}{\gamma}\right)\,
q_k^{-1} \,   \sum_{l=1}^r \mu_{\mathbf{N}, \mathbf{M}}^2(k,l) \, \kappa_l \lesssim \alpha, \qquad k=1,\ldots, r
$$ 
and
 $$
r\, B^2\, \log\left(\frac{4 \tilde M}{\gamma}\right) \, \sum_{k=1}^r (q_k^{-1} -1) \, \mu_{\mathbf{N},\mathbf{M}}^2(k,j) \, \hat \kappa_k \lesssim \alpha^2, \qquad j=1,\ldots, r.
$$

\end{proposition}
\begin{proof}
Without loss of generality, assume that $\nm{T D g}_{\ell^2} =1$.  
 Let $\br{\delta_j}_{j=1}^M$ be random Bernoulli variables such that $\bbP(\delta_j = 1) = \tilde q_j$ where $\tilde q_j = q_k$ for $j=M_{k-1}+1,\ldots, M_k$. Observe that $$Q_{\cW}^\perp = Q_{\cW}^\perp D^*D = Q_{\cW}^\perp D^* P_\Delta^\perp D,$$ so we have that
 \eas{
 & P_\Delta^\perp   D  Q_{\cW}^\perp  V^*(q_1^{-1} P_{\Omega_1 }\oplus \cdots \oplus q_r^{-1}  P_{\Omega_r})  V   g \\
 &=
 \sum_{j=1}^M (\tilde q_j^{-1} \delta_j-1) P_\Delta^\perp   D  Q_{\cW}^\perp D^* P_\Delta^\perp D V^* (e_j\otimes \overline e_j)  V   g +  P_\Delta^\perp   Q_{\cW}^\perp  V^*  P_{[M]}  V Q_\cW  D^* T^{-1} T D g.
 }
Since $\nm{ P_\Delta^\perp  D  Q_{\cW}^\perp  V^*  P_{[M]}  V Q_\cW D^*T^{-1}}_{\ell^2 \to \ell^\infty} \leq \alpha/2$, it suffices to show that
 $$
 \bbP\left(\nm{\sum_{j=1}^M (\tilde q_j^{-1} \delta_j -1) P_\Delta^\perp   D  Q_{\cW}^\perp  V^* (e_j\otimes \overline e_j)  V g}_{\ell^\infty}>  \alpha/2 \right) \leq \gamma.
 $$ 
For each $i\in\Delta^c$ and $j=1,\ldots, M$, let
$$
Z_j^i = (\tilde q_j^{-1} \delta_j -1)\ip{ P_\Delta^\perp   D  Q_{\cW}^\perp  V^* (e_j\otimes \overline e_j)  V D g}{e_i}.
$$
For each $i\in\Delta^c$,  we will first apply Bernstein's inequality (Theorem \ref{thm:bernstein}) to consider upper bounds for $\bbP\left( \abs{ \sum_{j=1}^M Z_j^i} \geq   \alpha \right)$. Observe that
\eas{
\bbE\left(\abs{Z_j^i}^2\right) &= (\tilde q_j^{-1} -1) \absu{\ip{ P_\Delta^\perp   D   Q_{\cW}^\perp  V^* (e_j\otimes \overline e_j)  V  g}{e_i}}^2\\
&= (\tilde q_j^{-1} -1) \absu{\ip{  V^* e_j}{  Q_{\cW}^\perp  P_\Delta^\perp D e_i}}^2 \absu{\ip{e_j}{ V g}}^2.
}
Let $B =  \nm{DQ_\cW^\perp D^*}_{\ell^\infty \to \ell^\infty}$,
Then,
\spl{\label{eq:card_bound2}
&\sum_{j=1}^M \bbE(\abs{Z_j^i}^2)
\leq \sup_{\nm{\alpha}_{\ell^1}\leq B} \sum_{k=1}^r (q_k^{-1}-1) \nm{P_{\Gamma_k}VD^* \alpha}_{\ell^\infty}^2 \nm{P_{\Gamma_k}V  g}_{\ell^2}^2\\
&\leq B \sup_{\nm{\alpha}_{\ell^1}\leq B} \sum_{k=1}^r (q_k^{-1}-1) \mu(P_{\Gamma_k}VD^*) \nm{P_{\Gamma_k}VD^* \alpha}_{\ell^\infty} \nm{P_{\Gamma_k}V  g}_{\ell^2}^2\\
&\leq B \sup_{\nm{\alpha}_{\ell^1}\leq B} \sum_{k=1}^r (q_k^{-1}-1) \mu(P_{\Gamma_k}VD^*) \sum_{l\in\bbN}\abs{\alpha_l} \nm{P_{\Gamma_k}VD^* e_l}_{\ell^\infty}  \nm{P_{\Gamma_k}V g}_{\ell^2}^2\\
&= B \sup_{\nm{\alpha}_{\ell^1}\leq B} \sum_{l\in\bbN}\abs{\alpha_l} \sum_{k=1}^r (q_k^{-1}-1) \mu(P_{\Gamma_k}VD^*)  \nm{P_{\Gamma_k}VD^* e_l}_{\ell^\infty}  \nm{P_{\Gamma_k}V  g}_{\ell^2}^2\\
&\leq B^2 \sup_{l\in\bbN} \sum_{k=1}^r (q_k^{-1}-1) \mu(P_{\Gamma_k}VD^*)  \nm{P_{\Gamma_k}VD^* e_l}_{\ell^\infty}  \nm{P_{\Gamma_k}V g}_{\ell^2}^2\\
&\leq B^2 \max_{j=1}^r \sum_{k=1}^r (q_k^{-1}-1) \mu_{\mathbf{N}, \mathbf{M}}^2(k,j)  \nm{P_{\Gamma_k}V g}_{\ell^2}^2\\
&\leq rB^2 \max_{j=1}^r \sum_{k=1}^r (q_k^{-1}-1) \mu_{\mathbf{N}, \mathbf{M}}^2(k,j) \hat \kappa_k =:\sigma^2,
} 
where the last line follows because $\nm{TDg}_{\ell^2}=1$ implies that $\nm{P_{\Lambda_k} D g}_{\ell^2}\leq \sqrt{r\kappa_k}$, and by definition of $\hat \kappa_k$,  $\nm{P_{\Gamma_k}V g}_{\ell^2}^2\leq r\hat \kappa_k$.

Also, we have that
\eas{
\abs{Z_j^i} &\leq \tilde q_j^{-1} \absu{\ip{  V^* e_j}{  Q_{\cW}^\perp  D^* e_i}} \absu{\ip{e_j}{ V D^*D g }}\\
&\leq  \max_{k=1}^r q_k^{-1} \, \mu( P_{\Gamma_k}  V  Q_{\cW}^\perp   D^*  P_\Delta^\perp ) \, \sum_{l=1}^r \mu(  P_{\Gamma_k}  V D^* P_{\Lambda_l}) \, \nm{P_{\Lambda_l}D g}_{\ell^1}\\
&\leq   \nm{DQ_\cW^\perp D^*}_{\ell^\infty \to \ell^\infty} \max_{k=1}^r q_k^{-1} \,   \sum_{l=1}^r \mu_{\mathbf{N}, \mathbf{M}}(k,l) \,  \nm{P_{\Lambda_l}D g}_{\ell^1}\\
&\leq   \sqrt{r} \nm{DQ_\cW^\perp D^*}_{\ell^\infty \to \ell^\infty} \max_{k=1}^r q_k^{-1} \,   \sum_{l=1}^r \mu_{\mathbf{N}, \mathbf{M}}(k,l) \,  \kappa_l =:K,
}
where the last line follows because $\nm{P_{\Lambda_l}D g}_{\ell^2}\leq \sqrt{r\kappa_l}$ along 
with (4) of Corollary \ref{cor:loc_sp} implies that $\nm{P_{\Lambda_l}Dg}_{\ell^1}\leq \sqrt{r}\kappa_l$.

Let $\Upsilon \subset \bbN$ be  such that
$$
\bbP\left(\sup_{i\in\Upsilon} \abs{\sum_{j=1}^M Z_j^i} \geq   \alpha \right) = 0
$$
and suppose that $\abs{\Upsilon^c}\leq \tilde M$.
Then, by Theorem \ref{thm:bernstein} and the union bound,
$$
\bbP\left(\sup_{i\in\Delta^c} \abs{\sum_{j=1}^M Z_j^i} \geq \frac{\alpha}{2}\right)
\leq \bbP\left(\sup_{i\in\Upsilon^c} \abs{\sum_{j=1}^M Z_j^i} \geq \frac{\alpha}{2}\right)
\leq 4\abs{\Upsilon^c} \exp\left(-\frac{\alpha^2/16}{\sigma^2 + K\alpha/(6\sqrt{2})}\right),
$$
which is true provided that
$$
\log\left(\frac{4\tilde M}{\gamma}\right) \, \sigma^2 \leq \frac{\alpha^2}{32}, \qquad 
\log\left(\frac{4\tilde M}{\gamma}\right) \, K \leq \frac{\alpha}{8}.
$$
which are simply the assumptions of this theorem.

It remains to show that such as set $\Upsilon^c$ exists: First note that $\nm{Dg} \leq \nm{T^{-1}}\nm{TDg} \leq \sqrt{\kappa_{\max}}$ with $\kappa_{\max}= r\max\br{\kappa_j}_{j=1}^r$. So, using the fact that $D$ and $V$ are isometries and hence of norm 1, notice that
\eas{
\abs{\sum_{j=1}^M Z_j^i} &= \abs{\ip{ (\tilde q_k^{-1}-1) P_\Delta^\perp   D  Q_{\cW}^\perp  V^* (q_1^{-1} P_{\Omega_1} \oplus \cdots \oplus q_r^{-1} P_{\Omega_r})  V g}{e_i}}\\
&\leq \nm{D g}_{\ell^2} \nm{D Q_\cW  V^*  (q_1^{-1} \Omega_1 \oplus \cdots q_r^{-1} \Omega_r)  V  Q_{\cW}^\perp  D^*  P_\Delta^\perp  e_i}_{\ell^2}\\
&\leq \sqrt{\kappa_{\max}} \, \max_{k=1}^r q_k^{-1} \, \nm{ P_{[M]}  V  Q_{\cW}^\perp  D^* e_i}_{\ell^2}\to 0
}
as $i\to \infty$ since $ P_{[M]}  V  Q_{\cW}^\perp  D^*$ is of finite rank. Thus, for $\alpha >0$, it suffices to let
$$
 \Upsilon^c := \br{i\in\bbN : 
\sqrt{\kappa_{\max}}   \, \max_{k=1}^r q_k^{-1} \, \nm{ P_{[M]}  V  Q_{\cW}^\perp  D^* e_i}_{\ell^2} \geq   \alpha}
$$
which is a finite set. To conclude this proof, observe that $\abs{\Upsilon^c} \leq \tilde M<\infty$. 
\end{proof}

\begin{proposition}\label{prop3}
Let $\alpha>0$ and $\gamma \in [0,1]$. Suppose that
\be{\label{eq:prop3:assump}
\nm{  Q_{\cW}  V^*  P_{[M]}^\perp  V  Q_{\cW}}_{\cH\to \cH} \leq \alpha/2.
}
Then,
$$
\bbP\left( \nm{ Q_{\cW}  V^* (q_1^{-1} P_{\Omega_1 }\oplus \cdots \oplus q_r^{-1}  P_{\Omega_r})  V  Q_{\cW}  -  Q_{\cW}  V^*  V  Q_{\cW}}_{\cH \to \cH}\geq \alpha \right) \leq \gamma,
$$
provided that
$$
q_k \gtrsim \alpha^{-2} \,  \log \left( \frac{4\tilde M}{\gamma}\right) \sum_{l=1}^r \mu_{\mathbf{N}, \mathbf{M}}^2(k,l) \, \kappa_l, \qquad k=1,\ldots, r.
$$

\end{proposition}
\begin{proof}
Let $\tilde q_j = q_k$ for $j\in\Gamma_k$. Let $\br{\delta_j}_{j=1}^r$ be Bernoulli random variables such that $\bbP(\delta_j = 1) = \tilde q_j$. 
\eas{
& \nm{ Q_{\cW}  V^* (q_1^{-1} P_{\Omega_1 }\oplus \cdots \oplus q_r^{-1}  P_{\Omega_r})  V  Q_{\cW}  -  Q_\cW  V^*  V  Q_{\cW}} 
 \leq \nm{ \sum_{j=1}^M \tilde q_j^{-1} \delta_j   Q_\cW  V^* (e_j \otimes \overline e_j)  V  Q_{\cW}} \\
& \leq  \nm{ \sum_{j=1}^M (\tilde q_j^{-1} \delta_j -1)     D Q_\cW V^* (e_j \otimes \overline e_j)  V  Q_\cW D^*   } 
+ \nm{  Q_{\cW}  V^*  P_{[M]}^\perp  V  Q_{\cW}} \\
&\leq   \nm{ \sum_{j=1}^M (\tilde q_j^{-1} \delta_j -1)     D Q_\cW V^* (e_j \otimes \overline e_j)  V  Q_\cW D^*   } +\alpha/2.
}
Let $$U = \sum_{j=1}^M (\tilde q_j^{-1} \delta_j -1) (e_j \otimes \overline e_j).
$$ Then, for each $J\in\bbN$, $\nm{D Q_\cW V^*  U V Q_\cW D}$ is bounded above by
\eas{
 &  \nm{P_{[J]} D Q_\cW V^*  U V Q_\cW D^* P_{[J]}}
 +\nm{D Q_\cW V^*  U V Q_\cW D^* P_{[J]}^\perp}
 + \nm{P_{[J]}^\perp D Q_\cW V^*  U V Q_\cW D^*}\\
 &\leq \nm{P_{[J]} D Q_\cW V^*  U V Q_\cW D^* P_{[J]}} + 2 \nm{P_{[J]}^\perp D Q_\cW}\nm{U}\\
& \leq \nm{P_{[J]} D Q_\cW V^*  U V Q_\cW D^* P_{[J]}} + 2 q^{-1}\nm{P_{[J]}^\perp D Q_\cW}
}
where $q = \min\br{q_j}_{j=1}^r$. Note that since $DQ_\cW$ has finite rank, $\nm{P_{[J]}^\perp D Q_\cW}\to 0$ as $J\to \infty$. Let
$$
\tilde M = \min\br{j\in\bbN : 8\nm{P_{[j]}^\perp D Q_\cW} \leq q\alpha }.
$$
Then, it suffices to show that
$$
\bbP\left(\nm{P_{[\tilde M]} D Q_\cW V^*  U V Q_\cW D^* P_{[\tilde M]}} > \alpha/4 \right)\leq \gamma
$$
For each $j=1,\ldots, M$, define
$Z_j = (\tilde q_j^{-1} \delta_j -1) P_{[\tilde M]} D Q_\cW V^*  (e_j \otimes \overline e_j)  V Q_\cW D^* P_{[\tilde M]}$. 
We will aim to apply Theorem \ref{thm:matrixBernstein} to derive the following.
$$
\bbP\left(\nm{ \sum_{j=1}^M Z_j} \geq   \frac{\alpha}{4} \right) \leq \gamma.
$$ Notice that $\br{Z_j}_{j=1}^M$ are independent mean-zero matrices.  Let $\xi_j = P_{[\tilde M]}  D Q_\cW V^* e_j$. Then,
$$
\nm{Z_j} \leq \tilde q_j^{-1} \norm{\xi_j \otimes \overline\xi_j} \leq \tilde q_j^{-1} \nm{\xi_j}^2.
$$
To bound this, note that for each $j=1,\ldots, M$,
\eas{
\nm{\xi_j}^2 &= \sup_{\alpha\in\bbC^{\tilde M}} \abs{\ip{DQ_\cW V^* e_j}{\alpha}}^2
= \sup_{\alpha\in\bbC^{\tilde M}, \nm{\alpha}_{\ell^2}=1} \abs{\ip{DQ_\cW D^*D V^* e_j}{\alpha}}^2\\
&\leq \sup_{\alpha\in\bbC^{\tilde M}, \nm{\alpha}_{\ell^2}=1} \sum_{l=1}^r\abs{\ip{DQ_\cW D^* P_{\Lambda_l}D V^* e_j}{\alpha}}^2\\
&\leq \sup_{\alpha\in\bbC^{\tilde M}, \nm{\alpha}_{\ell^2}=1} \sum_{l=1}^r \nm{P_{\Lambda_l}D V^* e_j}_{\ell^\infty}^2
\nm{ P_{\Lambda_l} D Q_\cW D^* \alpha}_{\ell^1}^2\\
&\leq \sup_{g\in\cW, \nm{g}=1} \sum_{l=1}^r \nm{P_{\Lambda_l}D V^* e_j}_{\ell^\infty}^2
\nm{ P_{\Lambda_l} D g}_{\ell^1}^2= \sup_{g\in\cW, \nm{g}=1} \sum_{l=1}^r \mu_{\mathbf{N}, \mathbf{M}}^2(j,l)
\nm{ P_{\Lambda_l} D g}_{\ell^1}^2.
}
 So, since $\nm{D g}=1$ implies that $\nm{ P_{\Lambda_l} D g}_{\ell^1}^2  \leq \kappa_l(\mathbf{N},\mathbf{s})$ for $l=1,\ldots,r$ by (4) of Corollary \ref{cor:loc_sp}, we have that
$$
\nm{Z_j} \leq \max_{j=1}^r q_j^{-1}  \sup_{g\in\cW} \sum_{l=1}^r \mu_{\mathbf{N}, \mathbf{M}}^2(j,l)
\, \kappa_l(\mathbf{N},\mathbf{s}) =:K.
$$
Also,
\eas{
&\nm{ \sum_{j=1}^M \bbE(Z_j^* Z_j)}_{\ell^2 \to \ell^2} = \sup_{\nm{x}_{\ell^2} =1} \abs{\sum_{j=1}^M (\tilde q_j^{-1}-1) \ip{(\xi_j \otimes \overline\xi_j) x}{(\xi_j \otimes \overline\xi_j) x}}\\
& = \sup_{\nm{x}  =1} \abs{\sum_{j=1}^M (\tilde q_j^{-1}-1) \ip{\xi_j}{\xi_j} \ip{\xi_j}{x} \ip{\xi_j}{x}}\\
&\leq  \sup_{\nm{x}  =1} \left(\max_{k=1}^M \br{(\tilde q_k^{-1} -1) \, \nm{\xi_k }^2  }\right) \,   \sum_{j=1}^M \abs{ \ip{e_j}{ V  D^*  P_{[\tilde M]}x}}^2\\
&\leq   \sup_{\nm{x} =1} \left(\max_{k=1}^M \br{(\tilde q_k^{-1} -1) \, \nm{\xi_k }^2  }\right) \,  \nm{VD^* } \nm{x} \\
&=\max_{k=1}^M \br{(\tilde q_k^{-1} -1) \, \nm{\xi_k }^2  } \leq K.
}
Thus, by Theorem \ref{thm:matrixBernstein},
\eas{
&\bbP\left( \nm{ Q_{\cW}  V^* (q_1^{-1} P_{\Omega_1 }\oplus \cdots \oplus q_r^{-1}  P_{\Omega_r})  V  Q_{\cW}  -  Q_\cW  V^*  V  Q_{\cW}}_{\cH \to \cH} \geq \alpha \right)
\leq  4\tilde M \exp\left(- \frac{\alpha^2/8}{K+ K  \alpha/6}\right) \leq \gamma
}
provided that
\be{\label{eq:prop3_final}
\log\left(\frac{4\tilde M}{\gamma}\right)K  \leq \frac{\alpha^2}{16}, \qquad
\log\left(\frac{4\tilde M}{\gamma}\right)K \leq \alpha/2,}
which are implied by the given assumptions.

\end{proof}

\begin{proposition}\label{prop4}
Let $\alpha >0$ and let $\gamma \in [0,1]$. 
Let
$$
\tilde M = \min\br{i\in\bbN : \sup_{j\geq i} \nm{ P_{[M]}  V    D^* e_j}_{\ell^2} + \nm{Q_{\cR(D^*P_{[N]})}  D^* e_j}_{\ell^2} < \sqrt{\frac{5 q}{4}}  } .
$$
Then $\tilde M$ is finite and
$$
\bbP\left( \sup_{j\in\bbN} \nm{ P_{\br{j}}  D  Q_{\cW}^\perp  V^* (q_1^{-1}  P_{\Omega_1} \oplus \cdots \oplus q_r^{-1}  P_{ \Omega_r})  V  Q_{\cW}^\perp  D^*  P_{\br{j}}} \geq \frac{5}{4} \right) \leq \gamma
$$
provided that for each $k=1,\ldots, r$ and each $j\in\bbN$,
$$
1 \gtrsim   B^2 \max_{j=1}^r (q_k^{-1}-1) \, \mu_{\mathbf{N}, \mathbf{M}}(k,j) \, \log\left(\frac{2\tilde M}{\gamma}\right),
$$
where $B= \nm{DQ_\cW^\perp D^*}_{\ell^\infty \to \ell^\infty}$.

\end{proposition}
\begin{proof} 
Let $\br{\delta_j}_{j=1}^M$ be Bernoulli random variables such that $\bbP(\delta_j = 1) = \tilde q_j$ where $\tilde q_j = q_k$ for $j=\Gamma_k$. Observe that for each $j\in\bbN$,
\eas{
 &\nm{ P_{\br{j}}  D  Q_\cW^\perp V^* (q_1^{-1} \Omega_1 \oplus \cdots \oplus q_r^{-1} \Omega_r)  V   Q_\cW^\perp D^*  P_{\br{j}}}\\
& = \abs{\sum_{k=1}^M (\tilde q_k^{-1} \delta_k -1)  P_{\br{j}}  D  Q_\cW^\perp V^* (e_k \otimes \overline e_k) V  Q_\cW^\perp D^*  P_{\br{j}}
 +  P_{\br{j}}  D Q_\cW^\perp V^*  P_{[M]} V Q_\cW^\perp D^*  P_{\br{j}}
 }\\
 &\leq \abs{\sum_{k=1}^M (\tilde q_k^{-1} \delta_k -1) \abs{\ip{ V Q_\cW^\perp D^* e_j}{e_k}}^2} + 1,
}
where we have applied $\nm{V} = \nm{D} = 1$ in the last line.
For each $j\in\bbN$ and $k=1,\ldots, M$, define
$Z_k^j =  (\tilde q_k^{-1} \delta_k -1) \abs{\ip{ V Q_\cW^\perp D^* e_j}{e_k}}^2$. To prove this proposition, we need to derive conditions under which
\be{\label{eq:prop4:ets}
\bbP\left(\sup_{j\in\bbN}\abs{\sum_{k=1}^M Z_k^j} \geq \frac{1}{4}\right).
}
We first seek to apply  Theorem \ref{thm:bernstein} to analyse  $
\bbP\left(\abs{\sum_{k=1}^M Z_k^j} > \frac{1}{4}\right)
$ for each $j\in\bbN$. Observe that
$$
\abs{Z_k^j} \leq \sup_{j\in\bbN} \max_{l=1}^r (q_l^{-1}-1) \, \mu^2( P_{\Gamma_l}  V Q_\cW^\perp  D^*  P_{\br{j}} ) =:K,
$$
and
\eas{
&\sum_{k=1}^M \bbE(\absu{Z_k^j}^2)
 = \sum_{k=1}^M ( \tilde q_k^{-1} -1) \abs{\ip{ V Q_\cW^\perp D^* e_j}{e_k}}^4\\
& \leq \max_{l=1}^r (q_l^{-1}-1) \, \mu^2( P_{\Gamma_l}  V Q_\cW^\perp D^*  P_{\br{j}})\, \nm{ V Q_\cW^\perp  D^* e_j}_{\ell^2}^2\\
& \leq \sup_{j\in\bbN}\max_{l=1}^r (q_l^{-1}-1) \, \mu( P_{\Gamma_l}  V Q_\cW^\perp D^* P_{\br{j}} )^2 =: \sigma^2.
 } Thus, by applying Theorem \ref{thm:bernstein},
 $$
 \bbP\left(\abs{\sum_{k=1}^r Z_k^j}^2 \geq \frac{1}{4}\right) \leq 2\exp\left(    -\frac{1/32}{\sigma^2 + K  /12} \right).
 $$
 In order to use this to bound (\ref{eq:prop4:ets}), we will proceed as in Proposition \ref{prop2} to show that there exists $\Upsilon$ be such that $\Upsilon^c$ is a finite set and
 $$
 \bbP\left(\sup_{j\in\Upsilon} \nm{ P_{\br{j}}  D Q_\cW^\perp   V^* (q_1^{-1} \Omega_1 \oplus \cdots \oplus q_r^{-1} \Omega_r)  V  Q_\cW^\perp  D^*  P_{\br{j}}}_{\ell^2 \to \ell^2} \geq \frac{5 }{4}\right) = 0.
 $$
 Let $q = \min\br{q_k: k=1,\ldots, r}$.
 Since $ P_{[M]}  V D^*$ and $Q_{\cR(D^*P_{[N]})}  D^* $ are both of finite rank,
 \eas{
 &\nm{ P_{\br{j}}  D  Q_\cW^\perp   V^* (q_1^{-1} \Omega_1 \oplus \cdots \oplus q_r^{-1} \Omega_r)  V   Q_\cW^\perp  D^*  P_{\br{j}}}_{\ell^2 \to \ell^2} \\
& \leq \frac{1}{q} \nm{ P_{[M]}  V   Q_\cW^\perp  D^* e_j}_{\ell^2}^2  \leq \frac{1}{q}\left(\nm{ P_{[M]}  V    D^* e_j}_{\ell^2} + \nm{Q_{\cR(D^*P_{[N]})}  D^* e_j}_{\ell^2}\right)^2 \to 0
 }
 as $j\to \infty$. Therefore, it suffices to let
 $$
 \Upsilon^c = \br{j\in\bbN : \nm{ P_{[M]}  V    D^* e_j}_{\ell^2} + \nm{Q_{\cR(D^*P_{[N]})}  D^* e_j}_{\ell^2}  > \sqrt{\frac{5 q}{4}}}
 $$
which is a finite set. Observe also that
 $\abs{\Upsilon^c} \leq \tilde M  < \infty$.  Therefore, by applying the union bound,
  $$
 \bbP\left(\sup_{j\in\bbN} \nm{ P_{\br{j}}  D Q_\cW^\perp   V^* (q_1^{-1} \Omega_1 \oplus \cdots \oplus q_r^{-1} \Omega_r)  V  Q_\cW^\perp  D^*  P_{\br{j}}}  \geq \frac{5 }{4}\right)  \leq
 2 \tilde M \exp\left(    -\frac{1/32}{\sigma^2 + K /12} \right).
 $$
 Thus, it suffices to let, for each $k=1,\ldots, r$ and each $j\in\bbN$,
 $$
 1 \gtrsim   (q_k^{-1}-1) \, \mu( P_{\Gamma_k}  V  Q_{\cW}^\perp  D^*  e_j )^2 \, \log\left(\frac{2\tilde M}{\gamma}\right).
$$
Let $B = \nm{DQ_{\cW}^\perp  D^*  e_j }_{\ell^1}$. Finally, the following observation concludes the proof of this proposition.
\eas{
&(q_k^{-1}-1) \, \mu( P_{\Gamma_k}  V  Q_{\cW}^\perp  D^*  e_j )^2 \, \log\left(\frac{2\tilde M}{\gamma}\right)\\
&\leq B \sup_{\nm{\alpha}_{\ell^1}\leq B} (q_k^{-1}-1) \, \mu( P_{\Gamma_k}  V D^*) \nm{ P_{\Gamma_k}  V D^* \alpha }_{\ell^\infty} \, \log\left(\frac{2\tilde M}{\gamma}\right)\\
&\leq B \sup_{\nm{\alpha}_{\ell^1}\leq B} \sum_{j\in\bbN} \abs{\alpha_j} (q_k^{-1}-1) \, \mu( P_{\Gamma_k}  V D^*) \nm{ P_{\Gamma_k}  V D^* e_j }_{\ell^\infty} \, \log\left(\frac{2\tilde M}{\gamma}\right)
\\
&\leq B^2 \max_{j=1}^r (q_k^{-1}-1) \, \mu^2_{\mathbf{N}, \mathbf{M}}(k,j) \, \log\left(\frac{2\tilde M}{\gamma}\right).
}

\end{proof}

\section{Constructing the dual certificate}\label{sec:dual_constr}

This section will show that, with high probability, one can construct $\rho \in\cR( V^*  P_\Omega)$ which satisfies conditions (iii) to (v) of Proposition \ref{prop:dual_certificate} if $\Omega = \Omega_{\mathrm{M}, \mathbf{m}}^{\mathrm{Ber}}$ is a Bernoulli multilevel sampling scheme satisfying Assumption \ref{dual_assumptions}.

As explained in \cite{BAACHGSCS}, the sampling model of $\Omega = \Omega_1 \cup \cdots \cup \Omega_r$ with $\Omega_i \sim \mathrm{Ber}(q_k, \Gamma_k)$ is equivalent to the following sampling model. $\Omega = \Omega_1\cup \cdots \Omega_r$ with
$$
\Omega_k = \Omega_k^1 \cup \cdots \Omega_k^\mu, \qquad \Omega_k^j\sim \mathrm{Ber}(q_k^j, \Gamma_k), \quad k=1,\ldots, r, \quad j=1,\ldots, \mu,
$$
for $\mu\in\bbN$ and $\br{q_k^j}_{j=1}^\mu$ such that
\be{\label{eq:q_k^j}
(1-q_k^1)(1-q_k^r)\cdots (1-q_k^\mu) = 1-q_k.
}
We will consider this alternative sampling model throughout this section so that we can apply the golfing scheme of \cite{gross2011recovering} to construct the dual certificate described in Proposition \ref{prop:dual:1}.
This section consists of the following steps:
\begin{enumerate}
\item Define the dual certificate.
\item Show that the constructed dual certificate satisfies conditions (iii) to (v) of Proposition \ref{prop:dual_certificate} provided that certain events occur.
\item Show that the events described in step 2 occur with high probability.
\end{enumerate}

\paragraph{Definition of the dual certificate}
Let $\gamma = \epsilon/6$. 
Let $\cL = \log(4 q^{-1} \sqrt{\kappa} \tilde M \nm{ D D^*}_{\ell^\infty \to \ell^\infty})$. Define $\mu, \nu\in\bbN$, $\br{\alpha_j}_{j=1}^\mu$ and $\br{\beta_j}_{j=1}^\mu$ as follows.
\eas{
&\mu = 8 \lceil 3\nu + \log(\gamma^{-1/2}) \rceil, \quad \nu = \log\left(8  q^{-1}\sqrt{\kappa_{\max}}  \tilde M\nm{ D D^*}_{\ell^\infty \to \ell^\infty}\right),\\
&q_k^1=q_k^1=\frac{1}{4} q_k, \quad \tilde q_k = q_k^3=\cdots = q_k^\mu,\qquad k=1,\ldots,r,\\
&\alpha_1 = \alpha_2 = \frac{1}{2 \cL^{1/2}}, \quad \alpha_i = \frac{1}{2}, \quad i=3,\ldots, \mu,\\
&\beta_1 =\beta_2 = \frac{1}{4}, \quad \beta_i = \frac{\cL}{4}, \quad i=3,\ldots, \mu.
}
For $j=1,\ldots, \mu$, define $ U_j: \cB(\ell^2(\bbN), \ell^2(\bbN))$ by $$ U_j =
\frac{1}{q^j_1} P_{\Omega^j_1}\oplus \cdots \oplus \frac{1}{q^j_r}  P_{\Omega^j_r}.$$

Let $Z_0 =  D^* \sgn( P_{\Delta}  D f)$ and for $i=1,2$, define
$$
Z_i = Z_0 -  Q_{\cW} Y_i, \quad Y_i = \sum_{j=1}^i  V^*  U_j  V  Z_{j-1}.
$$
Let $\Theta_1 = \br{1}$, $\Theta_2 = \br{1,2}$ and for $i\geq 3$, define 
\eas{
\Theta_i &= \begin{cases}
\Theta_{i-1}\cup \br{i} & \nm{ T D (Z_{i-1} -    V^*  U_i  V     Z_{i-1})}_{\ell^2} \leq \alpha_i \nm{T  D Z_{i-1}}_{\ell^2}\\
& \nm{  P_\Delta^\perp   D  Q_\cW^\perp  V^*  U_i  V     Z_{i-1} }_{\ell^\infty} \leq \beta_i \nm{ T D Z_{i-1}}_{\ell^2}\\
\\
\Theta_{i-1} & \text{otherwise.}
\end{cases}\\
Y_i &= \begin{cases}
\sum_{j\in\Theta_i}  V^* U_j  V     Z_{j-1} & i \in \Theta_i\\
Y_{i-1} &\text{otherwise}.
\end{cases}\\
Z_i &= \begin{cases}
Z_0 -  Q_{\cW}Y_i & i \in \Theta_i\\
Z_{i-1} &\text{otherwise}.
\end{cases}
} 
Note that $Z_i\in\cW$ for each $i=1,\ldots, \mu$.
Define the following events.
\eas{
A_i :\quad &\nm{  T D(  Z_{i-1} -  V^* U_i  V     Z_{i-1})}_{\ell^2} \leq \alpha_i \nm{ T D Z_{i-1}}_{\ell^2}, \quad i=1,2,\\
B_i :\quad &\nm{  P_\Delta^\perp   D  Q_\cW^\perp  V^* U_i  V     Z_{i-1} }_{\ell^\infty} \leq \beta_i \nm{  T D Z_{i-1}}_{\ell^2}, \quad i=1,2,\\
B_3: \quad& \abs{\Theta_\mu} \geq \nu,\\
B_4 :\quad &\cap_{i=1}^2 A_i \cap \cap_{i=1}^3 B_i.
}
Let $\tau(j)$ denote the $j^{th}$ element of $\Theta_\mu$ (in order of appearance). 

\paragraph{Properties of the dual certificate}
Suppose that $B_4$ occurs, and let $\rho = Y_{\tau(\nu)}$. By definition, $\rho =  V^*  P_\Omega w$ for some $w\in\ell^2(\bbN)$. We now show that $\rho$ satisfies (iii) and (iv) of Proposition \ref{prop:dual_certificate} and derive an upper bound on $\nm{w}_{\ell^2}$. 
By definition,
\be{\label{eq:Z_j}
Z_{\tau(i)} = Z_0 -  Q_\cW \sum_{j\in\Theta(i)}  V^*  U_{j}  V  Z_{j-1}
= ( Q_{\cW}-  Q_{\cW}  V^*  U_{\tau(i)}  V   ) Z_{\tau(i-1)}.
}
\begin{enumerate}
\item Since $D^*D = I$, by construction of $\rho$,   
\eas{
&\nm{Z_0 -  Q_\cW \rho} = \nm{Z_0 - Q_\cW Y_{\tau(\nu)}} =\nm{Z_{\tau(\nu)}} \\
&= \nm{D^*D( Q_\cW-  Q_{\cW}  V^* U_{\tau(\nu)}  V   Q_{\cW}) Z_{\tau(\nu-1)}}\\
&= \nm{T^{-1}T D ( Q_\cW-  Q_{\cW}  V^* U_{\tau(\nu)}  V   Q_{\cW}) Z_{\tau(\nu-1)} }.
}
Recalling the definition of $(\alpha_{i})_{i=1}^\mu$ and recalling that $\kappa_{\max} = r  \max\br{\kappa_j}$ and $\nm{T^{-1}}\leq \sqrt{\kappa_{\max}}$ , it follows that
\eas{
&\nm{Z_0 -  Q_\cW \rho}  \leq \sqrt{r \kappa }    \,  \nm{T   D \left(Z_{\tau(\nu-1)} -   V^*  U_{\tau(\nu)}  V   Q_{\cW} Z_{\tau(\nu-1)} \right)}_{\ell^2}\\
&\leq \sqrt{\kappa_{\max}}\,  \alpha_{\tau(\nu)} \,  \nm{T D Z_{\tau(\nu-1)}}_{\ell^2}\\
&\leq \sqrt{ \kappa_{\max}}\,    \nm{ T  D Z_0}_{\ell^2} \, \prod_{j=1}^\nu \alpha_{\tau(j)}\\
& \leq \sqrt{\kappa_{\max}} \, \nm{    D  D^*   }_{\ell^\infty \to \ell^\infty} \, \prod_{j=1}^\nu \alpha_{\tau(j)}\\
&\leq \sqrt{ \kappa_{\max}}\,  \nm{  D  D^* }_{\ell^\infty \to \ell^\infty}\, 2^{-\nu} \leq \frac{\sqrt{q}}{8},
}
by our choice of $\nu$. Note that to get from the third line to the forth line, we observe that $\nm{DZ_0}_{\ell^\infty}\leq \nm{DD^*}_{\ell^\infty\to \ell^\infty}$, then recall from (3) of Corollary \ref{cor:loc_sp} that $\nm{P_{\Lambda_j} D g}_{\ell^2}^2 \leq \kappa_j$ for all $g=D^*P_\Delta x$ with $\nm{Dg}_{\ell^\infty}\leq 1$. Therefore,
$$
\nm{TDZ_0}_{\ell^2}^2 = \sum_{j=1}^r \frac{\nm{P_{\Lambda_j}D Z_0}}{r\kappa_j}\leq \nm{DD^*}_{\ell^\infty \to \ell^\infty}.
$$
\item 
Recalling our definition of $(\beta_j)_{j=1}^\mu$ and the estimate used in the previous step to bound $\nm{  T D Z_{\tau(j)}}_{\ell^2}$,
\eas{
&\nm{ P_\Delta^\perp   D  Q_\cW^\perp \rho}_{\ell^\infty} \leq \sum_{j=1}^\nu \nm{ P_\Delta^\perp   D  Q_\cW^\perp  V^*  U_{\tau(j)}  V    Q_\cW Z_{\tau(j-1)}}_{\ell^\infty}\\
&\leq \sum_{j=1}^\nu \beta_{\tau(j)} \nm{  T D Z_{\tau(j-1)}}_{\ell^2}\\
&\leq \nm{   D  D^* }_{\ell^\infty \to \ell^\infty} \sum_{j=1}^\nu \beta_{\tau(j)} \prod_{i=1}^{j-1} \alpha_{\tau(i)}\\
&\leq  \frac{1}{4} \left(1+ \frac{1}{2\sqrt{\cL}}+ \frac{\cL}{2^2 \cL}+ \cdots  + \frac{\cL}{\cL 2^{\nu-1}}\right)
\leq \frac{1}{2}.
}

\item By definition, $\rho =   V^*  P_\Omega w$ with $w= \sum_{j=1}^\nu w_j$ and $w_j = U_{\tau(j)}  V    Z_{\tau(j-1)}$.
For each $j=1,\ldots,\nu$,
\eas{
\nm{w_j}^2 &= \ip{U_{\tau(j)}  V  Z_{\tau(j-1)}}{ U_{\tau(j)}  V Z_{\tau(j-1)}}
 = \sum_{k=1}^r \frac{1}{(q_k^{\tau(j)})^2} \nm{ P_{\Omega_k^{\tau(j)}}  V  Z_{\tau(j-1)}}^2\\
&\leq K_{\tau(j)} \sum_{k=1}^r \frac{1}{q_{k}^{\tau(j)}} \ip{ V^*  P_{\Omega_k^{\tau(j)}}  V  Z_{\tau(j-1)}}{ \rZ_{\tau(j-1)}}
=  K_{\tau(j)} \ip{ V^* U_{\tau(j)}  V  Z_{\tau(j-1)}}{ Z_{\tau(j-1)}},
}
where $K_{\tau(j)}  = \max \br{1/q_k^{\tau(j)}: k=1,\ldots, r}$.
Using (\ref{eq:Z_j}), we have that
\eas{
&\ip{ V^* U_{\tau(j)}  V  Z_{\tau(j-1)}}{ Z_{\tau(j-1)}} \\
&=\ip{    V^* U_{\tau(j)}  V  Z_{\tau(j-1)} - Z_{\tau(j-1)}}{ Z_{\tau(j-1)}} + \ip{Z_{\tau(j-1)}}{ Z_{\tau(j-1)}}\\
&= \ip{Z_{\tau(j)}}{ Z_{\tau(j-1)}} + \ip{Z_{\tau(j-1)}}{ Z_{\tau(j-1)}}.
}
Since $I = D^*D$ and $\nm{T^{-1}} \leq \kappa_{\max}$, we have that 
\eas{
&\abs{\ip{ V^*  U_{\tau(j)}  V   Z_{\tau(j-1)}}{Z_{\tau(j-1)}} } =\abs{\ip{D Z_{\tau(j)}}{D Z_{\tau(j-1)}} + \ip{DZ_{\tau(j-1)}}{ DZ_{\tau(j-1)}} } \\
&\leq  \nm{ D Z_{\tau(j-1)}} \left( \nm{ D Z_{\tau(j)}} + \nm{  D Z_{\tau(j-1)}} \right)\\
&\leq \kappa_{\max} \, \nm{T D Z_{\tau(j-1)}}_{\ell^2} \left( \nm{ T   D Z_{\tau(j)}}_{\ell^2} + \nm{ T   D Z_{\tau(j-1)}}_{\ell^2}\right).
}
Using $\nm{ T  D Z_{\tau(j)}}_{\ell^2} \leq \alpha_{\tau(j)}\nm{ T   D Z_{\tau(j-1)}}_{\ell^2}$, we obtain,
\bes{
\abs{\ip{ V^*  U_{\tau(j)}  V  Z_{\tau(j-1)}}{Z_{\tau(j-1)}} } 
\leq \kappa_{\max} \,(\alpha_{\tau(j)}+ 1) \, \left(\prod_{i=1}^{j-1} \alpha_{\tau(i)}\right)^2.
}
Therefore,
$$
\nm{w_j}_{\ell^2} \leq \sqrt{K_{\tau(j)}\, \kappa_{\max} \,(\alpha_{\tau(j)}+ 1)} \, \left(\prod_{i=1}^{j-1} \alpha_{\tau(i)}\right).
$$
Recall that for $k=1,\ldots, r$, $q_k^j = q_k/4$ for $j=1,2$ and $q_k^j = \tilde q_k$ for all $j\geq 3$. Let $K = \max_{k=1}^r\br{q_k^{-1}}$ and  for $j\geq 3$, first note that $(1-q_k^1) \cdots (1-q_k^\mu) = 1-q_k$ implies that $q_k^1+\cdots + q_k^\mu \geq q_k$. So, we have that $2(\mu - 2) \tilde q_k \geq q_k$ since $q_k^1=q_k^2 = q_k/4$. By our choice of $\mu$, this implies that 
$$
2\left(8\left\lceil 3 \log\left(8 q^{-1}\sqrt{\kappa_{\max}}  \tilde M\nm{ D D^*}_{\ell^\infty \to \ell^\infty}\right) + \log(\gamma^{-1})\right\rceil -2 \right) \tilde q_k \geq q_k,
$$
and for $j\geq 3$,
$$
K_{\tau(j)} \leq 2\left(8\left\lceil 3 \log\left(8 q^{-1}\sqrt{\kappa_{\max}}  \tilde M\nm{ D D^*}_{\ell^\infty \to \ell^\infty}\right) + \log(\gamma^{-1})\right\rceil -2 \right) K.
$$
So, $\nm{w_j}_{\ell^2}$ are bounded as follows.
\bes{
\nm{w_1}_{\ell^2}  \leq 2 \sqrt{K \kappa_{\max}} \sqrt{1+ \frac{1}{2\cL^{1/2}}},\qquad
\nm{w_2}_{\ell^2}  \leq 2 \sqrt{K \kappa_{\max}}  \sqrt{1+ \frac{1}{2\cL^{1/2}}} \sqrt{\frac{1}{2 \cL^{1/2}}},}
and for $j\geq 3$,
\bes{
\nm{w_j}_{\ell^2} \leq 2\sqrt{3} \sqrt{K \kappa_{\max}}   \sqrt{\frac{1}{2^{j-1} \cL^{1/2}}} \sqrt{\left(8\left\lceil 3 \log\left(8 q^{-1}\sqrt{\kappa_{\max}}  \tilde M\nm{ D D^*}_{\ell^\infty \to \ell^\infty}\right) + \log(\gamma^{-1})\right\rceil -2 \right)} .
} Summing these terms yields
$$
\nm{w}_{\ell^2} \lesssim\sqrt{ K \, \kappa_{\max} \,   \left(\frac{\log\left(8 q^{-1}\sqrt{\kappa_{\max}}  \tilde M\nm{ D D^*}_{\ell^\infty \to \ell^\infty}\right) + \log(6/\epsilon)}{\log(4 q^{-1} \sqrt{\kappa_{\max}} \tilde M\nm{ D D^*}_{\ell^\infty \to \ell^\infty})}\right)}.
$$
\end{enumerate}

To show that conditions (iii) to (v) of Proposition \ref{prop:dual_certificate} are satisfied with probability exceeding $1-5\gamma = 1-5\epsilon/6$ under Assumption \ref{dual_assumptions}, we will show that
$\bbP(A_i^c)<\gamma$ for $i=1,2$ and $
\bbP(B_j^c)<\gamma
$ for $j=1,2,3$.

\paragraph{Proof of $\bbP(B_3^c) < \gamma$}
Define the random variables $X_1,\cdots, X_{\mu-2}$ by
$$
X_j = \begin{cases}
0 &\Theta_{j+2} \neq \Theta_{j+1},\\
1 &\text{otherwise}.
\end{cases}
$$
Observe that 
$$\bbP(B_3^c) = \bbP(\abs{\Theta_\mu} < \nu) = \bbP\left(X_1+\cdots + X_{\mu-2} > \mu - \nu\right).$$
Although $\br{X_j}_{j=1}^{\mu-2}$ are not independent random variables, from \cite[Eqn. (7.80) - (7.85)]{adcockbreaking} the above probability can be controlled by independent binary random variables and the standard Chernoff bound, so that $\bbP(B_3^c) \leq \gamma$ 
provided that
\be{\label{eq:toshow_b3}
\frac{1}{4} \geq \bbP(X_j=1| X_{l_1}=\cdots = X_{l_k} = 1)
}
for all $j=1,\ldots, \mu-2$ and $l_1,\ldots, l_k \in \br{1,\ldots, \mu-2}$ such that $j\not\in \br{l_1,\ldots, l_k}$ and t $\mu \geq 8 \lceil 3\nu + \log(\gamma^{-1/2}) \rceil$. It remains to verify that (\ref{eq:toshow_b3}) holds with $p=1/4$. Observe that $X_j = 0$ whenever
$$
\nm{ T D ( Q_\cW -  Q_\cW  V^*  U_j  V  Q_\cW) Z_{i-1}}_{\ell^2} \leq \frac{1}{2}\nm{T D Z_{i-1}}_{\ell^2}
$$
and
$$
\nm{  P_\Delta^\perp  D  Q_\cW^\perp  V^*  U_j  V  Q_\cW Z_{i-1}}_{\ell^\infty} \leq \frac{\cL}{4}\nm{  T D Z_{i-1}}_{\ell^2}
$$
for $i=j+2$. Thus,
(\ref{eq:toshow_b3}) holds with $p=1/4$ if
\be{\label{eq:X_j_prob1}
\bbP\left(
\nm{ T D ( Q_\cW -  Q_\cW  V^*  U_j  V  Q_\cW) Z_{i-1}}_{\ell^2} > \frac{1}{2}\nm{ T D Z_{i-1}}_{\ell^2} \right) \leq \frac{1}{8},
}
and
\be{\label{eq:X_j_prob2}
\bbP\left( \nm{  P_\Delta^\perp  D  Q_\cW^\perp  V^*  U_j  V  Q_\cW Z_{i-1}}_{\ell^\infty} > \frac{\cL}{4}\nm{ T  D Z_{i-1}}_{\ell^2} \right)\leq \frac{1}{8}.
}
By Proposition \ref{prop1}, (\ref{eq:X_j_prob1}) is implied by (C1) below, and by Proposition \ref{prop2}, (\ref{eq:X_j_prob2}) is implied by (C2) below.
\begin{itemize}

\item[(C1)]
Let
\be{\label{eq:C1a}
\nm{T  D (Q_\cW V^*  P_{[M]}  V Q_\cW  -   Q_\cW ) D^* T^{-1}}_{\ell^2 \to \ell^2} \leq \frac{1}{4}, 
} 
\be{\label{eq:C1b}
\tilde q_k \gtrsim  \sqrt{r}B \sum_{l=1}^r \mu_{\mathbf{N}, \mathbf{M}}^2( k,l)\, \kappa_l, \qquad k=1,\ldots, r,
} 
and  
\be{\label{eq:C1c}
1\gtrsim \, rB^2  \sum_{k=1}^r (\tilde q_k^{-1}-1) \, \mu_{\mathbf{N}, \mathbf{M}}^2( k,j)
\, \hat \kappa_k, \qquad  j=1, \ldots, r.
} 
\item[(C2)]
Let 
\be{\label{eq:C2a}
\nm{ P_\Delta^\perp   D  Q_{\cW}^\perp  V^*  P_{[M]}  V Q_\cW D^*T^{-1}}_{\ell^2 \to \ell^\infty} \leq \frac{\cL}{8},
}   
\be{\label{eq:C2c}
\tilde q_k \gtrsim  \frac{\sqrt{r} B  \log\left(32 \tilde M\right)}{\cL}\,
 \sum_{l=1}^r \mu_{\mathbf{N}, \mathbf{M}}^2( k,l)\, \kappa_l, \qquad k=1,\ldots, r,
} 
and
\be{\label{eq:C2b}
 1 \gtrsim
\frac{rB^2 \log\left(32 \tilde M\right)}{\cL^2}  \,  \sum_{k=1}^r (\tilde q_k^{-1} -1) \, \mu_{\mathbf{N}, \mathbf{M}}^2( k,j) \, \hat \kappa_k, \qquad j=1,\ldots, r.
} 
\end{itemize}
It remains to show that (C1) and (C2) are implied by Assumption \ref{dual_assumptions}. First, (\ref{eq:C1a}) and (\ref{eq:C2a}) are implied by (a) and (b) of Assumption \ref{dual_assumptions} respectively because $\nm{T^{-1}}\leq \sqrt{\kappa_{\max}}$ and $\nm{T}\leq 1/\sqrt{\kappa_{\min}}$. We now show that (d) of Assumption \ref{dual_assumptions} implies conditions (\ref{eq:C1c}) and (\ref{eq:C2b}).
Since $(1-q_k^1) \cdots (1-q_k^\mu) = 1-q_k$ implies that $q_k^1+\cdots + q_k^\mu \geq q_k$, by our choice of $q_k^1=q_k^2 = q_k/4$ and $q_k^j = \tilde q_k$ for $j\geq 3$, it follows that $2(\mu - 2) \tilde q_k \geq q_k$.
 From (d) of Assumption \ref{dual_assumptions}, we have that  for some appropriate constant $C$,
$q_k \gtrsim rB (\log( \epsilon^{-1}) +1) \log(\tilde M   \sqrt{\kappa_{\max}}  ) \, \hat q_k$ such that 
$\br{\hat q_k}_{k=1}^r$ satisfies, 
$$
1\gtrsim  \sum_{k=1}^r (\hat q_k^{-1}-1) \, \mu_{\mathbf{N}, \mathbf{M}}^2(k,j)
\, \hat \kappa_k, \qquad j=1,\ldots, r.
$$ 
So,
\eas{
&2(8\lceil 3 \log(8 q^{-1} \tilde M \sqrt{ \kappa_{\max}}\nm{ D D^*}_{\ell^\infty \to \ell^\infty}) + \log(\gamma^{-1})\rceil -2) \tilde q_k \\&\geq q_k \gtrsim
(\log( \epsilon^{-1}) +1) \log(8 \tilde M q^{-1} \sqrt{\kappa_{\max}}  \nm{ D D^*}_{\ell^\infty \to \ell^\infty}) \, \hat q_k.
}
Since $\gamma = \epsilon/6$, it follows that $\tilde q_k \gtrsim \hat q_k$.
Thus, it follows that given any  $j=1,\ldots, r$,
\eas{
&1\gtrsim     \sum_{k=1}^r (\hat q_k^{-1}-1) \, \mu_{\mathbf{N}, \mathbf{M}}^2(k,j)
\, \hat \kappa_k \\&
\gtrsim    \sum_{k=1}^r (\tilde q_k^{-1} -1) \, \mu_{\mathbf{N}, \mathbf{M}}^2(k,j)
\, \hat \kappa_k
}
as required.

Finally, we show that the remaining conditions (\ref{eq:C1b}) and (\ref{eq:C2c}) are implied by (c) of Assumption \ref{dual_assumptions}. Recall that (c) imposes that for some appropriate constant $C$ and  each $k=1,\ldots, r$
$$
 q_k \geq C\, (\log( \epsilon^{-1}) +1) \log(8\tilde M  q^{-1}\sqrt{\kappa_{\max}}  \nm{ D D^*}_{\ell^\infty \to \ell^\infty})  \,  \sum_{l=1}^r \mu_{\mathbf{N}, \mathbf{M}}^2(k,l)
\,  \kappa_k.
$$
Since
$$2(8\lceil 3 \log(8 q^{-1} \tilde M \sqrt{\kappa_{\max}}\nm{ D D^*}_{\ell^\infty \to \ell^\infty}) + \log(\gamma^{-1})\rceil -2) \tilde q_k \geq q_k,$$
it follows that 
$$
 \tilde q_k \gtrsim    \sum_{l=1}^r \mu_{\mathbf{N}, \mathbf{M}}^2(k,l)
\, \kappa_l,
$$
as required.

\paragraph{Proof of $\bbP(A_i^c)\leq \gamma$ for $i=1,2$}
Recall from (a) of Assumption \ref{dual_assumptions} that$$
\nm{  T D  Q_\cW V^*  P_{[M]}^\perp V Q_\cW D^*   }_{\ell^2 \to \ell^2} \leq \frac{\sqrt{\kappa_{\min}}}{4\sqrt{\cL \kappa_{\max}}},
$$
which implies (\ref{eq:prop1_assp}) of Proposition \ref{prop1} with $\alpha = \alpha_1$.
So, by Proposition \ref{prop1},
$$
\bbP\left(A_i^c\right) \leq \gamma
$$
provided that  
$$
q_k \gtrsim \cL^{1/2} \log(3/\gamma) \, \sqrt{r}B \, \sum_{l=1}^r \mu_{\mathbf{N}, \mathbf{M}}^2(k,l)
\,  \kappa_l, \qquad k=1,\ldots, r
$$
and 
$$
1\gtrsim \cL \, \log(3/\gamma) \,  rB^2 \sum_{k=1}^r (q_k^{-1}-1) \, \mu_{\mathbf{N}, \mathbf{M}}^2(k,j)
\, \hat \kappa_k, \qquad j=1,\ldots, r.
$$
These two conditions are implied by (c) and (d) of Assumption \ref{dual_assumptions}.

\paragraph{Proof  of $\bbP(B_i^c) \leq \gamma$ for $i=1,2$}
Recall from (b) of Assumption \ref{dual_assumptions} that \bes{
\nm{    D  Q_{\cW}^\perp  V^*  P_{[M]}  V Q_\cW D^*  }_{\ell^2 \to \ell^\infty} \leq \frac{1}{8\sqrt{\kappa_{\max}}},
}
which implies (\ref{eq:prop2assp}) with $\alpha = \beta_1$.
So, by Proposition \ref{prop2},$$
\bbP\left(B_i^c \right) \leq \gamma
$$
provided that
$$
1 \gtrsim \sqrt{r}B\log\left(\frac{4 \tilde M}{\gamma}\right) \,   \sum_{k=1}^r (q_k^{-1} -1) \,  \mu_{\mathbf{N}, \mathbf{M}}^2(k,j)
\, \hat \kappa_k, \qquad j=1,\ldots, r,
$$ 
and 
$$
  q_k \gtrsim rB^2 \log\left(\frac{4 \tilde M}{\gamma}\right)
 \,   \sum_{l=1}^r \mu_{\mathbf{N}, \mathbf{M}}^2(k,l)\, \kappa_l, \qquad k=1,\ldots, r,
.
$$
These two conditions are implied by (c) and (d) of Assumption \ref{dual_assumptions}.

\section{Properties of the subsampled matrix}\label{sec:isometr_props}
In this section we show that  conditions (i) and (ii) of Proposition \ref{prop:dual_certificate} are satisfied with probability exceeding $1-\epsilon/6$ under Assumption \ref{dual_assumptions}. 

Recall that conditions (i) and (ii) of Proposition \ref{prop:dual_certificate} are
\begin{itemize}
\item[(i)] $
\nm{ Q_{\cW}  V^* (q_1^{-1} P_{\Omega_1 }\oplus \cdots \oplus q_r^{-1}  P_{\Omega_r})  V  Q_{\cW}  -  Q_\cW   }_{\cH \to \cH} < \frac{1}{4} .
$
\item[(ii)] $\sup_{j\in\bbN} \nm{ P_{\br{j}}  D  Q_{\cW}^\perp  V^* (q_1^{-1} \Omega_1 \oplus \cdots \oplus q_r^{-1} \Omega_r)  V  Q_{\cW}^\perp  D^*  P_{\br{j}}}_{\ell^2 \to \ell^2} < \frac{5}{4}.$
\end{itemize}
It is sufficient to show that
\be{\label{eq:isom_toshow1}
\bbP\left(\nm{ Q_{\cW}  V^* (q_1^{-1} P_{\Omega_1 }\oplus \cdots \oplus q_r^{-1}  P_{\Omega_r})  V  Q_{\cW}  -    Q_{\cW}}_{\cH \to \cH} \geq \frac{1}{4} \right) < \epsilon/12.
}
and
\be{\label{eq:isom_toshow2}
\bbP\left(\sup_{j\in\bbN} \nm{ P_{\br{j}}  D  Q_{\cW}^\perp  V^* (q_1^{-1} \Omega_1 \oplus \cdots \oplus q_r^{-1} \Omega_r)  V  Q_{\cW}^\perp  D^*  P_{\br{j}}}_{\ell^2 \to \ell^2} \geq \frac{5}{4} \right) \leq \epsilon/12.
}
The fact that (\ref{eq:isom_toshow1}) holds under Assumption \ref{dual_assumptions} follows from Proposition \ref{prop3}.

To see that Proposition \ref{prop4} implies that (\ref{eq:isom_toshow2}) under Assumption \ref{dual_assumptions}, first note that by our choice of $\tilde M$ and Proposition \ref{prop4}, (\ref{eq:isom_toshow2}) follows if
  for each $k,j=1,\ldots, r$,
\be{\label{eq:remainstoshow_prop4}
1 \gtrsim (q_k^{-1}-1) \,B^2\, \mu_{\mathbf{N},\mathbf{M}}(k,j)^2 \, \log\left(\frac{12\tilde M}{\epsilon}\right).
}which is implied by (d) of Assumption \ref{dual_assumptions}.

\section{Concluding remarks}
Recent works \cite{ward2013stable, adcockbreaking} have identified the need for further theoretical development on  the use of variable density sampling in compressed sensing. Furthermore, variable density sampling schemes are dependent not only on sparsity but also the sparsity structure of the underlying signal.  To address this, \cite{adcockbreaking}  showed that in the case of where the sparsifying operator is associated with an orthonormal basis, by considering levels of the sampling and sparsifying operators, the amount of subsampling possible can be described in terms of the local coherences between the different sections and the sparsity of the underlying signal within each level. This paper presented an extension of this result to the case where the sparsifying operator is constructed from a tight frame. By defining the notions of localized sparsity and localized level sparsities, we derived a recovery guarantee for multilevel sampling patterns based on local coherences and localized level sparsities. One direction of future work would be to apply our abstract result to analyse the use of multilevel sampling schemes in the case of Fourier sampling with some multi-scale analysis operator such as wavelet frames and shearlets. By deriving estimates on the local coherences of such operators, one can expect to obtain a better understanding on how to exploit sparsity structure to subsample.  Finally, although this paper considered only the case of a tight frame regularizer, this does not seem to be necessary in practice and it is likely that that similar estimates to Theorem \ref{thm:main} can be derived by considering the canonical dual operator of $D$.

\section*{Acknowledgements}
 This work was supported by the UK Engineering and Physical Sciences
Research Council (EPSRC) grant EP/H023348/1 for the University of Cambridge
Centre for Doctoral Training, the Cambridge Centre for Analysis. The author would like to thank Anders Hansen for constructive comments.

\appendix

\section{The discrete Haar wavelet frame}\label{sec:discreteHaar_def}
The discrete Haar frame of redundancy two is defined as follows. Let $N=2^p$ for some $p\in\bbN$ and $\br{c_0}\cup\br{h_{k,j}: k=0,\ldots, p-1, j=1,\ldots, 2^p}$ be the discrete Haar basis for $\bbC^N$. Specifically, $c_0 = 2^{-p/2}(1,\ldots, 1)$ and for $l=0,\ldots,p-1$ and $k=1,\ldots, 2^k$,
$$
h_{l,k}[j] = \begin{cases}
2^{(l-p)/2} &  j=k2^{p-l}+1,\ldots, k2^{p-l} + 2^{p-l-1}\\
-2^{(l-p)/2} &  j=k2^{p-l} + 2^{p-l-1}+1,\ldots, k2^{p-l} + 2^{p-l}\\
\end{cases}, \qquad 1\leq j\leq 2^p.
$$
Let $\tilde c_0 = c_0$,  and for each $k=0,\ldots, p-1$, $j=1,\ldots, 2^k$, let
$$\tilde h_{k,j}[n] = \begin{cases} h_{k,j}[n-1] &n=2,\ldots, N\\ h_{k,j}[N] &n=1
\end{cases}.
$$
The two  discrete Haar wavelet frame of redundancy two is defined by
\eas{
&\br{2^{-1/2} c_0} \cup \br{2^{-1/2}h_{k,j}: k=0,\ldots, p-1, j=1,\ldots, 2^k}\\
&\cup \br{2^{-1/2}\tilde c_0} \cup \br{2^{-1/2} \tilde h_{k,j}: k=0,\ldots, p-1, j=1,\ldots, 2^k}.
}
For analysis purposes, we will order these frame elements in increasing order of scaling with
\eas{
\br{\varphi_j}_{j=1}^{2N} = &\Bigg\{2^{-1/2} c_0,2^{-1/2} \tilde c_0, 2^{-1/2} h_{0,1}, 2^{-1/2} \tilde h_{0,1}, \ldots, 2^{-1/2} h_{k,j}, \tilde h_{k,j}, 2^{-1/2} h_{k,j+1}, 2^{-1/2}\tilde h_{k,j+1},\ldots\\
&\ldots, 2^{-1/2} h_{k+1,j}, \tilde 2^{-1/2} h_{k+1,j}, \ldots,2^{-1/2} \tilde h_{p-1,N},2^{-1/2} \tilde h_{p-1,N}\Bigg\}}
and let $D x = \left(\ip{x}{\varphi_k}\right)_{k=1}^{2N}$. Note that $D^*D = I$.

The following lemma shows that $\nm{DD^*}_{\ell^\infty\to \ell^\infty}$ can be upper bounded independently of $N$.
\begin{lemma}
$\nm{DD^*}_{\ell^\infty\to \ell^\infty}  \leq 2/(\sqrt{2}-1)$.
\end{lemma}
\prf{
Let $\varphi_t = 2^{-1/2}\tilde h_{m,n}$. Then, since $\br{\tilde h_{i,j}}_{i,j}$ is an orthonormal system,
\eas{
&\sum_{j} \abs{\ip{\varphi_j}{\varphi_t}}=
\frac{1}{2}\sum_{l,k} \abs{\ip{\tilde h_{m,n}}{h_{l,k}}} = 
\frac{1}{2}\sum_{l,k} \abs{\ip{\tilde h_{m,n} - h_{m,n}}{h_{l,k}}} + \frac{1}{2} \sum_{l,k} \abs{\ip{ h_{m,n}}{h_{l,k}}} \\
& \frac{1}{2}+ \frac{1}{2}\sum_{l,k} \abs{\ip{\tilde h_{m,n} - h_{m,n}}{h_{l,k}}}.
}
Now,
$$
\tilde h_{m,n}-h_{m,n}[j] = \begin{cases}
2^{(m-p)/2} & j\in\br{ n2^{p-m}-1,\, n2^{p-m}+2^{p-m}-1}\\
2^{(m-p)/2+1} & j=n2^{p-m}+2^{p-m-1}-1\\
0&\text{otherwise}
\end{cases}.
$$
Note that for each $l$, $\ip{\tilde h_{m,n}-h_{m,n}}{h_{l,k}}\neq 0$ for at most 3 values  of $k$ in $\br{0,\ldots, 2^{l}-1}$, and
$$
\sum_{k}\abs{\ip{\tilde h_{m,n}-h_{m,n}}{h_{l,k}}} \leq 4\cdot 2^{(m+l)/2-p}.
$$
Therefore,
$$\sum_{l,k} \abs{\ip{\tilde h_{m,n} - h_{m,n}}{h_{l,k}}}
\leq \frac{4\cdot 2^{m/2}}{2^p} \sum_{l=0}^{p-1} 2^{l/2}\leq \frac{4}{\sqrt{2}-1},
$$
and
$$
\sum_{j} \abs{\ip{\varphi_j}{\varphi_t}} \leq \frac{2}{\sqrt{2}-1}.
$$
The case where $\varphi_t = 2^{-1/2} h_{n,m}$ for some $n,m$ can be approached similarly with the same upper bound. Thus, $\nm{DD^*}_{\ell^\infty} = \sup_j \nm{DD^* e_j}_{\ell^1} \leq \frac{2}{\sqrt{2}-1}$.
}

\section{Existence of minimizers}
\begin{proposition}\label{prop:existence}
Let $\Omega\subset \bbN$ be finite and let $y\in\ell^2(\bbN)$.
There exists $g_*\in \cH$ such that
\bes{
g_* \in \argmin_{g\in\cH} \nm{D g}_{\ell^1} \text{ subject to } \nm{P_{\Omega}V g- y}_{\ell^2} \leq \delta.
}
\end{proposition}
\begin{proof}
Let $(f_n)_{n\in\bbN}\subset \cH$ be a minimizing sequence such that $\nm{ P_\Omega V f_n - y}_{\ell^2} \leq \delta$ for each $n\in\bbN$ and
$$
\nm{D f_n}_{\ell^1} \to  \inf_{g\in\cH} \nm{D g}_{\ell^1} \text{ subject to } \nm{P_{\Omega}V g- y} \leq \delta, \qquad n\to \infty.
$$
This implies that $(D f_n)_{n\in\bbN}$ is a bounded sequence in $\ell^1(\bbN)$. Since the dual of $c_0(\bbN)$ (the Banach space of sequences converging to zero) is $\ell^1(\bbN)$, and the unit ball of $\ell^1$ is weak-* compact, there exists $x\in\ell^1(\bbN)$ and a subsequence $(D f_{n_k})_{k\in\bbN}$ such that $D f_{n_k} \stackrel{*}{\rightharpoonup} x$ as $k\to \infty$, and for each $z\in c_0(\bbN)$,
$
\ip{D f_{n_k}}{z} \to \ip{x}{z}
$
as $k\to \infty$.

 Since $DD^* \in \cB(\ell^\infty(\bbN), \ell^\infty(\bbN))$, it follows that given any
 $z\in c_0(\bbN)$, $ D D^* z\in c_0(\bbN)$. To see this, note that given any $\delta>0$, we can choose $N_1\in\bbN$ such that $\nm{P_{[N]}^\perp z}_{\ell^\infty}\leq \delta/(2\nm{DD^*}_{\ell^1\to\ell^1})$ for all $N\geq N_1$. Furthermore, for this choice of $N_1$, we can choose $N_2$ such that $\nm{P_{[N_1]} D D^* e_k}_{\ell^1}\leq \delta/(2\nm{z}_{\ell^\infty})$ for all $k\geq N_2$. Thus, for all $k\geq \max\br{N_1,N_2}$,
\eas{
 \abs{\ip{D^* D z}{e_k}} &\leq  \abs{\ip{D^* P_{[N_1]} D z}{e_k}} +  \abs{\ip{D^* P_{[N_1]}^\perp D z}{e_k}} \\
 &\leq \nm{D D^*}_{\ell^1\to\ell^1}\nm{P_{[N_1]}^\perp z}_{\ell^\infty} + \nm{z}_{\ell^\infty} \nm{ P_{[N_1]}  D  D^* e_k}_{\ell^1} \leq \delta.
 }
 Therefore, $ D^* D z\in c_0(\bbN)$, 
$
\ip{ D f_{n_k}}{ D  D^* z} \to \ip{x}{ D D^* z}
$
as $k\to \infty$
and consequently, $ D D^* D f_{n_k} = D f_{n_k}  \stackrel{*}{\rightharpoonup} D D^* x$. This implies that,
$$
\liminf_{k\to \infty} \nm{ D f_{n_k}}_{\ell^1} \geq \nm{ D D^* x}_{\ell^1}.
$$
Furthermore, because $ D f_{n_k}$ converges weakly to $x$ in $\ell^2(\bbN)$ and $ P_{\Omega}V  D^*$ is a compact operator (since it is of finite rank), $P_\Omega V f_{n_k} = P_\Omega V D^* D f_{n_k} \to P_\Omega V D^* x$ as $k\to \infty$. So,
$\nm{ P_\Omega V  D^* x - y}_{\ell^2}\leq \delta$. Thus, $g_* :=  D^* x$ is a minimizer.

\end{proof}

 \addcontentsline{toc}{section}{References}
\bibliographystyle{abbrv}
\bibliography{References}
 
\end{document}